\documentclass[11pt]{article}
\usepackage{amsmath,amsfonts,amsthm,amssymb,color}
\usepackage{fancybox}
\usepackage{multirow}

\usepackage{ifthen}

\newboolean{@full}
\setboolean{@full}{true}  

\newcommand{\iffull}{\ifthenelse{\boolean{@full}}}

\newtheorem{theorem}{Theorem}[section]
\newtheorem{corollary}[theorem]{Corollary}
\newtheorem{lemma}[theorem]{Lemma}

\newtheorem{proposition}[theorem]{Proposition}

\newtheorem{fact}[theorem]{Fact}

\theoremstyle{definition}
\newtheorem{definition}[theorem]{Definition}
\newtheorem{remark}[theorem]{Remark}

\newenvironment{fminipage}%
  {\begin{Sbox}\begin{minipage}}%
  {\end{minipage}\end{Sbox}\fbox{\TheSbox}}

\iffull{
	\textheight 8.5in
	\topmargin -0.2in
	\oddsidemargin 0.20in
	\textwidth 6.3in
	
	\newenvironment{algbox}[0]{\vskip 0.2in
	\noindent 
	\begin{fminipage}{6.3in}
	}{
	\end{fminipage}
	\vskip 0.2in
	}
}{
	\usepackage[margin=1in]{geometry}
	\usepackage{paralist}
	\renewenvironment{itemize}[1]{\begin{compactitem}#1}{\end{compactitem}}
	\renewenvironment{enumerate}[1]{\begin{compactenum}#1}{\end{compactenum}}
	
	\setlength{\abovecaptionskip}{5pt}
	\setlength{\belowcaptionskip}{-15pt}
	
	\expandafter\def\expandafter\normalsize\expandafter{%
	    \normalsize
	    \setlength\abovedisplayskip{5pt}
	    \setlength\belowdisplayskip{5pt}
	    \setlength\abovedisplayshortskip{5pt}
	    \setlength\belowdisplayshortskip{5pt}
	}
	
	\newenvironment{algbox}[0]{
	\noindent 
	\begin{fminipage}{6.3in}
	}{
	\end{fminipage}
	}
	\usepackage{comment}
	\excludecomment{proof}
	
	\usepackage{titlesec}
	\titlespacing*{\section}{0pt}{0.2\baselineskip}{0.2\baselineskip}
	\titlespacing*{\subsection}{0pt}{0.2\baselineskip}{0.2\baselineskip}
	\titlespacing*{\theorem}{0pt}{0\baselineskip}{0\baselineskip}
	\titlespacing*{\lemma}{0pt}{0\baselineskip}{0\baselineskip}
}

\def\trace#1{\mathrm{Tr} \left(#1 \right)}

\def\pleq{\preccurlyeq}
\def\pgeq{\succcurlyeq}

\def\poly{\text{poly}}

\def\bvec#1{{\mbox{\boldmath $#1$}}}

\def\prob#1#2{\mbox{Pr}_{#1}\left[ #2 \right]}

\def\expec#1#2{{\mathbb{E}}_{#1}\left[ #2 \right]}

\def\defeq{\stackrel{\mathrm{def}}{=}}
\def\setof#1{\left\{#1  \right\}}
\def\sizeof#1{\left|#1  \right|}

\def\floor#1{\left\lfloor #1 \right\rfloor}
\def\ceil#1{\left\lceil #1 \right\rceil}

\def\union{\cup}

\def\abs#1{\left|#1  \right|}

\def\norm#1{\left\| #1 \right\|}

\newcommand\bb{\boldsymbol{\mathit{b}}}
\newcommand\bbbar{\overline{\boldsymbol{\mathit{x}}}}
\newcommand\cc{\boldsymbol{\mathit{c}}}
\newcommand\dd{\boldsymbol{\mathit{d}}}

\newcommand\ww{\boldsymbol{\mathit{w}}}

\newcommand\xx{\boldsymbol{\mathit{x}}}
\newcommand\xxbar{\overline{\boldsymbol{\mathit{x}}}}

\renewcommand\AA{\boldsymbol{\mathit{A}}}
\newcommand\BB{\boldsymbol{\mathit{B}}}
\newcommand\CC{\boldsymbol{\mathit{C}}}
\newcommand\DD{\boldsymbol{\mathit{D}}}

\newcommand\II{\boldsymbol{\mathit{I}}}
\newcommand\JJ{\boldsymbol{\mathit{J}}}

\newcommand\MM{\boldsymbol{\mathit{M}}}
\newcommand\Mtil{\boldsymbol{\widetilde{\mathit{M}}}}
\newcommand\PP{\boldsymbol{\mathit{P}}}

\newcommand\LL{\boldsymbol{\mathit{L}}}

\newcommand\UU{\boldsymbol{\mathit{U}}}
\newcommand\WW{\boldsymbol{\mathit{W}}}

\newcommand\XX{\boldsymbol{\mathit{X}}}
\newcommand\YY{\boldsymbol{\mathit{Y}}}

\newcommand\ZZ{\boldsymbol{\mathit{Z}}}

\newcommand\DDhat{\boldsymbol{\widehat{\mathit{D}}}}

\newcommand\MMhat{\boldsymbol{\widehat{\mathit{M}}}}

\newcommand\UUhat{\boldsymbol{\widehat{\mathit{U}}}}

\newcommand\xxhat{\boldsymbol{\widehat{\mathit{x}}}}
\newcommand\bbhat{\boldsymbol{\widehat{\mathit{b}}}}

\newcommand\xxtil{\boldsymbol{\tilde{\mathit{x}}}}

\newcommand\Ghat{{\widehat{{G}}}}
\newcommand\Vhat{{\widehat{{V}}}}
\newcommand\Ehat{{\widehat{{E}}}}
\newcommand\ddhat{{\hat{\dd}}}
\newcommand\dhat{{\hat{{d}}}}
\newcommand\nhat{{\hat{{n}}}}
\def\Gtil{\widetilde{G}}

\newcommand{\schur}[2]{Sc \left(#1,  #2\right) }

\newcommand{\todo}[1]{}
\newcommand{\yintat}[1]{}
\newcommand{\dan}[1]{}
\newcommand{\richard}[1]{}

\newcommand{\one}{\mathbf{1}}
\newcommand{\diag}{\textsc{diag}}

\usepackage{thm-restate}

\begin{document}

\title{
Sparsified Cholesky Solvers for SDD linear systems}
\author{
Yin Tat Lee
\thanks{Supported in part by NSF awards 0843915 and 1111109.  Part of this work was done while visiting the Simons Institute for the Theory of Computing, UC Berkeley.}
\\
M.I.T.\\
yintat@mit.edu
\and
Richard Peng\\
M.I.T.\\
rpeng@mit.edu
\and
Daniel A. Spielman\thanks{Supported by AFOSR Award FA9550-12-1-0175,
   NSF grant CCF-1111257, a Simons Investigator Award, and a MacArthur Fellowship.}
\\ 
Yale University\\
spielman@cs.yale.edu
}

\maketitle

\begin{abstract}
We show that Laplacian and 
  symmetric diagonally dominant (SDD) matrices can be well 
  approximated by linear-sized sparse Cholesky factorizations.
Specifically, $n \times n$ matrices of these types have constant-factor
  approximations of the form $\LL  \LL^{T}$, where $\LL$ is a
  lower-triangular matrix with $O(n)$ non-zero entries.
This factorization allows us to solve linear systems in
  such matrices in $O(n)$ work and $O(\log{n}\log^2\log{n})$ depth.

We also present nearly linear time algorithms that construct solvers that
  are almost this efficient.
In doing so, we give the first nearly-linear work routine for
constructing spectral vertex sparsifiers---that is, spectral
approximations of Schur complements of Laplacian matrices.


\end{abstract}

\iffull{}{
\pagenumbering{gobble}

\newpage

\pagenumbering{arabic}
}

\section{Introduction}

There have been incredible advances in the design of algorithms for 
  solving systems of linear equations in Laplacian and symmetric, diagonally dominant (SDD) matrices.
Cohen \textit{et. al.} \cite{CohenKMPPRX} have recently designed algorithms that find $\epsilon$-approximate solutions to such
  systems of equations in time $O (m \log^{1/2} n \log \epsilon^{-1})$,
  where $n$ is the dimension of the matrix and $m$ is its number of nonzero entries.
Peng and Spielman \cite{PengS14}
  recently discovered the first parallel algorithms that require only poly-logarithmic time and nearly-linear work.
In this paper, we prove that for every such matrix
  there is an operator that approximately solves equations in this matrix
  and that can be evaluated in linear work and depth $O(\log n (\log \log n)^{2})$.
These operators are analogous to the LU decompositions produced by Gaussian elimination:
  they take longer to compute than to apply.

We present two fast parallel algorithms for finding solvers that are almost as fast.
One runs in nearly linear time and polylogarithmic depth (Theorem \ref{thm:black_box}).
The algorithm presented in Theorem \ref{thm:non_combin_result}
 has preprocessing depth $n^{o (1)}$, but is more efficient in terms of work 
and produces a solver whose work and depth are
 within a logarithmic factor of the best one we can show exists.
   
A matrix $\AA$ is diagonally dominant if each of its diagonal entries is
  at least the sum of the absolute values of the off-diagonal entries in
  its row. 
The most famous symmetric, diagonally dominant matrices
  are the Laplacian matrices of graphs: those with non-positive off-diagonal
  such that every diagonal is exactly equal to the sum of the absolute values
  of the off-diagonal entries in its row.
Laplacian and SDD matrices arise in many applications, including the solution
  of optimization problems such as maximum flow \cite{ChristianoEtAl,KelnerMillerPeng,LeeRS13,Madry13},
  minimium cost flow \cite{daitch2008faster,lsMaxflow},
  semi-supervised learning \cite{Zhu03},
  and the solution of elliptic PDEs \cite{BomanHV04}.

Building on the work of  Vaidya \cite{Vaidya},
  Spielman and Teng \cite{SpielmanTengLinsolve} discovered that through
  the use of two constructions in graph theory---sparsifiers and low stretch spanning trees---one
  could design algorithms for solving such linear equations that run
  in nearly-linear time.
Kelner \textit{et. al.} \cite{KOSZ} construct an elementary algorithm
  for solving SDD systems in nearly linear time that only makes use of
  low stretch spanning trees.
Conversely, Peng and Spielman \cite{PengS14} design an algorithm that
  only uses sparsifiers.
The present paper builds on their approach.

The parallel algorithm of Peng and Spielman \cite{PengS14}
  approximates the inverse of a matrix by the sum and product
  of a small number of sparse matrices.
The main bottleneck in their algorithm is that all of the matrices it produces
  have the same dimension, and that the number of these matrices
  depends on the condition number of the system to be solved.
This leads to each matrix having an average number of nonzero entries
  per column that is proportional to the square of the logarithmic of
  the condition number, leading to work $O ((m + n \log^{3}\kappa) \log \epsilon^{-1})$.


Our result improves on the construction of Peng and
  Spielman~\cite{PengS14} in a number of ways.
First, the depth and work of our new algorithms are independent
  of the condition number of the matrix.
Second, the matrices in the product that approximates the inverse are
  of geometrically decreasing sizes.
This leads to much faster algorithms.
That said, 
  our efficient algorithms for constructing solvers and spectral vertex sparsifiers
  critically relies on their work.

We introduce sparsified Cholesky factorization
  in 
  in Section \ref{sec:vertexReduce}, where we prove
  that the inverse of every SDD matrix $\AA $ 
  can be approximated by an operator that can be evaluated
  in linear work and depth $O (\log^{2} n \log \log n)$.
By using this operator as a preconditioner, or by applying iterative refinement,
  this leads to a solver that produces $\epsilon$-approximate solutions to systems
  in $\AA$ in work $O (m \log \epsilon^{-1})$ and depth $O (\log^{2} n \log \log n \log \epsilon^{-1})$,
  where $m$ is the number of nonzeros in $\AA$.
We begin by eliminating a block consisting of a constant fraction of the vertices.
The elimination of these vertices adds edges to the subgraph induced on the remaining vertices.
We use the work of \cite{BSS} to sparsify the modified subgraph
  (Figure \ref{fig:applyChain}, Lemma \ref{lem:apply_chain} and Theorem \ref{thm:result_BSS}).
The choice of which vertices we eliminate is important.
We use subset of vertices whose degrees in their induced subgraph are substantially
  smaller than in the original graph (see Definition \ref{def:strongDD} and Lemma \ref{lem:subsetSimple}).

In Section \ref{sec:existence} we show how to convert this solver into
  a sparse approximate inverse.
That is, we show that $\AA$ can be approximated by a product of the form $\UU^T \DD \UU$ where $\UU$
  an upper-triangular matrix with $O (n)$ nonzero entries and $\DD$ is 
  diagonal.
While we can construct this $\UU$ and $\DD$ in polynomial time,
  we do not yet have a nearly linear time or low depth efficient parallel algorithm that does so.

\iffull{We obtain our best existence result in Section \ref{sec:depth}
  by reducing the depth of the parallel solvers by a logarthmic
  factor.}{In the full version, we obtain our best existence result by reducing the depth of the parallel solvers by a logarithmic
  factor. }The reduction comes from observing that the construction of Section \ref{sec:existence}
  would have the desired depth if every vertex in $\AA$ and in the smaller
  graphs produced had bounded degree.
While we can use sparsification to approximate an arbitrary graph by a sparse one,
  the sparse one need not have bounded degree.
We overcome this problem by proving that the Laplacian of every graph
  can be approximated by a Schur complement
  of the Laplacian of a larger graph of bounded degree \iffull{(Theorem \ref{thm:degreeReduction})}{}.


We then
  turn to the problem of computing our solvers efficiently in parallel.
The first obstacle is that we must quickly compute an approximation
  of a Schur complement of a set of vertices without actually constructing
  the Schur complement, as it could be too large.
This is the problem we call \textit{Spectral Vertex Sparsification}.
It is analogous to the problem of vertex sparsfication for cut and combinatorial flow problems
  \cite{LeightonM10,Moitra13}:
  given a subset of the vertices we must compute a graph on those vertices
  that allows us to compute
  approximations of electrical flows in the original graph between vertices in that subset.
In contrast with cut and combinatorial flow problems,
  there is a graph that allows for this computation exactly on the subset of vertices,
  and it is the Schur complement in the graph Laplacian.
In Section~\ref{sec:vertexSparsify}, we build on the techniques of \cite{PengS14} to give an efficient algorithm for
  spectrally approximating Schur complements.

The other obstacle is that we need to compute sparsifications of graphs efficiently
  in parallel.
We examine two ways of doing this in Section~\ref{sec:algo}.
The first, examined in Section \ref{ssec:blackBox}, is to use a black-box parallel algorithm for graph sparsification,
  such as that of Koutis \cite{Koutis14}.
This gives us our algorithm of best total depth.
The second, examined in Section \ref{ssec:recursive}, employs a recursive
  scheme in which we solve smaller linear systems to compute
  probabilities with which we sample the edges, as in \cite{SpielmanS08:journal}.
Following \cite{cohen2014uniform}, these smaller linear systems
  are obtained by crudely sub-sampling the original graph.
The resulting algorithm runs in depth $n^{o (1)}$, but produces
  a faster solver.
We expect that further advances in graph sparsification
such as~\cite{Allen-ZhuLL15} will result in even better algorithms.


\iffull{\section{Some Related Work}
\label{sec:related}

Gaussian elimination solves systems of equations in a matrix $\AA$ by computing
  lower and upper triangular matrices $\LL$ and $\UU$ so that $\AA = \LL \UU$.
Equations in $\AA$ may then be solved by solving equations in $\LL$ and $\UU$,
  which takes time proportional to the number of nonzero entries in those matrices.
This becomes slow if $\LL$ or $\UU$ has many nonzero entries, with is often the case.

Cholesky factorization is the natural symmetrization of this process: it writes symmetric
  matrices $\AA$ as a product $\LL \LL^{T}$.
Incomplete Cholesky factorizations \cite{ICC} instead approximate $\AA$ by a product of sparse matrices
  $\LL \LL^{T}$ by strategically dropping some entries in the computation of Cholesky factors.
One can then use these approximations as preconditioners to compute highly accurate solutions to systems in $\AA$.
While this is a commonly used heuristic, there have been few general theoretical analyses of the performance of the resulting algorithms.
Interestingly, Meijerink and van der Vorst \cite{ICC} analyze the performance of this algorithm on SDD matrices
  whose underlying graph is a regular grid. 

SDD linear systems have been extensively studied in scientific computing
  as they arise when solving elliptic partial differential equations.
Multigrid methods have proved very effective at solving the resulting systems.
Fedorenko~\cite{fedorenko1964speed} gave the first multigrid method for SDD systems
  on regular square grids and  proved that it is an nearly-linear time algorithm.
Multigrid methods have since been used to solve many types of linear systems
  \cite{brandt1977multi,hackbusch1985multi},
  and have been shown to solve special systems in 
  linear work and logarithmic depth \cite{nicolaides1978multigrid,hackbusch1982multi} under some smoothness assumptions.
Recently, Artem and Yvan~\cite{napov2012algebraic} gave the first algebraic multigrid method
  with a guranteed convergence rate.
However, to the best of our knowledge, a worst-case nearly-linear work bound
  has not been proved for any of these algorithms.

Our algorithm is motivated both by multigrid methods and incomplete Choleksy factorizations.
Both exploit the fact that elimination operations in SDD matrices result in SDD matrices.
That is, Schur complements of SDD matrices result in SDD matrices with fewer vertices.
However, where multigrid methods eliminate a large fraction of vertices at each level,
  our algorithms eliminate a small but constant fraction.
The main novelty of our approach is that we sparsify the resulting Schur complement.
A heuristic approach to doing this was recently studied by 
  Krishnan, Fattal, and Szeliski \cite{krishnan2013efficient}.

\dan{I cut a lot from this.}

}{}

\section{Background}

We will show that diagonally dominant matrix $A$ can be
  well-approximated by a product $\UU^{T} \DD \UU$ where $\UU$ is upper-triangular and sparse
  and $\DD$ is diagonal.
By solving linear equations in each of these matrices, we can quickly solve a
  system of linear equations in $A$.
We now  review the notion of approximation that we require
  along with some of its standard properties.

For symmetric matrices $A$ and $B$, we write
  $A \pgeq B$ if $A - B$ is positive semidefinite.
The ordering given by $\pgeq$ is called the ``Loewner partial order''.
\iffull{\begin{fact}\label{fact:orderInverse}
For $A$ and $B$ positive definite,
$A \pgeq B$ if and only if $B^{-1} \pgeq A^{-1}$.
\end{fact}
\begin{fact}\label{fact:orderCAC}
If $A \pgeq B$ and $C$ is any matrix of compatible dimension,
  then $C A C^{T} \pgeq C B C^{T}$.
\end{fact}

}{}We say that $\AA$ is an $\epsilon$-approximation of $\BB$, written
  $\AA  \approx_{\epsilon} \BB $,
if
\[
  e^{\epsilon} \BB \pgeq \AA \pgeq e^{-\epsilon} \BB.
\]
Observe that this relation is symmetric.
\iffull{
Simple arithmetic yields the following fact about compositions of approximations.
\begin{fact}\label{frac:orderComposition}
If $\AA \approx_{\epsilon} \BB$
  and $\BB \approx_{\delta } \CC$, then
  $\AA \approx_{\epsilon + \delta} \CC$.
\end{fact}

}{}

\iffull{
We say that $\xxtil$ is an $\epsilon$-approximate solution to the
  system $\AA \xx = \bb$ if
\[
  \norm{\xxtil - \AA^{-1} \bb}_{\AA} \leq \epsilon \norm{\xx}_{\AA},
\]
 where
\[
  \norm{\xx}_{\AA} =  (\xx^{T} \AA \xx)^{1/2}.
\]}
{
We say that $\xxtil$ is an $\epsilon$-approximate solution to the
  system $\AA \xx = \bb$ if
  $\norm{\xxtil - \AA^{-1} \bb}_{\AA} \leq \epsilon \norm{\xx}_{\AA}$,
 where $\norm{\xx}_{\AA} =  (\xx^{T} \AA \xx)^{1/2}$.
}
This is the notion of approximate solution typically used when
  analyzing preconditioned linear system solvers, and it is the notion assumed in the
  works we reference that use these solvers as subroutines.
\iffull{
\begin{fact}\label{fact:approxSolve}
If $\epsilon<1/2$, $\AA \approx_{\epsilon} \BB$ and 
  $\BB \xxtil = \bb$, then
  $\xxtil$ is a $2 \sqrt{\epsilon}$ approximate solution to
  $\AA \xx = \bb$.
\end{fact}


}{}So, if one can find a matrix $\BB$ that is a good approximation of
  $\AA$ and such that one can quickly solve linear equations in $\BB$, then
  one can quickly compute approximate solutions to systems of linear
  equations in $\AA$.
Using methods such as \textit{iterative refinement}, one can use
  multiple solves in $\BB$ and multiplies by $\AA$ to obtain
  arbitrarily good approximations.
For example, if $\BB$ is a constant approximation of $\AA$, then
  for every $\epsilon < 1$, one can obtain an $\epsilon$ approximate
  solution of a linear system in $\AA$ by performing 
  $O (\log (\epsilon^{-1}))$ solves in $\BB$ and multiplies by $\AA$
  (see, for example, \cite[Lemma 4.4]{PengS14}).

It is known that one can reduce the problem of solving systems of equations in SDD matrices
  to either the special case of Laplacian matrices or SDDM matrices---the family of SDD matrices that
  are nonsingular and have non-positive off diagonal entries (see, e.g. \cite{SpielmanTengLinsolve,CohenKMPPRX}).
We will usually consider SDDM matrices.
Every SDDM matrix $\AA$ can be uniquely written as a sum $\LL + \XX $ where $\LL$ is a Laplacian matrix
  and $\XX$ is a nonnegative diagonal matrix.
  
The main properties of SDDM matrices that we exploit are that they are closed under Schur complements
  and that they can be \textit{sparsified}.
The stongest known sparsifications come from the main result of \cite{BSS}, which implies the following.

\begin{theorem}\label{thm:BSS}
For every $n$-dimensional SDDM matrix $\AA $
  and every $\epsilon \leq 1$, there is a SDDM
   matrix $\BB $ having at most $10 n / \epsilon^{2}$ nonzero entries that is
  an $\epsilon$-approximation of $\AA$. In particular, the number of non-zero
  entries in $\BB$ above the diagonal is at most $4.1 n / \epsilon^{2}$.
\end{theorem}
\iffull{
While the matrix $\BB$ guaranteed to exist by this theorem may be
  found in polynomials time, this is not fast enough for the
  algorithms we desire.
So, we only use Theorem~\ref{thm:BSS} to prove existence results.
We
  later show how to replace it with faster algorithms, at some expense
  in the quality of the sparsifiers we produce.

}{}

\section{Block Cholesky Factorization}\label{sec:cholesky}

Our algorithm uses block-Cholesky factorization to eliminate
  a block of vertices all at once.
We now review how block-Cholesky factorization works.

To begin, we remind the reader that Cholesky factorization is the natural
  way of performing Gaussian elimination
  on a symmetric matrix:
  by performing eliminations on rows and columns simultaneously, one preserves the   
  symmetry of the matrix.
The result of Cholesky factorization is a representation of a matrix $\MM$
  in the form $\UU^{T} \UU$,
  where $\UU$ is an upper-triangular matrix.
We remark that this is usually written as $\LL \LL^{T}$ where $\LL$ is lower-triangular.
We have chosen to write it in terms of upper-triangular matrices so as to avoid confusion with the
  use of the letter $\LL$ for Laplacian matrices.

To produce matrices $\UU$ with 1s on their diagonals, and 
  to avoid the computation of square roots, one often instead forms a factorization
  of the form $\UU^{T} \DD \UU$, where $\DD$ is a diagonal matrix.
Block-Cholesky factorization forms a factorization of this form, but with $\DD$ 
  being a block-diagonal matrix.

To begin, we must choose a set of rows to be eliminated.
We will eliminate the same set of columns.
For consistency with the notation used in the description of multigrid algorithms,
  we will let $F$ (for finer) be the set of rows to be eliminated.
We then let $C$ (for coarse) be the remaining set of rows.
In contrast with multigrid methods, we will have $\sizeof{F} < \sizeof{C}$.
By re-arranging rows and colums, we can write $\MM$
  in block form:
\[
\MM
=
\left[
\begin{array}{cc}
\MM_{FF} &\MM_{FC}\\
\MM_{CF}  & \MM_{CC}
\end{array}
\right].
\]
Elimination of the rows and columns in $F$ corresponds to writing
\begin{equation} \label{eqn:blockformula}
\MM
= 
\begin{bmatrix}
\II  & 0\\
\MM_{CF} \MM_{FF}^{-1}   & \II
\end{bmatrix}
\begin{bmatrix}
\MM_{FF} & 0\\
0  & \MM_{CC} - \MM_{CF} \MM_{FF}^{-1} \MM_{FC}
\end{bmatrix}
\begin{bmatrix}
\II  &  \MM_{FF}^{-1} \MM_{FC } \\
0   & \II
\end{bmatrix}.
\end{equation}
Note that the left and right matrices are lower and upper triangular.
The matrix in the lower-right block of the middle matrix is the Schur
  complement of $F$ in $\MM$.
We will refer to it often by the notation
\[
\schur{\MM}{F} \defeq \MM_{CC} - \MM_{CF} \MM_{FF}^{-1} \MM_{FC}.
\]
We remark that one can solve a linear system in $\schur{\MM}{F}$
  by solving a system in $\MM$: one just needs to put zeros coordinates corresponding
  to $F$ in the right-hand-side vector.


Recall that
\begin{equation}\label{eqn:blockLinverse}
\begin{bmatrix}
\II  & 0\\
\MM_{CF} \MM_{FF}^{-1}   & \II
\end{bmatrix}^{-1}
=
\begin{bmatrix}
\II  & 0\\
-\MM_{CF} \MM_{FF}^{-1}   & \II
\end{bmatrix}.
\end{equation}
So, if we can quickly multiply by this last matrix,
  and if we can quickly solve linear systems in $\MM_{FF}$
  and in the Schur complement, then we can quickly solve systems in $\MM$.
Algebraically, we exploit the following identity:
\begin{fact}\label{fact:blockInverse}
\begin{equation}\label{eqn:blockInverse}
\MM^{-1}
=
\left[
\begin{array}{cc}
\II & -\MM_{FF}^{-1} \MM_{FC}\\
0 & \II
\end{array}
\right]
\left[
\begin{array}{cc}
\MM_{FF}^{-1} & 0 \\
0 & \schur{\MM}{F}^{-1}
\end{array}
\right]
\left[
\begin{array}{cc}
\II & 0\\
-\MM_{CF} \MM_{FF}^{-1} & \II
\end{array}
\right].
\end{equation}
\end{fact}


Our algorithms depend upon the following 
important property of Schur complements of SDDM matrices.

\begin{fact}
If $\MM$ is a SDDM matrix and $F$ is a subset of its columns,
the Schur complement
  $\schur{\MM}{F}$ is also a SDDM matrix.
\end{fact}

We now mention two other facts that we will use about the
  $\pleq $ order and Schur complements.

\begin{fact}\label{fact:blockSubstitute}
If $\MM_{FF} \pleq \Mtil_{FF}$, then
\[
\begin{pmatrix}
\MM_{FF} &\MM_{FC}\\
\MM_{CF}  & \MM_{CC}
\end{pmatrix}
\pleq 
\begin{pmatrix}
\Mtil_{FF} &\MM_{FC}\\
\MM_{CF}  & \MM_{CC}
\end{pmatrix}.
\]
\end{fact}

\begin{fact}[Lemma B.1. from~\cite{MillerP13}]
\label{fact:schurLoewner}
If $\MM$ and $\Mtil$ are positive semidefinite
matrices satisfying $\MM \preceq \Mtil$,
then
\[
\schur{\MM}{F} \preceq \schur{\Mtil}{F}.
\]
\end{fact}

The first idea that motivates our algorithms is that we can sparsify $\MM$
  and $\schur{\MM}{F}$.
If $\MM$ is sparse, then we can quickly multiply vectors by $\MM_{FC}$.
However, to be able to quickly apply the factorization of $\MM^{-1}$ given
  in Fact~\ref{fact:blockInverse}, we also need to be able to quickly
  apply $\MM_{FF}^{-1}$.
If we can do that, then we can 
  quickly solve systems in $\MM$ by recursively solving 
  systems in $\schur{\MM}{F}$.

The easiest way to find an $F$ for which we could quickly apply
  $\MM_{FF}^{-1}$ would be to choose $F$ to be a large independent set,
  in which case $\MM_{FF}$ would be diagonal.
Such a set $F$ must exist as we can assume $\MM$ is sparse.
However, the independent set we are guaranteed to find by the sparsity of $\MM$ is not big enough:
  if we repeatedly find large independent sets and then sparsify the resulting Schur complements,
  the error that accumulates could become too big.
The second idea behind our algorithms is that we can find a large set $F$ for which
  $\MM_{FF}$ is well-approximated by a diagonal matrix.
This will allow us to apply $\MM_{FF}^{-1}$ quickly.
In the next section, we show that a very good choice of $F$ always exists, and that
  the use of such sets $F$ yields nearly-optimal algorithms for solving linear systems in $\MM$.

In order to make the entire algorithm efficient, we are still left with the problem of quickly
  computing a sparsifier of the Schur complement.
In Section~\ref{sec:vertexSparsify}, we show how to quickly compute and use \textit{Spectral Vertex Sparsifiers},
  which are sparsifiers of the Schur complement.
In particular, we do this by expressing the Schur complement as the sum of the Schur complements of two simpler matrices:
  one with a diagonal $FF$ block, and the other with a better conditioned $FF$ block.
We handle the matrix with the diagonal block directly, and the matrix with the better conditioned block recursively.




\section{A Polynomial Time Algorithm for Optimal Solver Chains}
\label{sec:vertexReduce}


Our algorithms will begin by eliminating a set of vertices $F$
  that is $\alpha$-strongly diagonally dominant, a concept that we now define.

\begin{definition}
\label{def:strongDD}
A symmetric matrix $\MM$ is $\alpha$-strongly diagonally dominant
  if for all $i$ 
\[
\MM_{ii} \geq (1+\alpha) \sum_{j \neq i} \left| \MM_{ij} \right|.
\]
We say that a subset $F$ of the rows of a matrix $\MM$
  is $\alpha$-strongly diagonally dominant if $\MM_{FF}$
  is an $\alpha$-strongly diagonally dominant matrix.
\end{definition}

We remark that $0$-strongly diagonal dominance coincides with the
  standard notion of weak diagonal dominance.
In particular, Laplacian matrices are $0$-strongly diagonally dominant.

It is easy to find an $\alpha$-strongly diagonally dominant subset
  containing at least an $1/8 (1+\alpha)$ fraction of the rows of an SDD
  matrix: one need merely pick a random subset and then discard the rows
  that do not satisfy the condition.

\iffull{
Pseudocode for computing such a subset is given in Figure~\ref{fig:randF}.

\begin{figure}[h]
\begin{algbox} $F=\textsc{SDDSubset}(\MM , \alpha)$, where $\MM$
is an $n$-dimensional SDD matrix.
\begin{enumerate}
\item Let $F'$ be a uniform random subset of $\setof{1, \dots , n}$ of size
$\frac{n}{4(1+\alpha)}$.
\label{ln:randSubset}
\item Set \[
F=\left\{ i\in F'\text{ such that }\sum_{j \in F', j \neq i}\left|\MM_{ij}\right|\leq\frac{1}{1 + \alpha} \left|\MM_{ii}\right|\right\} .
\]
\item If $|F| < \frac{n}{8(1 + \alpha)}$, goto Step~\ref{ln:randSubset}.
\item Return $F$
\end{enumerate}
\end{algbox}

\caption{Routine for Generating an $\alpha$-strongly diagonally dominant
subset $F$}

\label{fig:randF}
\end{figure}
}{}

\begin{lemma}
\label{lem:subsetSimple}
For every $n$-dimensional SDD matrix $\MM$ and every $\alpha \geq 0$, 
  \textsc{SDDSubset} computes an $\alpha$-strongly diagonally dominant subset $F$
  of size at least $n/(8 (1+\alpha))$ in  $O(m)$ expected work and $O(\log{n})$ expected depth,
  where $m$ is the number of nonzero entries in $\MM$.
\end{lemma}

\begin{proof}
As $F$ is a subset of $F'$,
\[
  \sum_{j \in F, j \neq i}\left|\MM_{ij}\right| \leq   \sum_{j \in F', j \neq i}\left|\MM_{ij}\right|.
\]
So, when the algorithm does return a set $F$, it is guaranteed to be $\alpha$-strongly diagonally dominant.  

We now show that the probability that the algorithm finishes in each iteration is at least $1/2$.
Let $A_{i}$ be the event that $i \in F'$ and that $i \not \in F$.
This only happens if $i \in F'$ and
\begin{equation}\label{eqn:subsetSimple}
  \sum_{j \in F', j \not = i} \abs{\MM_{ij}}
>
  \frac{1}{1+\alpha } \abs{\MM_{ii}}.
\end{equation}
The set $F$ is exactly the set of $i \in F'$ for which $A_{i}$ does not hold.

Given that $i \in F'$, the probability that each other $j \not = i$ is in $F'$ is
\[
  \frac{1}{n-1}  \left(\frac{n}{4 (1+\alpha)}-1 \right) .
\]
So,
\[
  \expec{}{\sum_{j \in F', j \not = i } \abs{\MM_{ij}} \Big| i \in F'}
\leq 
\frac{1}{n-1}
  \left(\frac{n}{4 (1+\alpha)}-1 \right)  \sum_{j \not = i} \abs{\MM_{ij}}
<
  \frac{1}{4 (1+\alpha)}  \sum_{j \not = i} \abs{\MM_{ij}}
\leq 
  \frac{1}{4 (1+\alpha)}  \abs{\MM_{ii}},
\]
as $\MM$ is strongly diagonally dominant.
So, Markov's inequality tells us that
\[
  \prob{}{
\sum_{j \in F', j \not = i } \abs{\MM_{ij}} 
>
  \frac{1}{1+\alpha}  \abs{\MM_{ii}}
\Big| i \in F'
  }
< 1/4,
\]
and thus
\[
  \prob{}{A_{i}} = \prob{}{i \in F'} \prob{}{i \not \in F | i \in F'}
<
  \frac{1}{4 (1+\alpha)} \frac{1}{4}
=
  \frac{1}{16 (1+\alpha)}.
\]
Again applying Markov's inequality allows us to conclude
\[
  \prob{}{\sizeof{\setof{i : A_{i}}} \geq \frac{n}{8 (1+\alpha)}} < 1/2.
\]

So, with probability at least $1/2$, $\sizeof{F} \geq n / 8 (1+\alpha)$,
  and the algorithm will pass the test in line 3.
Thus, the expected number of iterations made by the algorithm is at most $2$.
The claimed bounds on the expected work and depth of the algorithm follow.
\end{proof}


Strongly diagonally dominant subsets are useful because linear systems
involving them can be solved rapidly.
Given such a set $F$, we will construct 
  an operator $\ZZ_{FF}^{(k)}$ that approximates
  $\MM_{FF}^{-1}$
  and that can be applied quickly.
To motivate our construction, 
  observe that if $\MM_{FF} = \XX_{FF} + \LL_{FF}$ where $\XX_{FF}$ is a nonnegative diagonal matrix and
  $\LL_{FF}$ is a Laplacian, then
\[
  \MM_{FF}^{-1} = \XX_{FF}^{-1}
  - \XX_{FF}^{-1} \LL_{FF} \XX_{FF}^{-1} + 
  \sum_{i \geq 2} (-1)^{i} \XX_{FF}^{-1} (\LL_{FF} \XX_{FF}^{-1})^{i}.
\]
We will approximate this series by its first few terms:
\begin{equation}
\ZZ_{FF}^{(k)} \defeq  \sum_{i = 0}^{k} \XX_{FF}^{-1} \left(-\LL_{FF} \XX_{FF}^{-1} \right)^{i}.
\label{eqn:defZ}
\end{equation}

\iffull{
In the following lemmas, we show that using $\ZZ_{FF}$
  in place of $\MM_{FF}^{-1}$ in \eqref{eqn:blockInverse} 
   provides a good approximation
  of $\MM^{-1}$.
We begin by pointing out that 
  $\XX_{FF}$ is much greater than $\LL_{FF}$.
In particular, this implies that 
  all diagonal entries of $\XX_{FF}$ are positive, so that
  $\XX_{FF}^{-1}$ actually exists.

\begin{lemma}\label{lem:Xdom}
Let $\MM$ be a SDDM matrix that is $\alpha$-strongly diagonally dominant.
Write $\MM = \XX + \LL$ where $\XX$ is a nonnegative diagonal matrix and
  $\LL$ is a Laplacian.
Then,
\[
  \XX \pgeq \frac{\alpha}{2} \LL.
\]
\end{lemma}
\begin{proof}
Write $\LL = \YY - \AA$ where $\YY$ is diagonal and $\AA$ has zero diagonal.
As $\LL$ is diagonally dominant, so is $\YY + \AA$.
This implies that $\YY \pgeq -\AA$, and so
  $2 \YY \pgeq \LL$.

As $\MM$ is $\alpha$-strongly diagonally dominant and the diagonal of $\MM$
  is $\XX + \YY$,
\[
  ((\XX + \YY) \bvec{1})_{i} \geq (\alpha +1) (\AA \bvec{1})_{i}.
\]
As $\LL$ is a Laplacian, $\LL \bvec{1} = \bvec{0}$, which implies $\YY \bvec{1} = \AA \bvec{1}$ and
\[
  (\XX \bvec{1})_{i} \geq \alpha (\AA \bvec{1})_{i} = \alpha (\YY \bvec{{1}})_{i}.
\]
As both $\XX$ and $\YY$ are diagonal,
  this implies that 
\[
\XX \pgeq \alpha \YY \pgeq \frac{\alpha}{2} \LL .
\]
\end{proof}

We now bound the quality of approximation of the power series \eqref{eqn:defZ}.

\begin{lemma}\label{lem:lightBlock2}
Let $\MM$ be a SDDM matrix and let $F$ be a set of columns
  so that when we write $\MM_{FF} = \XX_{FF} + \LL_{FF}$
  with $\XX_{FF}$ nonnegative diagonal and $\LL_{FF}$ a Laplacian,
  we have
  $\LL_{FF} \pleq \beta \XX_{FF}$.
Then, for odd $k$ and for $\ZZ_{FF}^{(k)}$ as defined in \eqref{eqn:defZ} we have:
\begin{equation}\label{eqn:lightBlock2}
\XX_{FF} + \LL_{FF} 
\preceq (\ZZ_{FF}^{(k)})^{-1} \preceq
\XX_{FF} + \left( 1 + \delta  \right) \LL_{FF},
\end{equation}
where
\[
  \delta = \beta^{k} \frac{1+\beta}{1-\beta^{k+1}}.
\]
\end{lemma}
\begin{proof}
The left-hand inequlity is equivalent to the statement that
 all the eigenvalues of
 $\ZZ^{(k)}_{FF} (\XX_{FF} + \LL_{FF})$ are at most $1$
   (see \cite[Lemma 2.2]{SupportGraph} or
  \cite[Proposition 3.3]{SpielmanTengLinsolve}).
To see that this is the case, expand
\begin{align*}
\ZZ_{FF}^{(k)} (\XX_{FF} + \LL_{FF})
& = 
\left(\sum_{i=0}^{k} \XX_{FF}^{-1} (-\LL_{FF} \XX_{FF}^{-1})^{i} \right)
(\XX_{FF} + \LL_{FF})
\\
& = 
\sum_{i=0}^{k} (-\XX_{FF}^{-1} \LL_{FF})^{i}
-
\sum_{i=1}^{k+1} (\XX_{FF}^{-1} \LL_{FF})^{i}
\\
& =
\II_{FF} 
- (\XX_{FF}^{-1} \LL_{FF})^{k+1} .
\end{align*}
As all the eigenvalues of an even power of a matrix are nonnegative,
 all of the eigenvalues of this last matrix are at most $1$.

Similarly, the other inequality is equivalent to the assertion
 that all of the eigenvalues of 
 $\ZZ^{(k)}_{FF} (\XX_{FF} + (1+\delta ) \LL_{FF})$
 are at least one.
Expanding this product yields
\begin{multline*}
\left(\sum_{i=0}^{k} \XX_{FF}^{-1} (-\LL_{FF} \XX_{FF}^{-1})^{i} \right)
(\XX_{FF} + (1+\delta ) \LL_{FF})
\\
= 
\II_{FF} 
- (\XX_{FF}^{-1} \LL_{FF})^{k+1} 
+
\delta 
\sum_{i=0}^{k} (-1)^{i} (\XX_{FF}^{-1} \LL_{FF})^{i+1} 
\end{multline*}
The eigenvalues of this matrix are precisely the numbers
\begin{equation}\label{eqn:inLightBlock2}
1 - \lambda^{k+1} + \delta  \sum_{i=0}^{k} (-1)^{i} \lambda^{i+1},
\end{equation}
where $\lambda$ ranges over the eigenvalues of 
 $\XX_{FF}^{-1} \LL_{FF}$.
The assumption $\LL_{FF} \pleq \beta \XX_{FF}$ implies
 that the eigenvalues of  $\XX_{FF}^{-1} \LL_{FF}$
  are at most $\beta$, so
  $0 \leq \lambda \leq \beta$.
We have chosen the value of $\delta$ precisely to guarantee that,
 under this condition on $\lambda$, the value of \eqref{eqn:inLightBlock2}
 is at least $1$.
\end{proof}

We remark that this power series is identical to
the Jacobi iteration for solving linear systems.

The following lemma allows us to extend the approximation of $\MM_{FF}$
  by the inverse of $\ZZ_{FF}^{(k)}$ to the entire matrix $\MM$.
\begin{lemma}\label{lem:lightBlock3}
Under the conditions of Lemma \ref{lem:lightBlock2} and assuming that
  $0 \leq \beta  \leq 1/2$,
\[
\MM \pleq 
\begin{pmatrix}
\left(\ZZ_{FF}^{(k)} \right)^{-1} & \MM_{FC}\\
\MM_{CF} & \MM_{CC}
\end{pmatrix}
\pleq 
(1 + 2 \beta^{k}) \MM .
\]
\end{lemma}
\begin{proof}
The left-hand inequality follows immediately from
  Fact \ref{fact:blockSubstitute} and
  the left-hand side of \eqref{eqn:lightBlock2}.
To prove the right-hand inequality we apply
  Fact \ref{fact:blockSubstitute}
  and the right-hand side of \eqref{eqn:lightBlock2}
  to conclude
\[
\begin{pmatrix}
\left(\ZZ_{FF}^{(k)} \right)^{-1} & \MM_{FC}\\
\MM_{CF} & \MM_{CC}
\end{pmatrix}
\pleq 
\begin{pmatrix}
\MM_{FF} + \delta \LL_{FF}  & \MM_{FC}\\
\MM_{CF} & \MM_{CC}
\end{pmatrix}
=
\MM + 
\delta \begin{pmatrix}
 \LL_{FF}  & 0\\
0 & 0
\end{pmatrix}.
\]

Consider the (unique) decomposition of $\MM$ into
$\LL + \XX$ where $\LL$ is a graph Laplacian.
When viewed as graphs, $\LL_{FF}$ is a subgraph $\LL$,
which means:
\[
\begin{pmatrix}
 \LL_{FF}  & 0\\
0 & 0
\end{pmatrix}
\preceq \LL \preceq \MM,
\]
by which we may conclude that
\[
\MM + 
\delta \begin{pmatrix}
 \LL_{FF}  & 0\\
0 & 0
\end{pmatrix}
\pleq 
\MM + \delta \MM .
\]
To finish the proof, recall that $\delta = \beta^{k} (1+\beta) / (1- \beta^{k+1})$
 and observe that for $k \geq 1$ and $\beta \leq 1/2$, $ \delta \leq 2 \beta^{k}$.
\end{proof}

}{}


We now show that we can obtain a good approximation of $\MM^{-1}$
  by replacing $\MM_{FF}^{-1}$ by $\ZZ_{FF}^{(k)}$ in the three places
  in which it explicitly appears in \eqref{eqn:blockInverse},
  but not in the Schur complement.

\begin{lemma}\label{lem:subZFF}
Let $\MM$ be a SDDM matrix and let $F$ be an $\alpha$-diagonally dominant
  set of columns for some $\alpha \geq  4$.
Then, for $k$ odd and $\ZZ^{(k)}$ as defined in \eqref{eqn:defZ},
\[
\left[
\begin{array}{cc}
\II & -\ZZ^{(k)}_{FF} \MM_{FC}\\
0 & \II
\end{array}
\right]
\left[
\begin{array}{cc}
\ZZ^{(k)}_{FF} & 0 \\
0 & \schur{\MM}{F}^{-1}
\end{array}
\right]
\left[
\begin{array}{cc}
\II & 0\\
-\MM_{CF} \ZZ^{(k)}_{FF} & \II
\end{array}
\right]
\approx_{\gamma} \MM^{-1},
\]
for $\gamma = 2 (2/\alpha)^{k}$.
\end{lemma}

\begin{proof}
Define
\[
  \MMhat = 
\left[
\begin{array}{cc}
(\ZZ^{(k)}_{FF})^{-1} &\MM_{FC}\\
\MM_{CF}  & \MM_{CC}
\end{array}
\right].
\]

Lemma~\ref{lem:Xdom} tells us that $\MM$ satisfies the conditions of Lemma~\ref{lem:lightBlock2}
  with $\beta = 2/\alpha$.
So, Lemma~\ref{lem:lightBlock3} implies 
\[
\MM \pleq \MMhat
\pleq \left(1+ \gamma\right) \MM.
\]
By facts~\ref{fact:blockInverse} and \ref{fact:orderInverse}, this implies
\[
\MM^{-1}
\pgeq 
\left[
\begin{array}{cc}
\II & -\ZZ^{(k)}_{FF} \MM_{FC}\\
0 & \II
\end{array}
\right]
\left[
\begin{array}{cc}
\ZZ^{(k)}_{FF} & 0 \\
0 & \schur{\MMhat}{F}^{-1}
\end{array}
\right]
\left[
\begin{array}{cc}
\II & 0\\
-\MM_{CF} \ZZ^{(k)}_{FF} & \II
\end{array}
\right]
\pgeq 
(1+\gamma)^{-1} \MM^{-1}.
\]
From Facts \ref{fact:schurLoewner} and \ref{fact:orderInverse}, we know that
\[
\schur{\MM}{F}^{-1}
\pgeq 
\schur{\MMhat}{F}^{-1}
\pgeq 
(1+\gamma)^{-1} \schur{{\MM}}{F}^{-1}.
\]
When we use Fact \ref{fact:orderCAC}  to substitute this inequality into the one above,
  we obtain
\[
(1+\gamma )\MM^{-1}
\pgeq 
\left[
\begin{array}{cc}
\II & -\ZZ^{(k)}_{FF} \MM_{FC}\\
0 & \II
\end{array}
\right]
\left[
\begin{array}{cc}
\ZZ^{(k)}_{FF} & 0 \\
0 & \schur{\MM}{F}^{-1}
\end{array}
\right]
\left[
\begin{array}{cc}
\II & 0\\
-\MM_{CF} \ZZ^{(k)}_{FF} & \II
\end{array}
\right]
\pgeq 
(1+\gamma)^{-1} \MM^{-1},
\]
which implies the lemma.
\end{proof}

We now use Lemma~\ref{lem:subZFF} to analyze a solver obtained by iteratively sparsifying Schur complements
  of strongly diagonally dominant subsets.
We refer to the sequence of subsets and matrices obtained as a
  \textit{vertex sparsifer chain},
  as an approximation of a Schur complement is a spectral vertex sparsifier.
In the following definition, $\MM^{(1)}$ is intended to be a sparse approximation of
  $\MM^{(0)}$.
The sparsity of the matrices will show up in the analysis of the runtime, but not in the
  definition of the chain.


\begin{definition}[Vertex Sparsifier chain]
\label{def:vertexSparsifierChain}
For any SDDM matrix $\MM^{(0)}$, a vertex sparsifier chain of $\MM^{(0)}$ 
  with parameters $\alpha_i \geq 4$ and $ 1/2\geq\epsilon_i>0$
  is a sequence of matrices and subsets $(\MM^{(1)}, \ldots, \MM^{(d)};F_1, \ldots, F_{d-1})$
 such that:
\begin{enumerate}
\item $\MM^{(1)} \approx_{\epsilon_0} \MM^{(0)}$,
\item $\MM^{(i + 1)} \approx_{\epsilon_i} \schur{\MM^{(i)}}{F_i}$,
\item $\MM^{(i)}_{F_i F_i}$ is $\alpha_i$-strongly diagonally dominant and
\item $\MM^{(d)}$ has size $O(1)$.
\end{enumerate}
\end{definition}

We present pseudocode that uses a vertex sparsifier chain to approximately solve
  a system of equations in $\MM^{(0)}$ in Figure~\ref{fig:applyChain}.
We analyze the running time and accuracy of this algorithm
  in Lemma~\ref{lem:apply_chain}.

\begin{figure}[h]
\begin{algbox} $\xx^{(1)} = \textsc{ApplyChain}(\MM^{(1)}, \ldots, \MM^{(d)}, F_1, \ldots, F_{d-1}, \alpha_{1} \ldots \alpha_{d - 1}, \epsilon_0 \ldots \epsilon_{d-1}, \bb^{(1)})$

\begin{enumerate}
	\item For i = $1, \ldots, d-1$
		\begin{enumerate}
			\item let $k_{i}$ be the smallest odd integer greater than or equal to $\log_{\alpha_{i}/2} (2/\epsilon_i)$.
			\item $\xx^{(i)}_{F_i} \leftarrow \ZZ_{F_{i} F_{i}}^{(k_i)} \bb^{(i)}_{F_i}$, 
   where $\ZZ_{F_{i} F_{i}}^{(k_i)}$ is obtained from $\MM_{F_{i} F_{i}}^{(i)}$ as in \eqref{eqn:defZ}.
			\item $\bb^{(i+1)} \leftarrow \bb^{(i)}_{C_i} - \MM^{(i)}_{C_i F_i} \xx^{(i)}_{F_i}$.
		\end{enumerate}
	\item $\xx^{(d)} \leftarrow \left( \MM^{(d)} \right)^{-1} \bb^{(d)}$.
	\item For i = $d-1, \ldots, 1$
		\begin{enumerate}
			\item $\xx^{(i)}_{C_i} \leftarrow \xx^{(i+1)}$.
			\item $\xx^{(i)}_{F_i} \leftarrow \xx^{(i)}_{F_i} -  \ZZ_{F_{i} F_{i}}^{(k_i)} \MM^{(i)}_{F_i C_i} \xx^{(i+1)}$.
		\end{enumerate}
	\end{enumerate}
\end{algbox}

\caption{Solver Algorithm using Vertex Sparsifier Chain}

\label{fig:applyChain}
\end{figure}

\begin{lemma}
\label{lem:apply_chain}
Given a vertex sparsifier chain where $\MM^{(i)}$ has $m_i$ non-zero entries,
the algorithm $\textsc{ApplyChain}(\MM^{(1)}, \ldots, \MM^{(d)}, F_1, \ldots, F_{d-1}, \alpha_{1} \ldots \alpha_{d - 1}, \epsilon_0 \ldots \epsilon_{d-1}, \bb )$ corresponds to a linear operator
$\WW$ acting on $\bb$ such that
\begin{enumerate}
\item \[
\WW^{-1} \approx_{\sum_{i = 0}^{d-1} 2\epsilon_i} \MM^{(0)},
\] and
\item for any vector $\bb$, $\textsc{ApplyChain}(\MM^{(1)}, \ldots, \MM^{(d)}, F_1, \ldots, F_{d-1}, \alpha_{1} \ldots \alpha_{d - 1}, \epsilon_0 \ldots \epsilon_{d-1}, \bb )$ runs in
$O\left(\sum_{i = 1}^{d-1} \left( \log_{\alpha_i} \left( \epsilon_i^{-1}  \right) \log{n} \right) \right)$ depth
and $O\left(\sum_{i = 1}^{d-1} \left( \log_{\alpha_i}\left( \epsilon_i^{-1} \right) \right) m_i \right)$ work.
\end{enumerate}
\end{lemma}

\begin{proof}
We begin by observing that the output vector $\xx^{(1)}$ is a linear transformation
  of the input vector $\bb^{(1)}$.
Let $\WW^{(1)}$ be the matrix that realizes this transformation.
Similarly, for $2 \leq i \leq d$, define $\WW^{(i)}$  to be the matrix so that
\[
  \xx^{(i)} = \WW^{(i)} \bb^{(i)}.
\]
An examination of the algorithm reveals that
\begin{equation}\label{eqn:apply_chain1}
  \WW^{(d)} = \left(\MM^{(d)} \right)^{-1},
\end{equation}
and
\begin{equation}\label{eqn:apply_chaini}
  \WW^{(i)}
 =
  \left[
\begin{array}{cc}
\II & -\ZZ^{(k_{i})}_{F_{i}F_{i}} \MM_{F_{i}C_{i}}\\
0 & \II
\end{array}
\right]
\left[
\begin{array}{cc}
\ZZ^{(k_{i})}_{F_{i}F_{i}} & 0 \\
0 & \WW^{(i+1)}
\end{array}
\right]
\left[
\begin{array}{cc}
\II & 0\\
-\MM_{C_{i}F_{i}} \ZZ^{(k_{i})}_{F_{i}F_{i}}  & \II
\end{array}
\right].
\end{equation}

We will now prove by backwards induction on $i$ that
\[
\left(\WW^{(i)} \right)^{-1} \approx_{\sum_{j = i}^{d-1} 2\epsilon_j} \MM^{(i)}.
\]
The base case of $i = d$ follows from \eqref{eqn:apply_chain1}.
When we substitute our choice of $k_{i}$
  from line $1a$ of \textsc{ApplyChain}
  into Lemma~\ref{lem:subZFF}, we find that
\[
  \left[
\begin{array}{cc}
\II & -\ZZ^{(k_{i})}_{F_{i}F_{i}} \MM^{(i)}_{F_{i}C_{i}}\\
0 & \II
\end{array}
\right]
\left[
\begin{array}{cc}
\ZZ^{(k_{i})}_{F_{i}F_{i}} & 0 \\
0 & \schur{\MM^{(i)}}{F_{i}}^{-1}
\end{array}
\right]
  \left[
\begin{array}{cc}
\II & 0\\
- \MM^{(i)}_{C_{i}F_{i}} \ZZ^{(k_{i})}_{F_{i}F_{i}} & \II
\end{array}
\right]
\approx_{\epsilon_{i}}
\left(\MM^{(i)} \right)^{-1}.
\]
As $\MM^{(i+1)} \approx_{\epsilon_{i}} \schur{\MM^{(i)}}{F_{i}}$,
\[
  \left[
\begin{array}{cc}
\II & -\ZZ^{(k_{i})}_{F_{i}F_{i}} \MM^{(i)}_{F_{i}C_{i}}\\
0 & \II
\end{array}
\right]
\left[
\begin{array}{cc}
\ZZ^{(k_{i})}_{F_{i}F_{i}} & 0 \\
0 & \left(\MM^{(i+1)} \right)^{-1}
\end{array}
\right]
  \left[
\begin{array}{cc}
\II & 0\\
- \MM^{(i)}_{C_{i}F_{i}} \ZZ^{(k_{i})}_{F_{i}F_{i}} & \II
\end{array}
\right]
\approx_{2 \epsilon_{i}}
\left(\MM^{(i)} \right)^{-1}.
\]
By combining this identity with \eqref{eqn:apply_chaini}
  and our inductive hypothesis, we obtain
\[
  \WW^{(i)}
\approx_{\sum_{j = i}^{d-1} 2 \epsilon_{j}}
\left(\MM^{(i)} \right)^{-1}.
\]

Finally, as $\MM^{(0)} \approx_{\epsilon_{0}} \MM^{(1)}$,
\[
  \WW^{(1)}
\approx_{\sum_{j = 0}^{d-1} 2 \epsilon_{j}}
\left(\MM^{(0)} \right)^{-1}.
\]

To bound the work and depth of the algorithm, we observe that we do not need to construct
  the matrices $\ZZ^{(k_{i})}_{F_{i}F_{i}}$ explicitly.
Rather, we multiply vectors by the matrices by performing
$k_{i} = O(\log_{\alpha_i} ( \epsilon_i^{-1}  ) )$ matrix-vector products
  by the submatrices of $\MM^{(i)}$ that appear in the expression \eqref{eqn:defZ}.
As each matrix-vector product can be performed in depth $O (\log n)$,
  the depth of the whole algorithm is bounded by
  $O ((\log n) \sum_{i} k_{i})$.
As each matrix $\MM^{(i)}$ has $m_{i}$ non-zero entries,
  and the work of the $i$th iteration is dominated by the cost of multiplying
  by submatrices of $\MM^{(i)}$ $O (k_{i})$ times, the total work of the algorithm
  is $O (\sum_{i=1}^{d-1} m_{i} k_{i} )$.
\end{proof}

\iffull{
\begin{definition}[Work and Depth of a Vertex Sparsifier chain]
An $\epsilon$-vertex sparsifier chain of an SDDM matrix $\MM^{(0)}$
  of depth $D$ and work $W$ is a vertex sparsifer chain of $\MM^{(0)}$
  with parameters $\alpha_i \geq 4$ and $ 1/2\geq\epsilon>0$
that satisfies
\begin{enumerate}
\item $2 \sum_{i=0}^{d-1} \epsilon_{i} \leq \epsilon$,
\item $\sum_{i = 1}^{d-1} m_{i} \log_{\alpha_i} \epsilon_i^{-1}   \leq W$,
  where $m_{i}$ is the number of nonzeros in $\MM^{(i)}$, and

\item $\sum_{i = 1}^{d-1} (\log n) \log_{\alpha_i}  \epsilon_i^{-1} \leq D$,
  where $n$ is the dimension of $\MM^{(0)}$.
\end{enumerate}
\end{definition}

\begin{theorem}
\label{thm:result_BSS}
Every SDDM matrix $\MM$ of dimension $n$
  has a $1$-vertex sparsifier chain
  of depth $O(\log^{2}n\log\log n)$
  and work $O (n)$.
Given such vertex sparsifier chain, for any vector $b$,
we can compute an $\epsilon$ approximate solution to $\MM^{-1}b$
in $O(m\log(1/\epsilon))$ work and $O(\log^{2}n\log\log n\log(1/\epsilon))$
depth.
\end{theorem}
}
{
\begin{theorem}
\label{thm:result_BSS}
Every SDDM matrix $\MM$ of dimension $n$
  has a $1$-vertex sparsifier chain.
Given such vertex sparsifier chain, for any vector $b$,
we can compute an $\epsilon$ approximate solution to $\MM^{-1}b$
in $O(m\log(1/\epsilon))$ work and $O(\log^{2}n\log\log n\log(1/\epsilon))$
depth.
\end{theorem}
}
\begin{proof}
We will show the existence of such a vertex sparsifier chain with
  $\alpha_i = 4$ for all $i$ and $\epsilon_i = \frac{1}{2(i+2)^2}$.
Lemma~\ref{lem:subsetSimple} tells us that every SDDM matrix
  has a $4$-strongly diagonally dominant subset consisting of at least
  a $1/ 8 (1+4) = 1/40$ fraction of its columns.
By taking such a subset, we ensure that the number of vertices of $\MM^{(i)}$,
  which we define to be $n_{i}$, satisfies
\[
n_i \leq \left( \frac{39}{40} \right)^{i-1} n.
\]
In particular, this means that $d$, the number of matrices in the chain,
  will be logarithmic in $n$.

If we use Theorem~\ref{thm:BSS} to find a matrix $\MM^{(1)}$
  that is an $\epsilon_{0}$ approximation of $\MM^{(0)} = \MM$,
  and to find a matrix $\MM^{(i+1)}$ that is an $\epsilon_{i}$
  approximation of $\schur{\MM^{(i)}}{F_{i}}$, then
  each matrix $\MM^{(i)}$ will have a number of nonzero entries satisfying
\[
m_i \leq 
O (n_{i} / \epsilon_{i-1}^{2})
\leq 
O\left( \left( \frac{39}{40} \right)^{i-1} (i+1)^4 n \right).
\]

Lemma~\ref{lem:apply_chain} tell us that the vertex sparsifier chain
  induces a linear operator that is an $\epsilon$-approximation of the inverse of $\MM$,
  where
\[
  \epsilon \leq 2 \sum_{i=0}^{d-1} \epsilon_{i}
  \leq 2 \sum_{i=0}^{d-1} \frac{1}{2(i+2)^2}
  \leq \sum_{i \geq 2} \frac{1}{i^{2}}
  \leq 1.
\]

To compute the work and depth of the chain, 
  recall that we set $k_{i}$ to be the 
  smallest odd integer that is at least 
  $\log_{\alpha_{i}/2} \epsilon_i^{-1}$,
  so $k_{i} \leq O (\log i)$.
Thus, the work of the chain is at most
\[
  \sum_{i = 1}^{d} k_{i} m_{i}  
\leq 
  O \left(\sum_{i=1}^{d} \log (i) \left( \frac{39}{40} \right)^{i-1} (i+1)^4 n \right)
\leq 
  O \left(\sum_{i=1}^{d} \left( \frac{39}{40} \right)^{i-1} i^5 n \right)
\leq 
  O (n).
\]
Similarly, the depth of the chain is at most
\[
  \sum_{i = 1}^{d} (\log n) k_{i}
\leq 
  O\left(  \sum_{i = 1}^{d} (\log n) \log d \right)
\leq 
  O (\log^{2} n \  \log \log n)  .
\]
\end{proof}

\dan{This proof should tie in better with the definition.}

\section{Linear sized $U^{T} D U$ approximations}\label{sec:existence}

We now show that the vertex sparsifier chains of $\MM$ from the previous section
  can be used to construct Cholesky factorizations of matrices that 
  are 2-approximations of $\MM$.
In particular, we prove that for every SDDM matrix $\MM$ of dimension $n$ there
  exists a diagonal matrix $\DD$ and an upper-triangular matrix $\UU$
  having $O (n)$ nonzero entries such that $\UU^{T} \DD \UU$
  is a 2-approximation of $\MM$.

The obstacle to obtaining such a factorization is that it does not allow
  us to multiply a vector by $\ZZ^{(k_{i})}_{F_{i}F_{i}}$ in many steps.
Rather, we must explicitly construct the matrices 
  $\ZZ^{(k_{i})}_{F_{i}F_{i}}$.
If we directly apply the construction suggested in the previous section,   
  these matrices could be dense and thereby result in a matrix $\UU$
  with too many nonzero entries.
To get around this problem, we show that we can always find strongly
  diagonally dominant subsets in which all the vertices have low degree.
This will ensure that all of the matrices $\ZZ^{(k_{i})}_{F_{i}F_{i}}$
  are sparse.

\begin{lemma}
\label{lem:subsetLowDeg}
For every $n$-dimensional SDD matrix $\MM$ and every $\alpha \geq 0$, 
  there is an $\alpha$-strongly diagonally dominant subset of columns $F$
  of size at least $\frac{n}{16(1+\alpha)}$
  such that the number of nonzeros in every column $F$ is at most twice the average
  number of nonzeros in columns of $\MM$.
\end{lemma}

\begin{proof}
Discard every column of $\MM$ that has more than twice the average number of nonzeros
  per column.
Then remove the corresponding rows.
The remaining matrix has dimension at least $n/2$.
Use Lemma~\ref{lem:subsetSimple}
  to find an $\alpha$-strongly diagonally subset of the columns of 
  this matrix.
\end{proof}

To obtain a $\UU^{T} \DD \UU$ factorization from a vertex sparsifier chain,
  we employ the procedure in Figure~\ref{fig:UDU}.

\begin{figure}[h]
\begin{algbox}
$(\DD, \UU) = \textsc{Decompose}\left(\MM^{(1)}, \dots , \MM^{(d)}, F_{1}, \dots , F_{d-1}\right)$,
 where each $\MM^{(i)}$ is a SDDM matrix.
\begin{enumerate}
\item let $k_{i}$ be the smallest odd integer greater than or equal to $\log_{\alpha_{i}/2} \epsilon_i^{-1}$.

\item For each $i < d$, write $\MM^{(i)} = \XX^{(i)} + \LL^{(i)}$
  where $\XX^{(i)}$ is a positive diagonal matrix and $\LL^{(i)}$ is a Laplacian.

\item Let $\XX^{(d)} = \II_{C_{d-1}}$
  and let $\UUhat$ be the upper-triangular Cholesky factor of $\MM^{(d)}$.

\item Let $\DD$ be the diagonal matrix with $\DD_{F_{i} F_{i}} = \XX_{i}$,
  for $1 \leq i < d$, and $\DD_{C_{d-1} C_{d-1}} = \II_{C_{d-1}}$.

\item Let $\UU$ be the upper-triangular matrix with $1s$ on the diagonal,
  $\UU_{C_{d-1} C_{d-1}} =  \UUhat$, and
  $\UU_{F_{i} C_{i}} = \ZZ_{F_{i} F_{i}}^{(k_{i})} \MM^{(i)}_{F_{i} C_{i}}$,
  for $1 \leq i < d$.
\end{enumerate}
\end{algbox}

\caption{Converting a vertex sparsifer chain into $\UU$ and $\DD$.}
\label{fig:UDU}
\end{figure}

\begin{lemma}\label{lem:uduApprox}
On input a vertex sparsifier chain of $\MM$
  with parameters $\alpha_{i} \geq 4$ and $\epsilon_{i}>0$,
  the algorithm \textsc{Decompose} produces matrices $\DD$ and $\UU$
  such that
\[
  \UU^{T} \DD \UU \approx_{\gamma} \MM ,
\]
where
\[
  \gamma \leq 2 \sum_{i=0}^{d-1} \epsilon_{i} + 4 / \min_{i} \alpha_{i}.
\]
\end{lemma}


\begin{proof}
Consider the inverse of the operator $\WW = \WW^{(1)}$ realized by the algorithm
  \textsc{ApplyChain},
  and the operators $\WW^{(i)}$ that appear in the proof of Lemma~\ref{lem:apply_chain}.

We have
\[
 \left(  \WW^{(i)} \right)^{-1}
 =
\left[
\begin{array}{cc}
\II & 0\\
\MM_{C_{i}F_{i}} \ZZ^{(k_{i})}_{F_{i}F_{i}}  & \II
\end{array}
\right]
\left[
\begin{array}{cc}
\left(\ZZ^{(k_{i})}_{F_{i}F_{i}} \right)^{-1} & 0 \\
0 & \left(\WW^{(i+1)} \right)^{-1}
\end{array}
\right]
  \left[
\begin{array}{cc}
\II & \ZZ^{(k_{i})}_{F_{i}F_{i}} \MM_{F_{i}C_{i}}\\
0 & \II
\end{array}
\right],
\]
and
\[
 \left(  \WW^{(d)} \right)^{-1}
  =
 \MM^{(d)}
 = \UUhat^{T}  \UUhat.
\]
After expanding and multiplying the matrices in this recursive factorization,
  we obtain
\[
 \left( \WW^{(1)} \right)^{-1}
= 
  \UU^{T}
\begin{bmatrix}
\left(\ZZ^{(k_{1})}_{F_{1}F_{1}} \right)^{-1} &  \dots & 0 & 0\\
0 & \ddots   & 0 & 0\\
0 &  \ldots & \left(\ZZ^{(k_{d-1})}_{F_{d-1}F_{d-1}} \right)^{-1} & 0 \\
0 &  \ldots & 0 & \II_{C_{d-1} C_{d-1}} 
\end{bmatrix}
  \UU.
\]
Moreover, we know that this latter matrix is a $2 \sum_{i=0}^{d-1} \epsilon_{i}$
  approximation of $\MM$.
It remains to determine the impact of replacing the matrix in the middle of
  this expression with $\DD$.

It suffices to examine how well each matrix 
 $\left(\ZZ^{(k_{i})}_{F_{i}F_{i}} \right)^{-1}$
  is approximated by $\XX^{(i)}$.
From Lemma~\ref{lem:Xdom} we know that 
\[
\XX^{(i)} \pgeq (\alpha_{i} /2) \LL^{(i)}.
\]
Thus, we may use Lemma~\ref{lem:lightBlock2} with $\beta = \alpha_{i} /2$
  to conclude that
\[
\XX^{(i)}
\approx_{4 / \alpha_{i}}
\left(\ZZ^{(k_{i})}_{F_{i}F_{i}} \right)^{-1}.
\]
This implies that replacing each of the matrices
  $\left(\ZZ^{(k_{i})}_{F_{i}F_{i}} \right)^{-1}$
  by $\XX^{(i)}$
  increases the approximation factor by at most
  $4 / \min_{i} \alpha_{i}$.
\end{proof}

Using this decomposition procedure in a way similar to
Theorem~\ref{thm:result_BSS}, but with subsets chosen
using Lemma~\ref{lem:subsetLowDeg} gives the linear
sized decomposition.

\begin{theorem}\label{thm:UDU}
For every $n$-dimensional SDDM matrix $\MM$ there exists a diagonal matrix $\DD$ 
  and an upper triangular matrix $\UU$ with $O (n)$ nonzero entries so that
\[
  \UU^{T} \DD \UU \approx_{2} \MM .
\]
Moreover, back and forward solves in $\UU$ can be performed 
  with linear work in depth $O (\log^{2} n)$.
\end{theorem}
\begin{proof}
We choose the same parameters as were used in the proof of
  Theorem \ref{thm:result_BSS}: $\alpha_{i} = 4$ for all $i$
  and $\epsilon_{i} = 1/2 (i+2)^{2}$.
Theorem~\ref{thm:BSS} then guarantees that the average number of
  nonzero entries in each column of $\MM^{(i)}$ is at most
  $10 / \epsilon_{i}^{2} = 40 (i+1)^{4}$.
If we now apply Lemma~\ref{lem:subsetLowDeg} to find $4$-diagonally dominant
  subsets $F_{i}$ of each $\MM^{(i)}$,
  we find that each such subset contains at least a $1/80$ fraction of the columns
  of its matrix and that each
  column and row of $\MM^{(i)}$ indexed by $F$ has at most $80 (i+1)^{4}$
  nonzero entries.
This implies that each row of
  $\ZZ^{(k_{i})}_{F_{i}F_{i}} \MM_{F_{i}C_{i}}$
  has at most
  $(80 (i+1)^{4})^{k_{i}+1}$ nonzero entries.

Let $n_{i}$ denote the dimension of $\MM^{(i)}$.
By induction, we know that
\[
  n_{i} \leq n \left(1 - \frac{1}{80} \right)^{i-1}.
\]
So, the total number of nonzero entries in $\UU$ is at most
\[
  \sum_{i=1}^{d} n_{i} (80 (i+1)^{4})^{k_{i}+1}
\leq 
  n \sum_{i=1}^{d} \left(1 - \frac{1}{80} \right)^{i-1} (80 (i+1)^{4})^{k_{i}+1}.
\]
We will show that the term multiplying $n$ in this later expression is upper bounded by
  a constant.
To see this, note that $k_{i} \leq 1 + \log_{\alpha_i / 2} (2 \epsilon_{i}^{-1} )\leq \nu \log (i+1)$
  for some constant $\nu$.
So, there is some other constant $\mu$ for which
\[
(80 (i+1)^{4})^{k_{i}+1}
\leq 
\exp (\mu \log^{2} (i+1)).
\]
This implies that the sum is at most
\[
  \sum_{i \geq 1} \exp (\mu \log^{2} (i+1) - i / 80),
\]
which is bounded by a constant.

The claimed bound on the work to perform backwards and forwards substitution with $\UU$
  is standard: these operations require work linear in the number of nonzero entries of $\UU$.
The bound on the depth follows from the fact that the substitions can be performed blockwise,
  take depth $O (\log n)$ for each block, and the number of blocks, $d$, is logarithmic in $n$.
\end{proof}

\dan{There are a few issues in the writing below this point that we should eventually fix, but maybe not for this version:
\begin{itemize}
\item [1.] Sometimes we talk about Laplacians of graphs, sometimes SDDM matrices, and sometimes SDD.  We should just do SDDM.
\item [2.] We sometimes refer to vertices of a graph, columns and rows of a matrix, or variables.  We should fix one (or one and a half).
  I personally dislike ``variables'', as it can refer to many other things.  Or, we could use ``coordinates''.
\end{itemize}}

\iffull{

\section{Existence of Linear Work and $O(\log n\log^2\log n)$ depth Solvers}\label{sec:depth}

The factorizations constructed in the previous section
can be evaluated in $O(\log^2{n})$ depth  and $O(n)$ work.
One $O(\log n)$ factor comes from the
depth of the recursion and another $O(\log n)$ factor comes from
  the depth of matrix vector multiplication.
The reason that matrix-vector multiplication can take logarithmic depth
  is that computing the sum of $k$ numbers takes $O(\log{k})$ depth.
Thus, if we can instead multiply by matrices with $k^{O (1)}$ nonzeros in each
  row and column, for some small $k$,
  we can reduce the depth of each matrix-vector multiplication to $O (\log k)$.

Although the number of non-zeros in each row of $\ZZ^{(k_{i})}_{F_{i}F_{i}} \MM_{F_{i}C_{i}}$
is bounded by $(80 (i+1)^{4})^{k_{i}+1}$, the number of non-zeros per column
can be high.
This is because although we picked $F_i$ to be of bounded degree, many of those vertices
can be adjacent to a few vertices in $C_i$.
For the factorization constructed in Section~\ref{sec:existence},
$k$ can be as large as $n$.
In this section, we reduce this degree to $\log^{O(1)}{n}$
by splitting high degree vertices.
This leads a factorization that can be evaluated in linear work
and $O(\log n\log^{2} \log n)$ depth.

\subsection{Splitting High Degree Vertices}

While sparsification produces graphs with few edges, it does not guarantee
  that every vertex has low degree.
We will approximate an arbitrary graph by one of bounded degree
  by splitting each high degree vertex into many vertices.
The edges that were originally attached to that vertex will be partitioned
  among the vertices into which it is split.
The vertices into which it is split will then be connected by a complete graph,
  or an expander if the complete graph would have too high degree.
The resulting bounded-degree graph has more vertices.
To approximate the original graph, we take a Schur complement of the bounded-degree graph
   with respect to the extra vertices.
We recall that one can solve a system of equations in a Schur complement of a matrix
  by solving one equation in the original matrix\footnote{%
To solve a system  $\schur{\MM}{S} \xx = \bb$, where $S$ is the last set of rows of $\MM$,
  one need merely solve the system $\MM \xxhat = \bbhat$, where $\bbhat$ is the same as $\bb$
  but has zeros appended for the coordinates in $S$.
The vector $\xx$ is then obtained by simply ignoring the coordinates of $\xxhat$ in $S$.
}

We begin our analysis by examining what happens when we split an individual vertex.

\begin{lemma}\label{lem:splitStar}
Let $G$ be a weighted star graph with vertex set
  $\setof{v_{1}, \dots , v_{n}, u}$
  and edges connecting $u$ to each $v_{i}$  with weight $w_{i}$.
Let $\Ghat$ be a graph with vertex set
  $\setof{v_{1}, \dots , v_{n}, u_{1}, \dots , u_{k}}$
  in which the vertices $\setof{u_{1}, \dots , u_{k}}$
  are connected by a complete graph of edges of weight
  $W = \delta^{-1} \sum_{i} w_{i}$,
  and each vertex $v_{i}$ is connected to exactly one vertex $u_{j}$,
  again by an edge of weight $w_{i}$.
Let $U = \setof{u_{2}, \dots , u_{k}}$.
Then, $\schur{\Ghat}{U} \pleq  G$,
  and in $\schur{\Ghat}{U}$ the edge 
  between $u_{1}$ and $v_{i}$
  has weight at least $w_{i} (1-2 \delta)$, for every $i$.
\end{lemma}
\begin{proof}
We will examine the Laplacian matrices of $G$ and $\Ghat$.
Define $w_{tot} = \sum_{i} w_{i}$, so $W = w_{tot} / \delta$.
Let $\bb$ be the vector of weights $w_{1}, \dots , w_{n}$,
  and let $\BB$ be the diagonal matrix of $\bb$, so that 
  so that 
\[
\LL_{G} = \begin{pmatrix}
\BB  & -\bb \\
-\bb^{T} & w_{tot}
\end{pmatrix}.
\]
Similarly, let $\CC$ be the adjacency matrix between 
  $v_{1}, \dots , v_{n}$ and $u_{1}, \dots , u_{k}$,
  and let $\DD$ be the diagonal matrix whose
  $j$th entry is the sum of the $w_{i}$ for which
  $v_{i}$ is connected to $u_{j}$.
Then,
\[
\LL_{\Ghat} = 
\begin{pmatrix}
\BB & -\CC \\
-\CC^{T} & \DD + W (k \II_{k} - \JJ_{k}), 
\end{pmatrix}
\]
where $\JJ_{k}$ is the $k \times k$ all ones matrix and $k \II_{k} -\JJ_{k}$
  is the Laplacian of the complete graph on $k$ vertices.

To express the Schur complement,
 let $\DD_{2}$ be the submatrix of $\DD$ obtained by excluding its
  first row and column,
 let $\CC_{2}$ be the submatrix of $\CC$ excluding its first column,
 and let $\cc_{1}$ be the first column of $\CC$.
Let $\cc_{2} = \CC_{2} \bvec{1}$, so $\bb = \cc_{1} + \cc_{2}$.
We then have that $  \schur{\LL_{\Ghat}}{U}$ equals
\begin{align*}
& \begin{pmatrix}
\BB & -\cc_{1} \\
-\cc_{1}^{T} & \DD (1,1) 
\end{pmatrix}
-
\begin{pmatrix}
-\CC_{2}^{T} \\
- W  \bvec{1}^{T}
\end{pmatrix}
\left( \DD_{2} + W (k \II_{k-1} - \JJ_{k-1}) \right)^{-1}
\begin{pmatrix}
-\CC_{2} &
- W  \bvec{1}
\end{pmatrix}
\\
& =
\begin{pmatrix}
\BB & -\cc_{1} \\
-\cc_{1}^{T} & \DD (1,1) 
\end{pmatrix}
-
\begin{pmatrix}
\CC_{2}^{T} \\
W  \bvec{1}^{T}
\end{pmatrix}
\left( \DD_{2} + W (k \II_{k-1} - \JJ_{k-1}) \right)^{-1}
\begin{pmatrix}
\CC_{2} &
W  \bvec{1}
\end{pmatrix}
\end{align*}
To understand this expression, 
  we will show that it approaches $\LL_{G}$ as $\delta $ goes to zero.
We first note that
\[
  (k \II_{k-1} - \JJ_{k-1})^{-1} = \frac{1}{k} (\II_{k-1} + \JJ_{k-1}),
\]
and so
\[
\CC_{2}^{T}  (k \II_{k-1} - \JJ_{k-1})^{-1}  \bvec{1} = \CC_{2}^{T} \bvec{1} = \cc_{2}^{T}.
\]
So, the last row and column of the Schur complement agrees with $\LL_{G}$ as $\delta$ goes to zero.
On the other hand, the upper-left block becomes
\[
\BB -  \CC_{2}^{T} \left( \DD_{2} + W (k \II_{k-1} - \JJ_{k-1}) \right)^{-1} \CC_{2},
\]
which goes to $\BB$ as $\delta $ goes to zero.

To bound the discrepancy in terms of $\delta $, we recall that $\JJ = \bvec{1} \bvec{1}^{T}$,
  and so we can use the Sherman-Morrison-Woodbury formula to 
  compute
\[
\left( \DD_{2} + W (k \II_{k-1} - \JJ_{k-1}) \right)^{-1}
=
\left( \DD_{2} + W k \II_{k-1} \right)^{-1}
+
\frac{\left( \DD_{2} + W k \II_{k-1} \right)^{-1} W \JJ_{k-1} \left( \DD_{2} + W k \II_{k-1} \right)^{-1}}
{1 - W \bvec{1}^{T} \left( \DD_{2} + W k \II_{k-1} \right)^{-1} \bvec{1}}.
\]

Note that $\DD_{2} + W k \II_{k-1}$ is a diagonal matrix.
As all entries of $\DD_{2}$ are less than $\delta W$, 
  every diagonal entry of this matrix is at least $(W k (1+\delta))^{-1}$.
So, we have the entry-wise inequality
\[
\frac{\left( \DD_{2} + W k \II_{k-1} \right)^{-1} W \JJ_{k-1} \left( \DD_{2} + W k \II_{k-1} \right)^{-1}}
{1 - W \bvec{1}^{T} \left( \DD_{2} + W k \II_{k-1} \right)^{-1} \bvec{1}}
\geq 
\frac{
W \JJ_{k-1}/  (W k (1+\delta) )^{2}
}{1/k}
=
\frac{1}{W k (1+\delta)^{2}} \JJ_{k-1}.
\]
This tells us that, entry-wise,
\[
\left( \DD_{2} + W (k \II_{k-1} - \JJ_{k-1}) \right)^{-1}
\geq 
(1 - 2 \delta)  \frac{1}{W k}
\left( \II_{k-1} + \JJ_{k-1}\right)
=
(1 - 2 \delta) 
(W ( k \II_{k-1} - \JJ_{k-1} ))^{-1}.
\]
The claimed bound on the entries in row and column corresponding to $u_{1}$ of the Schur complement 
  now follows
  from the fact that they are obtained by multiplying this matrix inverse on either side by
  $\CC_{2}$ and $W \bvec{1}$:
  as these are non-negative matrices, the entry-wise inequality propogates to the product.
\end{proof}

The following theorem states the approximation we obtain if we split all the vertices of high degree
  and connect the clones of each vertex by expanders.

\begin{theorem}
\label{thm:degreeReduction}
For any graph $G=(V,E)$ with $n$ vertices, 
  $\varepsilon > 0$ and $t > 1/\varepsilon^{2}$,
there is a graph $\widetilde{G}=(V\cup S,\widetilde{E})$
 of maximum degree $O (t)$
  such that 
\begin{equation}\label{eqn:thm:degreeReduction}
G \approx_{\varepsilon} \schur{\widetilde{G}}{S},
\end{equation}
 $\sizeof{S} = O\left(n/(\varepsilon^{2}t)\right)$,
and $\sizeof{\widetilde{E}} \leq O(n/\varepsilon^{2} + n/(\varepsilon^{4} t))$.
\end{theorem}

\begin{proof}
We first sparsify $G$ using Theorem~\ref{thm:BSS},
obtaining $\Ghat $ with $O(n / \epsilon^{2})$ edges
such that $\Ghat  \approx_{\varepsilon / 3} G$.

Let $U$ be the set of vertices in $\Ghat$ of degree more than $t$.
We will split each vertex in $U$ into many vertices.
For each $u \in U$, let $d_{u}$ be its degree in $\Ghat$.
We split $u$ into $\ceil{d_{u} / t}$ vertices, one of which we identify
  with the original vertex $u$, and the rest of which we put in $S$.
We then partition the edges that were attached to $u$ among these
  $\ceil{d_{u} / t}$ vertices, so that each is now attached to at most $t$ 
  of these edges.
We then place a complete graph between all of the
  vertices derived from $u$ in which every edge has weight equal to
  the sum of the weights of edges attached to $u$, times $12/\varepsilon$.
That is, we apply the construction of Lemma~\ref{lem:splitStar} with
  $\delta = \varepsilon / 3$.
Call the resulting graph $G'$.

If $\ceil{d_{u} / t} > t$, we replace that complete graph by
  a weighted expander of degree $O (1/\varepsilon^{2})$ that is an $\varepsilon/3$
  approximation of this weighted complete graph, as guaranteed to exist
  by Lemma \ref{lem:existExpanders}.
The resulting graph is $\Gtil$.

To show that \eqref{eqn:thm:degreeReduction} holds, 
  we first show that
\[
  \Ghat \approx_{\varepsilon/3} \schur{G'}{S}.
\]
Lemma \ref{lem:splitStar} tells us that
  $\schur{G'}{S} \pleq   \Ghat.$
It also tells us that the graph looks
  like $\Ghat$ except that
  it can have some extra edges and that the edges attached to vertices we split can
  have a slightly lower weight.
If an edge is attached to just one of the split vertices, its weight can be lower by a factor
  of $2 \delta  = \varepsilon / 6$.
However, some edges could be attached to two of the split vertices, in which case they could
  have weight that is lower by a factor of $\varepsilon/3$.
This implies that
  $(1-\varepsilon/3) \Ghat \pleq  \schur{G'}{S}$.
To prove \eqref{eqn:thm:degreeReduction}, we now combine this with the factors
  of $\varepsilon/3$ that we loose by sparsifying at the start
  and by replacing with expanders at the end.

It is clear that every vertex in $\Gtil$ has degree at most $t + O (1/\epsilon^{2})$.
To bound the number of edges in $\Gtil$, we observe that the sum of the degrees
  of vertices that are split is at most $O (n / \varepsilon^{2})$, and so the number
  of extra vertices in $S$ is at most $O (n / \varepsilon^{2} t)$.
Our process of adding expanders at the end can create at most $O (1 / \varepsilon^{2})$ new edges
  for each of these vertices,
  giving a total of at most $O (n / \varepsilon^{4} t)$ new edges.
\end{proof}

\begin{remark}\label{rem:thm:degreeReduction}
We do not presently know how to implement the exact construction from the above theorem in polynomial time,
  because it relies on the nonconstructive proof of the existance of expanders from \cite{IF4}.
One can transform this into a polynomial time construction by instead using the explicit constructions
  of Ramanujan graphs \cite{Margulis88,LPS} as described in Lemma \ref{lem:explicitExpanders}.
This would, however, add the requirement $t > 1/\epsilon^{6}$ to Theorem~\ref{thm:degreeReduction}.
While this would make Theorem~\ref{thm:degreeReduction} less appealing, it does not
 alter the statement of Theorem \ref{thm:lowerDepth}.
\end{remark}

It remains to incorporate this degree reduction routine into
 the solver construction.
 Since our goal is to upper-bound the degree by $O(\log^{c}{n})$ for some constant $c$, we
 can pick $t$ in Theorem~\ref{thm:degreeReduction}
 so that $\epsilon^2 t \leq \log^{O (1)}n$.
This leads to a negligible increase in vertex count at each step.
So we can use a construction similar to Theorem~\ref{thm:UDU}
to obtain the lower depth solver algorithm.

\begin{theorem}\label{thm:lowerDepth}
For every $n$-dimensional SDDM matrix $\MM$ there a linear operator
$\ZZ$ such that
\[
\ZZ \approx_{2} \MM^{-1}
\]
and matrix-vector multiplications in $\ZZ$ can be done in
 linear work and $O(\log{n} \log^{2} \log{n})$ depth.
 Furthermore, this operator can be obtained via a diagonal
$\DD$, an upper triangular matrix $\UU$ with $O(n)$  non-zero entries 
  and a set of vertices $\widehat{V}$ such that
\[
\MM\approx_{2}\schur{\UU^{T}\DD\UU}{\widehat{V}}.
\]
\end{theorem}

\begin{proof}
We will slightly modify the vertex sparsification chain from
Definition~\ref{def:vertexSparsifierChain}.
Once again, we utilize $\alpha_{i} = 4$ for all $i$
  and $\epsilon_{i} = 1/2 (i+2)^{2}$.
The main difference is that instead of using spectral sparsifiers
from Theorem~\ref{thm:BSS} directly, we use Theorem~\ref{thm:degreeReduction}
to control the degrees.
Specifically we invoke it with  $\epsilon = \epsilon_i$
and $t_i = 200 \epsilon_i^{-2}$ on $\schur{\MM^{(i)}}{F_i}$
to obtain $\MM^{(i + 1)}$ and $S_{i + 1}$ s.t.
\[
\schur{\MM^{(i)}}{F_i} \approx_{\varepsilon_{i}}\schur{\MM^{(i + 1)}}{S_{i + 1}}.
\]

This leads to a slightly modified version of the vertex sparsifier chain.
We obtain a sequence of matrices $\MM_{1}, \MM_{2} \ldots$and subsets $S_i$
and $F_i$ s.t.
\begin{enumerate}
\item [a.] $\MM^{(1)} \approx_{\epsilon_0} \MM^{(0)}$,
\item [b.] $\schur{\MM^{(i + 1)}}{S_{i + 1}} \approx_{\epsilon_i} \schur{\MM^{(i)}}{F_i}$,
\item [c.] Each row and column of $\MM^{(i)}$ has at most $t$ non-zeros.
\item [d.] Each column and row of $\MM^{(i)}$ indexed by $F_i$ has at most $80 (i+1)^{4}$
  nonzero entries. (obtained by combining the bound on non-zeros from
  Theorem~\ref{thm:degreeReduction} with Lemma~\ref{lem:subsetLowDeg}.
\item [e.] $\MM^{(i)}_{F_i F_i}$ is $4$-strongly diagonally dominant and
\item [f.] $\MM^{(d)}$ has size $O(1)$.
\end{enumerate}

This modified chain can be invoked in a way analogous to the vertex
sparsifier chain.
At each step we
\begin{enumerate}
\item Apply a recursively computed approximation to $(\MM^{(i)}_{F_i F_i})^{-1}$ on $\bb^{(i)}$
to obtain $\bbbar^{(i + 1)}$.
\item Pad $\bbbar^{(i + 1)}$ with zeros on $S_i$ to obtain $\bb^{(i + 1)}$
\item Repeat on level $i + 1$ 
\item Restrict the solution $\xx^{(i + 1)}$ to obtain $\xxbar^{(i + 1)}$
\item Apply a recursively computed approximation to $(\MM^{(i)}_{F_i F_i})^{-1}$ to
$\xxbar^{(i + 1)}$ to obtain $\xx^{(i)}$.
\end{enumerate}
\richard{Attempt at a pseudocode, steps 1 and 5 are identical to Schur complementing}

Let $n_{i}$ denote the dimension of $\MM^{(i)}$.
Since $t$ was set to $200 \epsilon_i^{-2}$, the increase in vertex size given by
$S^{(i)}$ is at most:
\[
n_{i + 1} \leq n_i \left(1 - \frac{1}{80} \right) \left( 1 + \frac{1}{\epsilon^2 t} \right)
\leq n_i \left(1 - \frac{1}{400} \right)
\]
By induction this gives
\[
  n_{i} \leq n \left(1 - \frac{1}{400} \right)^{i-1}.
\]
So the total work follows in a way analogous to Theorem~\ref{thm:result_BSS},
and it remains to bound depth.

The constant factor reduction in vertex count gives a bound
on chain length of $d = O(\log{n})$.
This in turn implies $t = O(\epsilon_i^{-2})= O(\log^{4}{n})$.
Therefore the depth of each matrix-vector multiplication by $\MM_{F_i C_i}$
is bounded by $O(\log\log{n})$.
Also, choosing $k_{i}$ as in Theorem~\ref{thm:UDU} gives that the number of non-zeros
in $\ZZ_{F_iF_i}^{(k_i)}$ is bounded by $(\log{n})^{O(\log\log{n})}$,
giving a depth of $O(\log^{2}\log{n})$ for each matrix-vector multiplication
involving $\ZZ_{F_iF_i}$.
The $O(\log{n})$ bound on $d$ then gives a bound on the total depth of $O (\log n \log^2 \log n)$.

This algorithm can also be viewed as a linear operator corresponding to a
$\UU^T \DD \UU$ factorization of a larger matrix.
We will construct the operators inductively.
Suppose we have $\DDhat^{(i + 1}$, $\UU^{(i + 1)}$, and $\widehat{V}^{(i + 1)}$ such that
\[
\MM^{(i + 1)} \approx_{2 \sum_{i' = i + 1}^{d} \epsilon_{i'}} \schur{\left( \UU^{(i + 1)} \right)^{T} \DDhat^{(i + 1)} \UU^{(i + 1)} }{\widehat{V}^{(i + 1)}}.
\]
An argument similar to that in the proof of Lemma \ref{lem:uduApprox} gives
\[
\MM^{(i)}\approx_{\varepsilon_{i}}\left[\begin{array}{cc}
\II & 0\\
\UU_{F_{i} C_{i}}^{T} & \II
\end{array}\right]\left[\begin{array}{cc}
\left(\ZZ^{(k_{1})}_{F_{1}F_{1}} \right)^{-1}  & 0\\
0 & \schur{\MM^{(i)}}{F_i}
\end{array}\right]\left[\begin{array}{cc}
\II & \UU_{F_{i} C_{i}}\\
0 & \II
\end{array}\right].
\]
Consider the entry $\schur{\MM^{(i)}}{F_i}$.
Combining condition b of the chain with the inductive
hypothesis and Fact~\ref{fact:schurLoewner} gives
\begin{align*}
\schur{\MM^{(i)}}{F_i}
& \approx_{\epsilon_i} \schur{\MM^{(i + 1)}}{S_{i + 1}}\\
& \approx_{\epsilon_i + \sum_{i' = i + 1}^{d} \epsilon_i} \schur{
 \schur{\left( \UU^{(i + 1)} \right)^{T} \DDhat^{(i + 1)} \UU^{(i + 1)} }{\widehat{V}^{(i + 1)}}}{S_{i + 1}}.
\end{align*}

Since the order by which we remove vertices when taking Schur complements
does not matter, we can set
\[
\widehat{V}^{(i)} = \widehat{V}^{(i + 1)} \union S_{i + 1},
\]
to obtain
\[
\schur{\MM^{(i)}}{F_i}
\approx_{\epsilon_i + \sum_{i' = i + 1}^{d} \epsilon_i}
 \schur{\left( \UU^{(i + 1)} \right)^{T} \DDhat^{(i + 1)} \UU^{(i + 1)} }
 {\widehat{V}^{(i)}}.
\]
Block-substituting this and using Fact~\ref{fact:blockSubstitute} then gives:
\[
\MM^{(i)}\approx_{2 \sum_{i' = i}^{d} \varepsilon_{i'}}\left[\begin{array}{cc}
\II & 0\\
\UU_{F_{i} C_{i}}^{T} & \II
\end{array}\right]\left[\begin{array}{cc}
\left(\ZZ^{(k_{1})}_{F_{1}F_{1}} \right)^{-1} & 0\\
0 & \schur{\left( \UU^{(i + 1)} \right)^{T} \DDhat^{(i + 1)} \UU^{(i + 1)}  }{\widehat{V}^{(i)}}
\end{array}\right]\left[\begin{array}{cc}
\II & \UU_{F_{i} C_{i}}\\
0 & \II
\end{array}\right].
\]
We will show in  Lemma~\ref{lem:schurRearrange} that the Schur
complement operation can be taken outside multiplications by $\UU^T$
and $\UU$.
This allows us to rearrange the right-hand side into:
\[
\schur{
\left[\begin{array}{ccc}
\II & 0 & 0\\
\UU_{F_{i} C_{i}}^{T} & \II & 0\\
0 & 0 & \II_{\widehat{V}^{(i)}}
\end{array}\right]
\left[\begin{array}{cc}
\left(\ZZ^{(k_{1})}_{F_{1}F_{1}} \right)^{-1}  & 0\\
0 & (\UU^{(i + 1)})^{T} \DDhat^{(i + 1)} \UU^{(i + 1)}
\end{array}\right]
\left[\begin{array}{ccc}
 \II & \UU_{F_{i} C_{i}} & 0\\
0 & \II & 0\\
0 & 0 & \II_{\widehat{V}^{(i)}}  \\ 
\end{array}\right]
}
{\widehat{V}^{(i)}}.
\]

Hence choosing
\[
\DDhat^{(i)} = 
\left[\begin{array}{cc}
\left(\ZZ^{(k_{1})}_{F_{1}F_{1}} \right)^{-1}  & 0\\
0 & \DDhat^{(i + 1)}
\end{array}\right],
\]
and
\[
\UU^{(i)} = 
\left[\begin{array}{cc}
\II & 0\\
0 & \UU^{(i + 1)}
\end{array}\right]
\left[\begin{array}{ccc}
\II & \UU_{F_{i} C_{i}} & 0\\
0 & \II & 0\\
0 & 0 & \II_{\widehat{V}^{(i)}}
\end{array}\right]
=
\begin{bmatrix}
 \II & \UU_{F_{i} C_{i}} & 0\\
0 & \II & 0\\
0 & 0 & \UU^{(i+1)}
\end{bmatrix},
\]
gives $\MM^{(i)} \approx_{2 \sum_{i' = i}^{d} \epsilon_{i'} } \schur{\left( \UU^{(i)} \right)^{T} \DD^{(i)} \UU^{(i)} }{\widehat{V}^{(i)}}$, and the inductive hypothesis holds for $i$ as well.

We then finish the proof as in Lemma \ref{lem:uduApprox}  by replacing $\DDhat^{(0)}$ with
 a matrix $\DD$ whose diagonals contain $\XX^{(i)}$ instead of $\left(\ZZ^{(k_{1})}_{F_{1}F_{1}} \right)^{-1}$.

\end{proof}

It remains to show the needed Lemma rearanging the order of taking Schur complements.

\begin{lemma}
\label{lem:schurRearrange}
Let $\PP$ be an arbitrary matrix, and $\MM = \schur{\MMhat}{\widehat{V}}$.
Then
\[
\PP^T \MM \PP
= \schur{\left[\begin{array}{cc}
\PP & 0\\
0 & \II_{\widehat{V}}
\end{array}\right]^T
\MMhat
\left[\begin{array}{cc}
\PP & 0\\
0 & \II_{\widehat{V}}
\end{array}\right]
}{\widehat{V}}.
\]
\end{lemma}

\begin{proof}
Let the rows and columns of $\MM$ be indexed by $V$.
It suffices to show that the matrix
\[
\left( 
\left[\begin{array}{cc}
\PP & 0\\
0 & \II_{\widehat{V}}
\end{array}\right]^T
\MMhat
\left[\begin{array}{cc}
\PP & 0\\
0 & \II_{\widehat{V}}
\end{array}\right] \right)^{-1}_{VV}
\]
is the same as $\left( \PP^T \MM \PP \right)^{-1}$.
This matrix can be written as:
\[
\left[\begin{array}{cc}
\PP^{-1} & 0\\
0 & \II_{\widehat{V}}
\end{array}\right]
\MMhat^{-1}
\left[\begin{array}{cc}
\PP^{-T} & 0\\
0 & \II_{\widehat{V}}
\end{array}\right].
\]
The top left block corresponding to $V$ gives
\[
\PP^{-1} \left(\MMhat^{-1} \right)_{VV} \PP^{-T}.
\]
The definition of Schur complements gives
$\MM^{-1} = \left(\MMhat^{-1} \right)_{VV}$,
which completes the proof.
\end{proof}

}{}

\section{Spectral Vertex Sparsification Algorithm}
\label{sec:vertexSparsify}

In this section, we give a nearly-linear work algorithm for computing
spectral vertex sparsifiers.
Recall that our goal is to approximate the matrix
\[
\schur{\MM}{F} = \MM_{CC} - \MM_{CF} \MM_{FF}^{-1} \MM_{FC}.
\]

Our algorithm approximates $\MM_{FF}^{-1}$ in a way analogous to the
recent parallel solver by Peng and Spielman~\cite{PengS14}.
It repeatedly writes the Schur complement as the average of the Schur complements
of two matrices.
The $FF$ block in one of these is diagonal, which makes its construction easy.
The other matrix is more strictly diagonally dominant than the previous one, so
that after a small number of iterations we can approximate it by a diagonal matrix.

\subsection{Spliting of Schur complement}
This spliting of the Schur complement is based on the following identity
from~\cite{PengS14}:
\begin{equation}\label{eqn:identity}
  (\DD - \AA)^{-1}
 =
 \frac{1}{2}
\left[
  \DD^{-1}
 + 
  \left(\II + \DD^{-1} \AA\right)
  \left(\DD - \AA \DD^{-1} \AA\right)^{-1}
  \left(\II + \AA \DD^{-1}\right)
 \right].
\end{equation}
We write $\MM_{FF} = \DD_{FF} - \AA_{FF}$ where $\DD_{FF}$ is diagonal
  and $\AA_{FF}$ has zero diagonal, and apply \eqref{eqn:identity} 
  to obtain the following expression for the Schur complement.
\begin{align} \label{eqn:rearrange}
\schur{\MM}{F}
& = \frac{1}{2} \left[ 2\MM_{CC} - \MM_{CF} \DD_{FF}^{-1} \MM_{FC}
\right. \nonumber \\ & \qquad \left.
-  \MM_{CF} \left(\II_{FF} + \DD_{FF}^{-1} \AA_{FF} \right)
  \left(\DD_{FF} - \AA_{FF} \DD_{FF}^{-1} \AA_{FF} \right)^{-1}
  \left(\II + \AA_{FF} \DD_{FF}^{-1} \right) \MM_{FC} \right].
\end{align}
Our key observation is that this is the average of the Schur complement
of two simpler matrices.
The first term is the Schur complement of:
\[
\left[
\begin{array}{cc}
\DD_{FF}& \MM_{FC}\\
\MM_{CF} & 0
\end{array}
\right],
\]
while the second term is the Schur complement of the matrix:
\[
\left[
\begin{array}{cc}
\DD_{FF} - \AA_{FF} \DD_{FF}^{-1} \AA_{FF} &
\left( \II + \AA_{FF} \DD_{FF}^{-1} \right)\MM_{FC}\\
\MM_{CF} \left( \II + \DD_{FF}^{-1} \AA_{FF}  \right)
&  2 \MM_{CC}
\end{array}
\right].
\]
This leads to a recursion similar to that used in~\cite{PengS14}.
However, to ensure that the Schur complements of both matrices are SDDM, 
  we move some of the diagonal from the $CC$ block of the second
  matrix to the $CC$ block of the first.
To describe this precisely, we use the notation 
  $\diag (\xx)$ to indicate the diagonal matrix whose entries
  are given by the vector $\xx$.
We also let $\one$ denote the all-ones vector.
So, $\diag(\xx)\one = \xx$.

\begin{lemma}
\label{lem:splitting}
Let $\MM$ be a SDDM matrix, and let $(F, C)$ be an arbitrary partition of its columns.
Let $\MM_{FF} = \DD_{FF} - \AA_{FF}$, where $\DD_{FF}$ is
a diagonal matrix and $\AA_{FF}$ is a nonnegative matrix with zero diagonal.
Define the matrices:
\begin{equation}
\MM_1 \defeq \left[
\begin{array}{cc}
\DD_{FF}& \MM_{FC}\\
\MM_{CF} & \diag(\MM_{CF} \DD_{FF}^{-1} \MM_{FC} \one_{C})
\end{array}
\right],
\label{eqn:firstHalf}
\end{equation}
and
\begin{equation}
\MM_2 \defeq
\left[
\begin{array}{cc}
\DD_{FF} - \AA_{FF} \DD_{FF}^{-1} \AA_{FF} &
\left( \II + \AA_{FF} \DD_{FF}^{-1} \right)\MM_{FC}\\
\MM_{CF} \left( \II + \DD_{FF}^{-1} \AA_{FF}  \right)
& 2 \MM_{CC} - \diag(\MM_{CF} \DD_{FF}^{-1} \MM_{FC} \one_{C}))
\end{array}
\right].
\label{eqn:secondHalf}
\end{equation}
Then $\schur{\MM_1}{F}$ is a Laplacian matrix, $\MM_2$ is a SDDM matrix,
  and 
\begin{equation}\label{part:average}
\schur{\MM}{F} = \frac{1}{2} \left( \schur{\MM_1}{F} + \schur{\MM_2}{F} \right).
\end{equation}
\end{lemma}

\begin{proof}
Equation \ref{part:average} follows immediately from equation~\ref{eqn:rearrange}.

To prove that $\schur{\MM_1}{F}$ is a Laplacian matrix, we observe that all of its
  off-diagonal entries are nonpositive, and that its row-sums are zero:
\[
  \schur{\MM_1}{F} \one_{C} = 
\diag(\MM_{CF} \DD_{FF}^{-1} \MM_{FC} \one_{C}) \one_{C}  
- \MM_{CF} \DD_{FF}^{-1} \MM_{FC} \one_{C} = \bvec{0}_{C}.
\]

To prove that $\MM_{2}$ is a SDDM matrix, we observe that
  all of its off-diagonal entries are also nonpositive.
For the $FF$ block this follows from from the nonnegativity
  of $\AA_{FF}$ and $\DD_{FF}$.
For the $FC$ and $CF$ blocks it follows from the nonpositivity of
  $\MM_{CF}$ and $\MM_{FC}$.
We now show that 
\[
  \MM_{2} \one \geq \MM \one .
\]
This implies that $\MM_{2}$ is an SDDM matrix, as it implies that
  its row-sums are nonnegative and not exactly zero.

We first analyze the row-sums in the rows in $F$.
\begin{align*}
(\MM_2 \one)_F 
& = 
\left[
\begin{array}{cc}
\DD_{FF} - \AA_{FF} \DD_{FF}^{-1} \AA_{FF} &
\left( \II + \AA_{FF} \DD_{FF}^{-1} \right)\MM_{FC}
\end{array}
\right]
\left[
\begin{array}{c}
\one_F\\
\one_C
\end{array}
\right]\\
& =  \DD_{FF} \one_{F} + \MM_{FC} \one_{C}
-\AA_{FF} \DD_{FF}^{-1} \left( \AA_{FF} \one_{F} - \MM_{FC} \one_{C} \right) \\
& \geq \DD_{FF} \one_{F} + \MM_{FC} \one_{C}
- \AA_{FF} \DD_{FF}^{-1} \DD_{FF} \one_{F}\\
& = \DD_{FF}  \one_{F} - \AA_{FF}  \one_{F} + \MM_{FC}  \one_{C}
\\
& = (\MM \one)_{F}.
\end{align*}

Before, analyzing the row-sums for rows in $C$,
  we derive an inequality.
As $\MM$ is diagonally dominant,
   every entry of 
  of $\DD_{FF}^{-1} ( \AA_{FF} \one_F - \MM_{FC} \one_C )$
  is between $0$ and $1$.
As $\MM_{FC}$ is non-positive, 
  this implies that
\[
\MM_{FC} \DD_{FF}^{-1} ( \AA_{FF} \one_F - \MM_{FC} \one_C )
\geq 
\MM_{FC} \one_{C}.
\]
Using this inequality, we obtain
\begin{align*}
(\MM_2 \one)_C
& =  \left[
\begin{array}{cc}
\MM_{CF} \left( \II + \DD_{FF}^{-1} \AA_{FF}  \right)
& 2 \MM_{CC} -  \diag(\MM_{CF} \DD_{FF}^{-1} \MM_{FC} \one_{C})
\end{array}
\right]
\left[
\begin{array}{c}
\one_F\\
\one_C
\end{array}
\right]\\
& = \MM_{CF} \one_F + \MM_{CF} \DD_{FF}^{-1} \AA_{FF}  \one_F
+ 2  \MM_{CC} \one_{C} 
- \diag(\MM_{CF} \DD_{FF}^{-1} \MM_{FC} \one_{C}) \one_C\\
& = \MM_{CF} \one_F + \MM_{CF} \DD_{FF}^{-1} \AA_{FF}  \one_F
+ 2  \MM_{CC} \one_{C} - \MM_{CF} \DD_{FF}^{-1} \MM_{FC} \one_{C}\\
& = \left( \MM_{CC} \one_{C} +  \MM_{CF} \one_F \right) 
+ \MM_{CC} \one_C + \MM_{CF} \DD_{FF}^{-1} \left( \AA_{FF} \one_F - \MM_{FC} \one_C \right)
\\
& \geq 
 \left( \MM_{CC} \one_{C} +  \MM_{CF} \one_F \right) 
+ \MM_{CC} \one_C + \MM_{CF} \one_{C}
\\
& = 2 (\MM \one)_{C}.
\end{align*}
\end{proof}

We first discuss how to approximate the Schur complement of $\MM_1$.

\begin{lemma}
\label{lem:schurDiag}
There is a procedure $\textsc{ApproxSchurDiag}(\MM, (F, C), \epsilon)$ that takes
a graph Laplacian matrix $\MM$ with $m$ non-zero entries, partition of variables $(F, C)$ and
returns a matrix matrix $\tilde{\MM}_{SC}$ such that:
\begin{enumerate}
\item $\tilde{\MM}_{SC}$ has $O(m \epsilon^{-4})$ non-zero entries, and
\item $\tilde{\MM}_{SC} \approx_{\epsilon} \MM_1$ where $\MM_1$ is defined in equation~\ref{eqn:firstHalf}.
\end{enumerate}
Furthermore, the procedure takes in $O(m \epsilon^{-4})$
work and $O(\log{n})$ depth.
\end{lemma}

The proof is based on the observation that this graph is a sum of product demand
graphs, one per vertex in $F$.
\iffull{
These demand graphs can be formally defined as:
\begin{definition}
\label{def:demand}
The product demand graph of a vector $\dd$, $G(\dd)$, is a complete weighted graph
whose weight between vertices $i$ and $j$ is given by
\[
\ww_{ij} = \dd_{i} \dd_{j}.
\]
\end{definition}
In Section~\ref{sec:weightedExp}, we give a result on directly
constructing approximations to these graphs that can be summarized as follows:
\begin{lemma}
\label{lem:weightedExpander}
There is a routine $\textsc{WeightedExpander}(\dd, \epsilon)$ such that
for any demand vector $\dd$ of length $n$ and a parameter $\epsilon$, $\textsc{WeightedExpander}(\dd, \epsilon)$ returns in $O(n \epsilon^{-4})$
work and $O(\log{n})$ depth a graph $H$ with $O(n \epsilon^{-4})$ edges such that
\[
\LL_{H} \approx_{\epsilon} \LL_{G(\dd)}.
\]
\end{lemma}
}{
These demand graphs is a complete weighted graph
whose weight between vertices $i$ and $j$ is given by $\ww_{ij} = \dd_{i} \dd_{j}$.
In the full version, we give a result on directly constructing linear-sized approximations to these graphs in linear time.
}

\begin{proof}(of Lemma~\ref{lem:schurDiag})
Since there are no edges between vertices in $F$,
the resulting graph consists of one clique among the neighbors
of each vertex $u \in F$.
Therefore it suffices to sparsify these separately.

It can be checked that the weight between two neighbors $v_1$ and $v_2$
in such a clique generated from vertex $u$ is $\frac{\ww_{uv_1} \ww_{u v_2}}{\dd_u}$.
Therefore we can replace it with a weighted expander given
in Lemma~\ref{lem:weightedExpander} above.
\end{proof}

Now, we can invoke Lemma~\ref{lem:schurDiag} on $\MM_1$ to compute its Schur complement,
which means it remains to iterate on $\MM_2$.
Of course, $\MM_2$ may be a dense matrix.
Once again, we approximate it implicitly using weighted expanders.
\iffull{
Here we also need weighted bipartite expanders:

\begin{definition}
\label{def:demand2}
The bipartite product demand graph of two vectors $\dd^{A}$, $\dd^{B}$,
$G(\dd^{A}, \dd^{B})$, is a weighted bipartite graph
whose weight between vertices $i \in A$ and $j \in B$ is given by
\[
\ww_{ij} = \dd^{A}_{i} \dd^{B}_{j}.
\]
\end{definition}

\begin{lemma}
\label{lem:weightedBipartiteExpander}
There is a routine $\textsc{WeightedBipartiteExpander}(\dd^{A}, \dd^{B} \epsilon)$ such that
for any demand vectors $\dd^{A}$ and $\dd^{B}$ of total length $n$ and a parameter $\epsilon$, it returns in $O(n \epsilon^{-4})$
work and $O(\log{n})$ depth a graph $H$ with $O(n \epsilon^{-4})$ edges such that
\[
\LL_{H} \approx_{\epsilon} \LL_{G(\dd^{A}, \dd^{B})}.
\]
\end{lemma}
}{
}

\begin{lemma}
\label{lem:squareSparsify}
There exists a procedure $\textsc{SquareSparsify}$ such that,
$\textsc{SquareSparsify}(\MM, (F, C), \epsilon)$ 
returns in $O(m \epsilon^{-4})$ work and $O(\log{n})$ depth a matrix
$\tilde{\MM}_2$ with $O(m \epsilon^{-4} )$ non-zero entries
such that $\tilde{\MM}_2 \approx_{\epsilon} \MM_2$,
where $\MM_2$ is defined in equation~\ref{eqn:secondHalf}.
\end{lemma}

\begin{proof}
The edges in this graph come from $-\MM_{FC}$,
$\AA_{FF} \DD_{FF}^{-1} \AA_{FF}$ and 
$\AA_{FF} \DD_{FF}^{-1} \MM_{FC}$.
The first is a subset, so we can keep them without increasing
total size by a more than a constant factor.
The later two consist of length two paths
involving some $u \in F$.
Therefore we can once again sum together a set of
expanders, one per each $u \in F$.

The edges in $\AA_{FF} \DD_{FF}^{-1} \AA_{FF}$
correspond to one clique with product demands
given by $\AA_{uv}$ for each $u \in F$, and
can be approximated using the weighted expander
in Lemma~\ref{lem:weightedExpander}.

The edges in $\AA_{FF} \DD_{FF}^{-1} \MM_{FC}$
can be broken down by midpoint into edges of weight
\[
\frac{\AA_{uv_F} \AA_{uv_C}}{\dd_{u}}
\]
where $v_F \in F$, $v_C \in C$ are neighbors of $u$.
This is a bipartite demand graph, so we can replace it
with the weighted bipartite expanders given in
Lemma~\ref{lem:weightedBipartiteExpander}.

The total size of the expanders that we generate is
$O(deg(u) \epsilon^{-4})$.
Therefore the total graph size follows from $\sum_{u \in F} deg(u) \leq m$.
\end{proof}

In the next subsection, we shows how to handle the case that the $\MM$ is $\alpha$-diagonally dominant matrix with large $\alpha$.
Therefore, the number of iterations of splitting depends on how diagonally dominant is the matrix.
Here we once again use the approach introduced in~\cite{PengS14}
by showing that $\MM_2$ is more diagonally dominant than $\MM$
by a constant factor.
This implies $O(\log(1 / \alpha \epsilon))$ iterations suffices for
obtaining a good approximation to the Schur complement.

\begin{lemma}
\label{lem:improve}
If $\DD - \AA$ is $\alpha$-strongly diagonally dominant and $\AA$
has $0$s on the diagonal, then
$\DD - \AA \DD^{-1} \AA$ is $ ((1 + \alpha)^2-1)$-strongly diagonally dominant.
\end{lemma}

\begin{proof}
Consider the sum of row $i$ in $\AA \DD^{-1} \AA$, it is
\[
\sum_{j } \sum_{k} \left| \AA_{ij} \DD_{jj}^{-1} \AA_{jk} \right|
= \sum_{j} \left|\AA_{ij} \right| \DD_{jj}^{-1} \sum_{k} \left| \AA_{jk} \right|
\leq \left(1 + \alpha\right)^{-1} \sum_{j} \left|\AA_{ij} \right|
\]
where the inequality follows from applying the fact that $\DD$
is $1 + \alpha$-strongly diagonally dominant to the $j\textsuperscript{th}$
row.
The result then follows from $\sum_{j} \left|\AA_{ij} \right| \leq  (1 + \alpha)^{-1} \DD_{ii}$.
\end{proof}

\iffull{
This notion is also stable under spectral sparsification.

\begin{lemma}
\label{lem:sparsifyOk}
If $\AA = \XX + \YY$ is $\alpha$-strongly diagonally dominant,
$\XX$ is diagonal, $\YY$ is a graph Laplacian, and $\YY \approx_{\epsilon} \tilde{\YY}$.
Then $\tilde{\AA} = \XX + \tilde{\YY}$ is $\exp\left(-\epsilon \right) \alpha$-strongly diagonally dominant.
\end{lemma}

\begin{proof}
Using $\YY \approx_{\epsilon} \tilde{\YY}$, we have
$$\tilde{\YY}_{i,i} \leq \exp(\epsilon)  \YY_{i, i}.$$
 
 The fact that $\AA$ is $\alpha$-strongly diagonally dominant also gives
 $\XX_{i, i} \geq \alpha \YY_{i, i}$.
 Combining these gives $ \XX_{i,i} \geq \exp(-\epsilon) \alpha  \tilde{\YY}_{i,i}$,
 which means $\XX + \tilde{\YY}$ is $\exp\left(-\epsilon \right) \alpha$-strongly diagonally dominant.
\end{proof}

}{}

\subsection{Schur Complement of Highly Strongly Diagonally Dominant Matrices}
\label{subsec:highlySDD}

It remains to show how to deal with the highly strongly diagonally
dominant matrix at the last step.
Directly replacing it with its diagonal, aka. \textsc{SquareSparsify}
is problematic.
Consider the case where $F$ contains $u$ and $v$
with a weight $\epsilon$ edge between them, and
$u$ and $v$ are connected to $u'$ and $v'$ in $C$
by weight $1$ edges respectively.
Keeping only the diagonal results in a Schur complement
that disconnects $u'$ and $v'$.
This however can be fixed by taking a step of random
walk within $F$.
Given a SDDM matrix $\MM_{FF} = \XX_{FF} + \LL_{FF}$ where $\LL_{FF}$
is a graph Laplacian and $\XX_{FF}$ is a diagonal matrix.
We will consider the linear operator
\begin{equation}
\ZZ_{FF}^{(last)}
\defeq
\frac{1}{2} \XX_{FF}^{-1}
	+ \frac{1}{2} \XX_{FF}^{-1} \left(  \XX_{FF} - \LL_{FF} \right) \XX_{FF}^{-1}
		\left( \XX_{FF} - \LL_{FF} \right) \XX_{FF}^{-1}
\label{eq:zzFinal}.
\end{equation}

\begin{lemma}
\label{lem:threeStep}
If $\MM_{FF} = \XX_{FF} + \LL_{FF} $ be a SDDM matrix that's $\alpha$-strongly
diagonally dominant for some $\alpha \geq 4$, then the operator
$\ZZ^{(last)}$ as defined in Equation~\ref{eq:zzFinal} satisfies:
\[
\MM_{FF} \preceq \left( \ZZ_{FF}^{(last)} \right)^{-1}
\preceq \MM_{FF} + \frac{2}{\alpha} \LL_{FF}.
\]
\end{lemma}

\begin{proof}
Composing both sides by $\XX^{-1/2}_{FF}$
and substituting in $\mathcal{L}_{FF} = \XX_{FF}^{-1/2} \LL_{FF} \XX_{FF}^{-1/2}$ means it
suffices to show
\[
\II + \mathcal{L}_{FF}
\preceq \left( \frac{1}{2} \II
	+ \frac{1}{2} \left( \II - \mathcal{L}_{FF} \right)^{2} \right)^{-1}
	\preceq \II + \mathcal{L}_{FF} +  \frac{2}{\alpha} \mathcal{L}_{FF}.
\]
The fact that $\MM_{FF}$ is $\alpha$-strongly diagonally dominant
gives $0 \preceq \LL_{FF} \preceq \frac{2}{\alpha} \XX_{FF}$, or
$0 \preceq \mathcal{L}_{FF} \preceq \frac{2}{\alpha} \II$ (Lemma \ref{lem:Xdom}).
As $\mathcal{L}_{FF}$ and $\II$ commute, the spectral theorem
means it suffices to show this for any scalar $0 \leq t \leq \frac{2}{\alpha}$.
Note that
\[
\frac{1}{2}
	+ \frac{1}{2} \left( 1 - t \right)^{2}
	= 1 - t + \frac{1}{2} t^2
\]
Taking the difference between the inverse of this
and the `true' value of $1 + t$ gives:
\[
\left( 1 - t + \frac{1}{2} t^2 \right)^{-1} - \left( 1 + t \right)
= \frac{1 - \left( 1 + t \right) \left( 1 - t + \frac{1}{2} t^2\right)}{1 - t + \frac{1}{2}t^2 }
= \frac{\frac{1}{2 } t^2 \left( 1 - t \right) }
	{1 - t + \frac{1}{2} t^2}
\]
Incorporating the assumption that $0 \leq t \leq \frac{2}{\alpha}$
and $\alpha \geq 4$ gives
that the denominator is at least
\[
1 - \frac{2}{\alpha} \geq \frac{1}{2},
\]
and the numerator term can be bounded by
\[
0 \leq \frac{t^2}{2} \left( 1 - t\right) \leq \frac{t}{\alpha}.
\]
Combining these two bounds then gives the result.
\end{proof}

To utilize $\ZZ^{(last)}$, note that the Schur complement of the matrix
\begin{equation}
\MM^{\left(last\right)} \defeq
\left[
\begin{array}{cc}
\left(\ZZ^{(last)}_{FF}\right)^{-1}& \MM_{FC}\\
\MM_{CF} & \MM_{CC}
\end{array}
\right]
\label{eqn:lastMM}
\end{equation}
equals to the average of the Schur complements of the matrices
\begin{equation}
\MM^{\left(last\right)}_{1} \defeq
\left[
\begin{array}{cc}
\XX_{FF}& \MM_{FC}\\
\MM_{CF} & \diag\left(\MM_{CF} \XX_{FF}^{-1} \MM_{FC} \one_{C}\right)
\end{array}
\right]
\label{eqn:lastFirstHalf}
\end{equation}
and
\begin{equation}
\MM^{\left(last\right)}_{2} \defeq
\left[
\begin{array}{cc}
\XX_{FF}& \left( \XX_{FF} - \LL_{FF} \right) \XX_{FF}^{-1} \MM_{FC}\\
\MM_{CF}  \XX_{FF} ^{-1} \left( \XX_{FF} - \LL_{FF} \right) &
2 \MM_{CC} - \diag\left(\MM_{CF} \XX_{FF}^{-1} \MM_{FC} \right).
\end{array}
\right]
\label{eqn:lastSecondHalf}
\end{equation}

The first term is SDDM by construction of its $CC$ portion
We can verify that the second term is also SDDM in a way that's
similar to Lemma~\ref{lem:splitting}.

\begin{lemma}
\label{lem:splitting2}
Let $\MM$ be a SDDM matrix, and let $(F, C)$ be an arbitrary partition of its columns.
Suppose that $\MM_{FF}$ is $\alpha$-strongly diagonally dominant for some $\alpha \geq 4$.
Define the matrices $\ZZ^{(last)}$, $\MM^{\left(last\right)}$, $\MM^{\left(last\right)}_{1}$ and $\MM^{\left(last\right)}_{2}$
as in Equations \ref{eq:zzFinal}, \ref{eqn:lastMM},  \ref{eqn:lastFirstHalf} and \ref{eqn:lastSecondHalf}.
Then, $\schur{\MM^{\left(last\right)}_{1}}{F}$ is a Laplacian matrix, $\MM^{\left(last\right)}_{2}$ is a SDDM matrix,
  and 
\begin{equation}\label{eqn:last_average}
\schur{\MM^{\left(last\right)}}{F} = \frac{1}{2} \left( \schur{\MM^{\left(last\right)}_{1}}{F} + \schur{\MM^{\left(last\right)}_{2}}{F} \right).
\end{equation}
\end{lemma}

\begin{proof}
\yintat{The proof is basically same as before, should we ignore that?}
Equation \ref{eqn:last_average} follows from substituting Equation~\ref{eq:zzFinal}
into Equations~\ref{eqn:lastMM}, ~\ref{eqn:lastFirstHalf} and~\ref{eqn:lastSecondHalf}.

To prove that $\schur{\MM^{\left(last\right)}_{1}}{F}$ is a Laplacian matrix, we observe that all of its
  off-diagonal entries are nonpositive, and that its row-sums are zero:
\[
  \schur{\MM^{\left(last\right)}_{1}}{F} \one_{C} = 
\diag(\MM_{CF} \XX_{FF}^{-1} \MM_{FC} \one_{C}) \one_{C}  
- \MM_{CF} \XX_{FF}^{-1} \MM_{FC} \one_{C} = \bvec{0}_{C}.
\]

To prove that $\MM^{\left(last\right)}_{2}$ is a SDDM matrix, we observe that
  all of its off-diagonal entries are also nonpositive.
For the $FF$ block this follows from from the nonnegativity
  of $\XX_{FF}$.
For the $FC$ and $CF$ blocks it follows from the nonpositivity of
  $\MM_{CF}$ and $\MM_{FC}$ and the fact that off-diagonal entries
  of $\LL_{FF}$ are nonpositive,
  the diagonal of $\LL_{FF}$ being bounded by $2/\alpha \XX_{FF}$,
  and $\alpha \geq 2$.
For the $CC$ block, it follows from the fact that
\[
2\MM_{CC} - \MM_{CF} \XX_{FF}^{-1} \MM_{FC}
\succeq 
2\MM_{CC} - 2\MM_{CF} \MM_{FF}^{-1} \MM_{FC}
\succeq 
\bvec{0}.
\]

We now show that 
\[
  \MM^{\left(last\right)}_{2} \one > \bvec{0} .
\]
This implies that $\MM^{\left(last\right)}_{2}$ is an SDDM matrix, as it implies that
  its row-sums are nonnegative and not exactly zero.

We first analyze the row-sums in the rows in $F$:
\begin{align*}
(\MM^{\left(last\right)}_{2} \one)_F 
& = 
\left[
\begin{array}{cc}
\XX_{FF} &
\left( \XX_{FF} - \LL_{FF} \right) \XX_{FF}^{-1} \MM_{FC}
\end{array}
\right]
\left[
\begin{array}{c}
\one_F\\
\one_C
\end{array}
\right]\\
& =  \XX_{FF} \one_{F} + \left( \XX_{FF} - \LL_{FF} \right) \XX_{FF}^{-1} \MM_{FC} \one_{C}\\
& >  \XX_{FF} \one_{F} - \left( \XX_{FF} - \LL_{FF} \right) \XX_{FF}^{-1} \XX_{FF} \one_{F}\\
& =  \bvec{0},
\end{align*}
where we used the fact $\MM_{FC} \one_{C} = -\LL_{FF} \one_{F} > -\XX_{FF} \one_{F}$ in the inequality.

For the row-sum in the rows in $C$, we obtain
\begin{align*}
(\MM^{\left(last\right)}_{2} \one)_C
& =  \left[
\begin{array}{cc}
\MM_{CF}  \XX_{FF} ^{-1} \left( \XX_{FF} - \LL_{FF} \right) &
2 \MM_{CC} - \diag\left(\MM_{CF} \XX_{FF}^{-1} \MM_{FC} \right)
\end{array}
\right]
\left[
\begin{array}{c}
\one_F\\
\one_C
\end{array}
\right]\\
& =  \MM_{CF} \one_F
+ 2  \MM_{CC} \one_{C} 
- \diag(\MM_{CF} \XX_{FF}^{-1} \MM_{FC} \one_{C}) \one_C\\
& >  \MM_{CC} \one_{C} 
- \MM_{CF} \XX_{FF}^{-1} \MM_{FC} \one_{C}\\
& >  \MM_{CC} \one_{C} 
+ \MM_{CF} \XX_{FF}^{-1} \XX_{FF} \one_{F}\\
& = \MM \one > \bvec{0}.
\end{align*}
\end{proof}

\richard{Can we go with $\DD$ instead?}
\yintat{I believe $\DD - \LL$ should be okay. I haven't check $\AA$ yet, they seems quite different.}
\begin{figure}[ht]
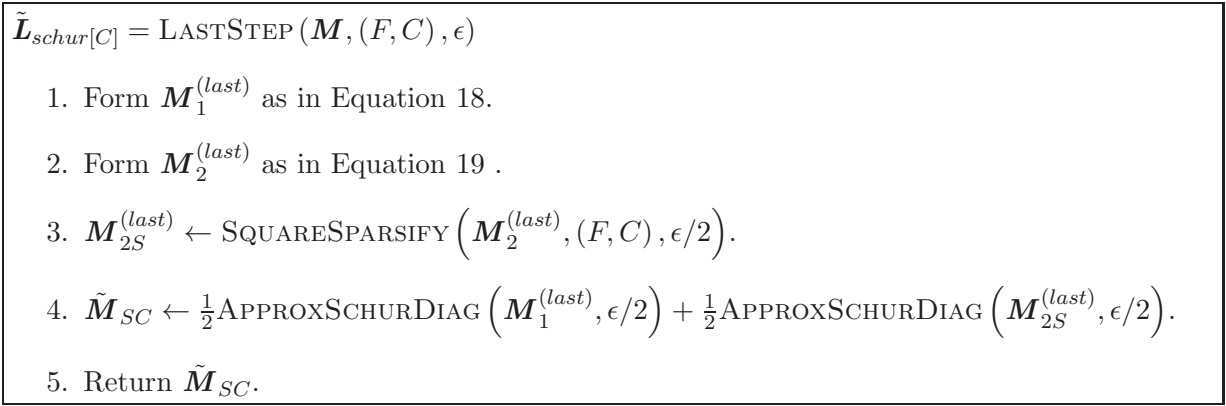


\begin{algbox}
$\tilde{\LL}_{schur[C]} = \textsc{LastStep}\left(\MM, \left(F, C\right), \epsilon \right)$
\begin{enumerate}
\item Form $\MM^{\left(last\right)}_{1}$ as in Equation~\ref{eqn:lastFirstHalf}.
\item Form $\MM^{\left(last\right)}_{2}$ as in Equation~\ref{eqn:lastSecondHalf} .
\item $\MM^{\left( last \right)}_{2S} \leftarrow
	\textsc{SquareSparsify}\left(\MM^{(last)}_2, \left(F, C\right), \epsilon/2 \right) $.
\item $\tilde{\MM}_{SC} \leftarrow
	\frac{1}{2} \textsc{ApproxSchurDiag}\left(\MM^{\left(last\right)}_{1}, \epsilon/2 \right)
	+ \frac{1}{2} \textsc{ApproxSchurDiag}\left(\MM^{\left(last\right)}_{2S}, \epsilon/2 \right)$.
\item Return $\tilde{\MM}_{SC}$.
\end{enumerate}
\end{algbox}

\caption{Pseudocode for approximating a highly strongly diagonally dominant matrix.
Small modifications on \textsc{ApproxSchurDiag} and
\textsc{SquareSparsify} is required to handle this case.}

\label{fig:lastStep}

\end{figure}

\begin{lemma}
\label{lem:lastStep}
Let $\MM$ be a SDDM matrix, and let $(F, C)$ be an arbitrary partition of its columns.
Suppose that $\MM_{FF}$ is $\alpha$-strongly diagonally dominant for some $\alpha \geq 4$.
There exists a procedure $\textsc{LastStep}$ such that,
$\textsc{LastStep}(\MM, (F, C), \epsilon)$ 
returns in $O(m \epsilon^{-8})$ work and $O(\log{n})$ depth a matrix
$\tilde{\MM}_{SC}$ with $O(m \epsilon^{-8} )$ non-zero entries
such that $\tilde{\MM}_{SC} \approx_{\epsilon+2/\alpha} \schur{\MM}{F}$.
\end{lemma}

\begin{proof}
We remark that Lemma~\ref{lem:schurDiag} is designed to compute Schur complement of the matrix~\ref{eqn:firstHalf} and 
Lemma~\ref{lem:squareSparsify} is designed to sparsify the matrix the matrix~\ref{eqn:secondHalf}. However, it is easy to modify them to work for computing the Schur complement of the matrix~\ref{eqn:lastFirstHalf} and sparsifying the matrix~\ref{eqn:lastSecondHalf}.

By Lemma~\ref{lem:squareSparsify}, we know that $\textsc{SquareSparsify}$  takes $O(m \epsilon^{-4})$ work and $O(\log{n})$ depth and outputs the matrix $\MM^{\left( last \right)}_{2S}$ with $O(m \epsilon^{-4} )$ non-zero entries. Therefore, Lemma~\ref{lem:schurDiag} shows that $\textsc{ApproxSchurDiag}$ takes $O(m \epsilon^{-8})$ work and $O(\log{n})$ depth and outputs a matrix with $O(m \epsilon^{-8} )$ non-zero entries. This proves the running time and the output size

For the approximation guarantee, Lemmas~\ref{lem:threeStep},~\ref{lem:squareSparsify},
and~\ref{lem:schurDiag} give:
\begin{align*}
\tilde{\MM}_{SC} & \approx_{2/\alpha} \frac{1}{2} \left( \schur{\MM^{\left(last\right)}_{1}}{F} + \schur{\MM^{\left(last\right)}_{2}}{F} \right) \\
& \approx_{\epsilon/2} \frac{1}{2} \left( \schur{\MM^{\left(last\right)}_{1}}{F} + \schur{\MM^{\left(last\right)}_{2S}}{F} \right) \\
& \approx_{\epsilon/2} \frac{1}{2} \textsc{ApproxSchurDiag}\left(\MM^{\left(last\right)}_{1}, \epsilon/2 \right)
	+ \frac{1}{2} \textsc{ApproxSchurDiag}\left(\MM^{\left(last\right)}_{2S}, \epsilon/2 \right) \\
& = \tilde{\MM}_{SC}.
\end{align*}
\end{proof}

\subsection{Summary}

Combining the splitting step and the final step gives our algorithm (Figure \ref{fig:approxSchur}).

\begin{figure}[ht]

\begin{algbox}
$\tilde{\LL}_{schur[C]} = \textsc{ApproxSchur}\left(\MM, \left(F, C\right), \alpha, \epsilon \right)$
\begin{enumerate}
\item Initialize $\tilde{\MM}_{SC} \leftarrow 0$, $\MM^{(0)} \leftarrow \MM$,
$d = \log_{1+\alpha}\left(13 \epsilon^{-1}\right)$
\item For $i$ from $1$ to $d$ do
\begin{enumerate}
\item Form $\MM^{\left(i - 1\right)}_{1}$ as in Equation~\ref{eqn:firstHalf}.
\item Form $\MM^{\left(i - 1\right)}_{2}$ as in Equation~\ref{eqn:secondHalf} .\item $\tilde{\MM}_{SC} \leftarrow \tilde{\MM}_{SC}
	+ \frac{1}{2} \textsc{ApproxSchurDiag}\left(\MM^{\left(i - 1\right)}_{1}, \frac{\epsilon}{3d} \right)$.
\item $\MM^{\left( i \right)} \leftarrow \frac{1}{2}
	\textsc{SquareSparsify}\left(\MM^{(i - 1)}, \left(F, C\right), \frac{\epsilon}{3d}\right) $.
\end{enumerate}
\item $\tilde{\MM}_{SC} \leftarrow \tilde{\MM}_{SC} + \textsc{LastStep}\left(\MM^{(d)}, \frac{\epsilon}{12}\right)$.
\item Return $\tilde{\MM}_{SC}$.
\end{enumerate}
\end{algbox}

\caption{Pseudocode for Computing Spectral Vertex Sparsifiers}

\label{fig:approxSchur}

\end{figure}

\begin{theorem}
\label{thm:aprox_schur}
Suppose that $\MM$ is $\alpha$-strongly diagonally dominant and $0<\epsilon<1$, then
$\textsc{ApproxSchur}$ returns
a matrix $\tilde{\MM}_{SC}$ with 
$O \left( m \left( \epsilon^{-1} \log_{\alpha}\left(\epsilon^{-1} \right) \right)^{O \left(\log_{\alpha}\left(\epsilon^{-1} \right) \right)} \right)$ non-zeros such that
\[
\tilde{\MM}_{SC} \approx_{\epsilon} \schur{\MM}{F}.
\]
in $O\left( m \left( \epsilon^{-1} \log_{\alpha}\left(\epsilon^{-1}  \right) \right)^{O \left(\log_{\alpha}\left(\epsilon^{-1}  \right) \right)} \right)$ work
and $O\left(\log_{\alpha}\left(\epsilon^{-1} \right) \log(n) \right)$ depth.
\end{theorem}

\begin{proof}
Let $\tilde{\MM}^{(i)}_{SC}$ denote the $\tilde{\MM}_{SC}$ after
$i$ steps of the main loop in \textsc{ApproxSchur}
We will show by induction that at each $i$,
\[
\schur{\MM}{F} \approx_{\frac{\epsilon i}{3d}} \tilde{\MM}^{(i)}_{SC} + \schur{\MM^{(i)}}{F}.
\]
The base case of $i = 0$ clearly holds.
For the inductive case, suppose we have the result for some $i$, then
\[
\schur{\MM}{F} \approx_{\frac{\epsilon i}{3d}} \tilde{\MM}^{(i)}_{SC}
+ \frac{1}{2} \left( \schur{\MM^{(i)}_1}{F} + \schur{\MM^{(i)}_2}{F} \right).
\]
Lemma~\ref{lem:schurDiag} gives
\begin{equation} \label{eqn:approx_schur_eq1}
\tilde{\MM}^{(i + 1)}_{SC}
= \tilde{\MM}^{(i)}_{SC} + \frac{1}{2} \textsc{ApproxSchurDiag}\left(\MM^{\left(i \right)}_{1}, (F,C), \frac{\epsilon}{3d} \right)
\approx_{\frac{\epsilon}{3d}} \tilde{\MM}^{(i)}_{SC} + \frac{1}{2}  \schur{\MM^{(i)}_1}{F},
\end{equation}
while Lemma~\ref{lem:squareSparsify} gives
\[
\MM^{(i + 1)} \approx_{\frac{\epsilon}{3d}} \frac{1}{2} \MM^{(i)}_2,
\]
which combined with the preservation of Loewner ordering
from Fact~\ref{fact:schurLoewner} gives
\begin{equation} \label{eqn:approx_schur_eq2}
\schur{\MM^{(i + 1)}}{F} \approx_{\frac{\epsilon}{3d}} \frac{1}{2} \schur{\MM^{(i)}_2}{F}.
\end{equation}
Combining these two bounds~\eqref{eqn:approx_schur_eq1} and~\eqref{eqn:approx_schur_eq2} then gives:
\[
\tilde{\MM}^{(i)}_{SC}
+ \frac{1}{2} \left( \schur{\MM^{(i)}_1}{F} + \schur{\MM^{(i)}_2}{F} \right)
\approx_{\frac{\epsilon}{3d}} \tilde{\MM}^{(i + 1)}_{SC} + \schur{\MM^{(i + 1)}}{F}.
\]
Hence, the inductive hypothesis holds for $i + 1$ as well.

By Lemmas~\ref{lem:improve} and~\ref{lem:sparsifyOk},
we have that $\MM^{(d)}_{FF}$ is $12\epsilon^{-1}$-strongly diagonally dominant
at the last step.
Lemma~\ref{lem:lastStep} then gives
\[
\MM^{(d)}_{1} \approx_{\frac{1}{3} \epsilon} \textsc{LastStep}\left(\MM^{(d)}_{1}, \frac{\epsilon}{12}\right).
\]
Composing this bound with the guarantees of the iterations
then gives the bound on overall error.
The work of these steps, and the size of the output graph
follow from Lemma~\ref{lem:schurDiag} and~\ref{lem:squareSparsify}.
\end{proof}

\iffull{
In our invocations to this routine, both $\alpha$ and $\epsilon$ will be set to constants.
As a result, this procedure is theoretically $O(m)$ time.
For a spectral vertex sparsification algorithm for handling general graph Laplacians,
$\alpha$ can be $0$ and we need to invoke spectral sparsifiers to $\LL_{i}$ after each step.
Any parallel algorithm for spectral sparsification
(e.g.~\cite{SpielmanT11,SpielmanS08:journal,OrecchiaV11,Koutis14}
will then lead to nearly linear work and polylog depth.
\begin{corollary}
Given a SDDM matrix with condition number $\kappa$,
a partition of the vertices into $(F, C)$, and error $\epsilon > 0$, we can compute in $O\left(m  \log^{O\left(1\right)} (n \kappa \epsilon^{-1}) \right)$ work
and $O\left( \log^{O\left(1\right)} (n \kappa \epsilon^{-1}) \right)$ depth a matrix
$\tilde{\MM}_{SC}$ with $O\left(n\log^{O(1)}n \epsilon^{-2} \right)$ non-zeros such that
\[
\tilde{\MM}_{SC} \approx_{\epsilon} \schur{\MM}{F}.
\]
\end{corollary}

\begin{proof}
We can add $\frac{\epsilon \trace{\MM} }{n\kappa}$ to each element on the diagonal to
obtain $\MM' \approx_{\epsilon} \MM$.
Therefore it suffices to assume that $\MM_{FF}$ is $\frac{1}{\poly(n) \kappa}$-strongly diagonally dominant.

Therefore Theorem~\ref{thm:aprox_schur} gives that \textsc{ApproxSchur} terminates
in $d = O(\log{\kappa} + \log{n})$ steps.
If we invoke a spectral sparsification algorithm at each step, the
number of non-zeros in each $\MM^{(i)}$ can be bounded by
$O(n \log^{O(1)}n (\epsilon/d)^{-2}) = O(n \log^{O\left(1\right)} (n \kappa \epsilon^{-1}) )$.
The overall work bound then follows from combining this with the
$\poly(\epsilon^{-1} d)$ increase in edge count at each step,
and the nearly-linear work guarantees of spectral sparsification algorithms.
\end{proof}

We remark that the setting of $\epsilon_i = 1/\log{\kappa}$ leads to a fairly
large number of log factors.
In the rest of this paper we only invoke spectral vertex sparsifiers with
moderate values of $\epsilon_i$ (unless we're at graphs that are smaller by
$\poly(n)$ factors).
Also, we believe recent developments in faster combinatorial spectral sparsification
algorithms~\cite{Koutis14} make faster algorithms for spectral vertex sparsifiers
a question beyond the scope of this paper.
}{}


\section{Algorithmic Constructions}
\label{sec:algo}

In this section, we gives two algorithms to compute vertex sparsifier
chains, the first algorithm uses existing spectral sparsifier for
graphs and the second algorithm does not. Although combining two approaches
gives a better theoretical result, we do not show it because we believe
there will be better spectral sparsifier algorithms for graphs soon
and hybrid approaches may not be useful then.

\subsection{Black Box Construction}\label{ssec:blackBox}

The first construction relies on existing parallel spectral sparsifer
algorithms. For concreteness, we use the parallel spectral graph sparsification
algorithm given by Koutis~\cite{Koutis14}.
\iffull{

\begin{theorem} \label{thm:sparsificaiton_result} Given any SDD
matrix $\MM$ with $n$ variables and $m$ non-zeros, there is an
algorithm $\textsc{BlackBoxSparsify}(\MM,\epsilon)$ outputs a
SDD matrix $\BB$ with $O(n\log^{3}n/\epsilon^2)$ non-zeros such that
$\MM\approx_{\epsilon}\BB$ in $O(\log^{3}n\log\alpha/\epsilon^2)$ depth
and $O((m+n\log^{3}n/\epsilon^2)\log^{2}n/\epsilon^2)$ work where
$\alpha =\frac{m}{n log^{3}n/\epsilon^2}.$
\end{theorem}
}{}
\begin{figure}[ht]

\begin{algbox}

$(\MM^{(1)},\MM^{(2)},\cdots;F_1,F_2,\cdots)=\textsc{BlackBoxConstruct}(\MM^{(0)})$ 
\begin{enumerate}
\item Let $k=1$, $\MM^{(1)}\leftarrow\MM^{(0)}$ and $F_0$ be the set of all variables.
\item While $\MM^{(k)}$ has more than $100$ variables

\begin{enumerate}
\item $\MM^{(k)}\leftarrow\textsc{BlackBoxSparsify}(\MM^{(k)},1/(k \log^{2}({k+4})))$.
\item Find a subset $F_k$ of size $\Omega(n^{(k)})$ such that $\MM_{F_k F_k}^{(k)}$
is $4$-strongly diagonally dominant. 
\item $\MM^{(k+1)}\leftarrow\textsc{ApproxSchur}(\MM^{(k)},(F_k,F_{k-1} \setminus F_{k}),4,1/(k \log^{2}({k+4})))$. 
\item $k\leftarrow k+1$. 
\end{enumerate}
\end{enumerate}
\end{algbox}

\caption{Pseudocode for Constructing Vertex Sparsifier
Chains Using Existing Spectral Sparsifiers}
\iffull{}{\ } 
\label{fig:blackBoxConstruct}

\end{figure}

In the $k^{th}$ step of the algorithm, we sparsify the graph and
compute an approximate Schur complement to $1/(k \log^{2}({k+1}))$ accuracy and
this makes sure the cumulative error is upper bounded by $\sum_{k=1}^{\infty}1/(k \log^{2}({k+1}))$ 
which is a constant.
 \begin{theorem} \label{thm:black_box}Given
any SDD matrix $\MM^{(0)}$ with $n$ variables and $m$ non-zeros,
the algorithm $\textsc{BlackBoxConstruct}(\MM^{(0)})$ returns a vertex sparsifier
chain such that the linear operator $\WW$ corresponding to it satisfies
\iffull{\[
\WW^{\dagger}\approx_{O(1)}\MM^{(0)}.
\]}{$\WW^{\dagger}\approx_{O(1)}\MM^{(0)}$.}
Also, we can evaluate $\WW b$ in $O(\log^{2}(n)\log\log n)$ depth
and $O(n\log^{3}n\log\log n)$ work for any vector $b$.

Furthermore, the algorithm $\textsc{BlackBoxConstruct}(\MM^{(0)})$
runs in $O(\log^{6}n\log^{4}\log n)$ depth and $O(m\log^{2}n + n \log^{5} n)$
work. 
\end{theorem}

\begin{proof}
Let $n^{(k)}$ and $m^{(k)}$
be the number of vertices and non zero entries in matrix $\MM^{(k)}$.
Let $s(n)=n\log^{3}n$ which is the output size of $\textsc{BlackBoxSparsify}$
and $\epsilon(k)=1/(k \log^{2}(k+4))$ which is the accuracy of the $k$-th sparsification and approximate schur complement.

We first prove the correctness of the algorithm.
The ending condition ensures $\MM^{(last)}$ has size $O(1)$;
step $(2a)$ and $(2c)$ ensures $\MM^{(k+1)}\approx_{2 \epsilon(k)}SC(\MM^{(k)},F_k)$ and
step $(2b)$ ensures $\MM_{F_k F_k}^{(k)}$ is $4$ strongly diagonally dominant.
Therefore, the chain $(\MM^{(1)},\cdots;F_1,\cdots)$ is a vertex sparsifier chain.
Since the cumlative error $\sum \epsilon(k)=O(1)$, Lemma \ref{lem:apply_chain}
shows that the resultant operator $\WW$ satisfies 
\[
\WW^{\dagger}\approx_{O(1)}\MM^{(0)}.
\]

Now, we upper bound the cost of evaluating $\WW b$.
Lemma \ref{lem:subsetSimple} shows that $\left|F_k\right|=\Omega(n^{(k)})$
and hence a constant portion of variables is eliminated each iteration.
Therefore, $n^{(k)}\leq c^{k-1}n$ for some $c$. Using this, Lemma
\ref{lem:apply_chain} shows the depth for evaluating $\WW b$ is
\[
O(\sum_{k=1}^{O(\log n)}\log(k)\log(n))=O(\log^{2}(n)\log\log n)
\]
and the work for evaluating $\WW b$ is 
\[
O(\sum_{k=1}^{O(\log n)}\log(k)s(c^{k-1}n)/\epsilon(k)^2).
\]
Using $s(n)=n\log^{3}n$ and $\epsilon(k)=1/(k \log^{2}(k+4))$, the work for evaluating is simply $O(s(n))$.

For the work and depth of the construction, Lemma \ref{lem:subsetSimple}
shows that it takes $O(m^{(k)})$ work and $O(\log n^{(k)})$ depth
to find $F_k$ and Theorem \ref{thm:aprox_schur} shows that $\textsc{ApproxSchur}$
takes $O(m^{(k)}k^{O(\log k)})$ work and $O(\log n^{(k)} \log k)$ depth.
Using $n^{(k)}\leq c^{k-1}n$ and $m^{(i)}=s(n^{(i)})/\epsilon(i)^2$, the total
work for this algorithm excluding $\textsc{BlackBoxSparsify}$ is
\[
\sum_{k=1}^{O(\log n)}O(s(c^{k-1}n)k^{O(\log k)}/\epsilon(k)^2)=O(s(n)).
\]
Hence, the total work for $\textsc{BlackBoxConstruct}$ is 
\[
O(s(n))+O(m\log^{2}n)+\sum_{k=2}^{O(\log n)}O(s(n^{(k)})k^{O(\log k)}\log^{2}n^{(k)}/\epsilon(k)^2).
\]
Using $s(n^{(k)})$ is geometric decreasing, the total work is $O(m\log^{2}n + n \log^{5}n)$.
We can bound the total depth similarly. 
\end{proof} 

\begin{remark} \iffull{Given an sparsifier algorithm that takes
$d(m,n)$ depth and $w(m,n)/\epsilon^2$ work to find a sparsifer of size $s(n)/\epsilon^2$,
the $\textsc{BlackBoxConstruct}$ roughly takes $O(\log^{2}n \log \log n)+O(d(m,n) \log n)$
depth and $O(w(m,n))$ work to construct a vertex sparsifier chain
and such chain has total depth $O(\log^{2}n\log\log n)$ and total
work $O(s(n))$.

Therefore}{In general}, the work for preprocessing is roughly linear to the
work needed to sparsify and the work for solving is linear to the
size of sparsifier. Hence, solving Laplacian system is nearly as simple
as computing sparsifier.\end{remark}

\subsection{Recursive Construction}\label{ssec:recursive}

We now give a recursive construction based on the idea that solvers
can be used to compute sampling probabilities~\cite{SpielmanS08:journal}.
We will describe the construction in phases, each containing $r$ iterations.
Each iteration decreases the number of vertices while maintaining the
density of graph.
\iffull{
We maintain
the density by the general sparsification technique introduced by~\cite{cohen2014uniform} as follows:
\begin{lemma}[\cite{cohen2014uniform}]
\label{lem:sparsify} Given $\mathcal{M}$ be a class of positive definite
$n\times n$ matrices. Let $\mathcal{M}(m)$ be the set of all $\BB^{T}\BB\in\mathcal{M}$
such that $\BB$ has $m$ rows. Assume that 
\begin{enumerate}
\item For any $\BB^{T}\BB\in\mathcal{M}$ and non negative diagonal matrix
$\DD$, we have $\BB^{T}\DD\BB\in\mathcal{M}$.
\item For any matrix $\BB^{T}\BB\in\mathcal{M}$, we can check if every row
$b$ is in $im(\BB^{T})$ or not in depth $d_{chk}(m)$ and work $w_{chk}(m)$. 
\item For any $\BB^{T}\BB\in\mathcal{M}(m)$, we can find an implicit representation
of a matrix $\WW$ such that $\WW\approx_{1}(\BB^{T}\BB)^{\dagger}$
in depth $d_{con}(m,n)$ and work $w_{con}(m,n)$ and for any vector
$b$, we can evaluate $\WW b$ in depth $d_{eval}(m,n)$ and work $w_{eval}(m,n)$. 
\end{enumerate}
For any $k\geq1$, $1\geq\epsilon>0$ and matrix $\BB^{T}\BB\in\mathcal{M}(m)$,
the algorithm $\textsc{Sparsify}(\BB^{T}\BB,k,\epsilon)$ outputs
an explicit matrix $\CC^{T}\CC\in\mathcal{M}(O(kn\log n/\epsilon^{2}))$
with $\CC^{T}\CC\approx_{\epsilon}\BB^{T}\BB$. 

Also, this algorithm
runs in $d_{con}\left(\frac{m}{k},n\right)+O(d_{eval}(m,n)+d_{chk}(m)+\log n)$
depth and $w_{con}\left(\frac{m}{k},n\right)+O(w_{eval}(m,n) \log n +w_{chk}(m)+m\log n)$
work.
\end{lemma}}{

We maintain the density by the general sparsification technique introduced by~\cite{cohen2014uniform},
this technique allows us to sparsify a graph with $m$ edges via solving the Laplacian $O(m/k)$ sampled graphs of size $O(k n log(n))$.

}
Each call of spectral vertex sparsiﬁcation increases edge density,
but the $\textsc{Sparsify}$ routine allows us to reduce the density
at a much faster rate. A higher reduction parameter $r$ in the algorithm $\textsc{RecursiveConstruct}_{r}$ allows us
to reduce cost of these recursive sparsiﬁcation steps. 

\begin{figure}[ht]

\begin{algbox}

$(\MM^{(1)},\MM^{(2)},\cdots;F_1,F_2,\cdots)=\textsc{RecursiveConstruct}_{r}(\MM^{(0)})$
\begin{enumerate}
\item $\MM^{(1)}\leftarrow\textsc{Sparsify}(\MM^{(0)},2^{c_{2}r},1/4)$,
 $k\leftarrow1$  and $F_0$ be the set of all variables.
\item While $\MM^{(k)}$ has more than $\Theta(1)^{r}$ vertices,

\begin{enumerate}
\item Find a subset $F_k$ of size $\Omega(n^{(k)})$ such that $\MM_{F_k F_k}$
is $4$-strongly diagonally dominant.
\item $\MM^{(k+1)}\leftarrow\textsc{ApproxSchur}(\MM^{(k)},(F_k,F_{k-1} \setminus F_{k}),4,(k+8)^{-2})$.
\item If $k+1\text{ mod }r=0$, Then

\begin{enumerate}
\item $\MM^{(k+1)}\leftarrow\textsc{Sparsify}(\MM^{(k+1)},(k+9)^{-2},2^{2c_{2}r\log^{2}(k+1)}).$
\end{enumerate}
\item $k\leftarrow k+1$.
\end{enumerate}
\end{enumerate}
\end{algbox}

\caption{Pseudocode for Recursively Constructing Vertex
Sparsifier Chains}

\label{fig:recursiveConstruct}

\end{figure}

\iffull{
The following lemma proves that the algorithm $\textsc{RecursiveConstruct}_{r}$
produces a vertex sparsifier chain and the linear operator corresponding
to the vertex sparsifier can be evaluated efficiently.
\begin{lemma}
\label{lem:eliminiate_vertex}Given a large enough constant $r$.
There are universal constants $0<c_{1}<1$ and $c_{2}>0$ such that
for any SDD matrix $\MM^{(0)}$ with $n$ variables, the algorithm
$\textsc{RecursiveConstruct}_{r}(\MM^{(0)})$ returns a vertex sparsifier
chain $(\MM^{(1)},\MM^{(2)},\cdots;F_1,F_2,\cdots)$ satisfying
the following conditions
\begin{enumerate}
\item For all $k\geq1$, $n^{(k)}\leq c_{1}^{k-1}n$ where $n^{(k)}$ are
the number of variables in $\MM^{(k)}$.
\item Except step 1, at any moment, all intermediate matrices $\MM$ appears
at the $k^{th}$ iteration has density
\[
\frac{m'}{n'\log n'}\leq2^{3c_{2}r\log^{2}k}
\]
for $k>1$ where $m'$ and $n'$ are the number of non-zeros and variables
of $\MM$.
\item For all $k\geq1$, $\MM_{F_k F_k}^{(k)}$ is $4$-strongly
diagonally dominant,
\item For all $k\geq1,$ $\MM^{(k+1)}\approx_{2(k+8)^{-2}}\schur{\MM^{(k)}}{F_k}$.
\end{enumerate}
Furthermore, the linear operator $\WW$ corresponding to the vertex
sparsifier chain satisfies
\[
\WW\approx_{1}\left(\MM^{(0)}\right)^{\dag}.
\]
Also, we can evaluate $\WW b$ in $O(\log^{2}n\log\log n)$ depth
and $2^{O(r\log^{2}r)}n\log n$ work for any vector $b$.\end{lemma}
\begin{proof}
For the assertion (1), we note that the step (2a) ensures $\left|F_k\right|=\Omega(n^{(k)})$
and hence a constant portion of variables is eliminated each iteration.
This proves $n^{(k)}\leq c^{k-1}n$ for some $c$.

For the assertion (2), Theorem \ref{thm:aprox_schur} shows that after the approximate Schur complement
\begin{eqnarray*}
m^{(k+1)} & = & O(m^{(k)}(k^{2}\log(k+8))^{O(\log(k+8))})\\
 & \leq & 2^{O\left(\log^{2}(k+1)\right)}m^{(k)}.
\end{eqnarray*}
Hence, it shows that each iteration the density increases by at most
$2^{c_{2}\log^{2}(k+1)}$ for some constant $c_{2}$. After the $\textsc{Sparsify}$
step in (2ci), we have
\[
\frac{m^{(sr)}}{n^{(sr)}\log n^{(sr)}}\leq2^{2c_{2}r\log^{2}(sr)}.
\]
Then, after $r$ iterations of $\textsc{ApproxSchur}$ and before
the sparsification of $\MM^{((s+1)r)}$, we have
\begin{eqnarray*}
\frac{m^{((s+1)r)}}{n^{((s+1)r)}\log n^{((s+1)r)}} & \leq & 2^{2c_{2}r\log^{2}(sr)}2^{c_{2}\log^{2}(sr+1)}\cdots2^{c_{2}\log^{2}((s+1)r)}\\
 & \leq & 2^{3c_{2}r\log^{2}((s+1)r)}.
\end{eqnarray*}
This proves the assertion (2).

For the assertion (3), it follows from the construction of $F_k$ in step
(2a).

For the assertion (4), we note that in step (2b), we construct the approximate Schur complement
$\MM^{(k+1)}$ such that $\MM^{(k+1)}\approx_{(k+8)^{-2}}\schur{\MM^{(k)}}{F_k}$.
Therefore, we only need to check $\MM^{(sr)}$ for all $s$ because
$\MM^{(sr)}$ is modified at step (2ci) after the sparsification.
Note that Lemma \ref{lem:sparsify} guarantee that $\MM^{(k)}$ changes
only by $(k+8)^{-2}$ factor. Hence, in total, we have $\MM^{(k+1)}\approx_{2(k+8)^{-2}}\schur{\MM^{(k)}}{F_k}$.

For the last claim, Lemma \ref{lem:apply_chain} shows that
\[
\WW\approx_{1/2+4\sum_{k}(k+8)^{-2}}\left(\MM^{(0)}\right)^{\dag}
\]
and we can evaluate $\WW b$ in $$O(\sum_{k}\log k\log n^{(k)})=O(\log^{2}n\log\log n)$$
depth and $$O(\sum_{k}2^{3c_{2}r\log^{2}k}n^{(k)}\log n^{(k)}\log k)=2^{O(r\log^{2}r)}n\log n$$
work.
\end{proof}
In the algorithm $\textsc{RecursiveConstruct}_{r}$, we call the $(sr+1)^{th}$
to the $((s+1)r)^{th}$ iteration as the $s^{th}$ phase. At the end
of each phase, the $\textsc{Sparsify}$ is called once. The previous
lemma showed that the density of the graph at the $k^{th}$ iteration
is less than $2^{3c_{2}r\log^{2}k}$. This explains our choice of
reduction factor $2^{2c_{2}r\log^{2}k}$ in the $\textsc{Sparsify}$
algorithm as follows:
\begin{lemma}
\label{lem:each_phase_blow_up}Let $n^{(k)}$ is the number of variables
of $\MM^{(k)}$. From the $(sr+1)^{th}$ to the $((s+1)r)^{th}$ iteration
including the $\textsc{Sparsify}$ call at the end, the algorithm
takes 
\[
2^{O(r\log^{2}(sr))}n^{(sr+1)}\log^2 n^{(sr+1)}
\]
work and
\[
O(r\log(sr)\log^2 n^{(sr)})+O(\log^{2}n^{(sr+1)}\log\log n^{(sr+1)})
\]
depth and the time to construct the vertex sparsifier chain for a
SDD matrix with $n^{((s+1)r)}$ variables and $2^{c_{2}r\log^{2}((s+1)r)}n^{((s+1)r)}\log n^{((s+1)r)}$
non zeros.\end{lemma}
\begin{proof}
Let $m^{(k)}$ and $n^{(k)}$ be the number of non zeros and variables
in $\MM^{(k)}$ before the $\textsc{Sparsify}$ call if there is.
Lemma \ref{lem:subsetSimple} and Theorem \ref{thm:aprox_schur} shows
that the depth and work of the $k^{th}$ iteration takes $O(m^{(k)}+m^{(k+1)})$
work and $O(\log k\log n^{(k)})$ depth. Lemma \ref{lem:eliminiate_vertex}
shows that 
\[
n^{(k)}\leq c_{1}^{k-1}n\text{ and }m^{(k)}\leq2^{3c_{2}r\log^{2}k}n^{(k)}\log n^{(k)}
\]
and hence, from the $(sr+1)^{th}$ to the $((s+1)r)^{th}$ iteration
(excluding the $\textsc{Sparsify}$ call at the end), the algorithm
takes
\begin{eqnarray*}
\sum_{k=sr+1}^{(s+1)r}O\left(m^{(k)}+m^{(k+1)}\right) & \leq & \sum_{k=sr+1}^{(s+1)r}2^{O(r\log^{2}k)}n^{(k)}\log n^{(k)}\\
 & \leq & 2^{O(r\log^{2}(sr))}n^{(sr+1)}\log n^{(sr+1)}
\end{eqnarray*}
work and
\begin{eqnarray*}
\sum_{k=sr+1}^{(s+1)r}O\left(\log k\log n^{(k)}\right) & \leq & O(r\log(sr)\log n^{(sr)})
\end{eqnarray*}
depth.

Now, we bound the cost of the $\textsc{Sparsify}$ call. Let $m^{*}$
and $n^{*}$ be the the number of non zeros and variables in $\MM^{((s+1)r)}$
before the $\textsc{Sparsify}$ call. Lemma \ref{lem:sparsify} shows
that the $\textsc{Sparsify}$ call takes $d_{con}\left(m^{*}2^{-2c_{2}r\log^{2}((s+1)r)},n^{*}\right)+O(d_{eval}(m^{*},n^{*})+d_{chk}(m^{*})+\log n^{*})$
depth and $w_{con}\left(m^{*}2^{-2c_{2}r\log^{2}((s+1)r)},n^{*}\right)+O(w_{eval}(m^{*},n^{*}) \log n^{*}+w_{chk}(m^{*})+m^{*}\log n^{*})$
work.

For any SDD matrix $\BB^{T}\BB$, an edge $b\in im(\BB^{T})$ if and
only if the end points of the edge is in the same connected component
of the graph corresponding to $\BB^{T}\BB$. Halperin and Zwick \cite{halperin1996optimal}
shows how to compute the connected components of a graph with $m$ edges
and $n$ vertices in $O(\log n)$ depth and $O(m+n)$ work for the
EREW PRAM model. Using this, we can check every edge in $O(\log n)$
depth and $O(m+n)$ work.

To construct an implicit approximate inverse for the sampled SDD matrix,
we can use $\textsc{RecursiveConstruct}_{r}$. Lemma \ref{lem:eliminiate_vertex}
showed that it takes $O(\log^{2}n^{*}\log\log n^{*})$ depth and $2^{O(r\log^{2}r)}n^{*}\log n^{*}$
work to apply the approximate inverse once.

Hence, the total running time from the $(sr+1)^{th}$ to the $((s+1)r)^{th}$
iteration including the $\textsc{Sparsify}$ call is the
time to construct the vertex sparsifier chain plus
\[
2^{O(r\log^{2}(sr))}n^{(sr+1)}\log^2 n^{(sr+1)}
\]
extra work and
\[
O(r\log(sr)\log n^{(sr)})+O(\log^{2}n^{(sr+1)}\log\log n^{(sr+1)})
\]
extra depth.
\end{proof}
Note that at the end of the $s^{th}$ phase, the time required to construct an extra vertex
sparsifier chain for the $\textsc{Sparsify}$ call is less than the
remaining cost after the $s^{th}$ phase. This is the reason why we
use $2^{2c_{2}r\log^{2}k}$ as the reduction factor for the $\textsc{Sparsify}$
call. The following theorem takes account for the recursive call and
show the total running time for the algorithm. 
}{}
\begin{lemma}
\label{lem:recursive_con}
With high probability, the algorithm
$\textsc{RecursiveConstruct}_{r}(\MM^{(0)})$ returns a
vertex sparsifier chain such that the linear operator $\WW$ corresponding
to it satisfies
\iffull{\[
\WW\approx_{1}\left(\MM^{(0)}\right)^{\dag}.
\]}{$\WW\approx_{1}\left(\MM^{(0)}\right)^{\dag}$.}
Assume $r\log^{2}r=o(\log n)$, we can evaluate $\WW b$ in $O(\log^{2}n\log\log n)$ depth
and $2^{O(r\log^{2}r)}n\log n$ work. Also, the algorithm $\textsc{RecursiveConstruct}_{r}\left(\MM^{(0)}\right)$
takes $2^{O(\log n/r)}$ depth and $m\log n+2^{O(r\log^{2}r)}n\log n$
work.\end{lemma}
\begin{proof}
All result is proved in lemma \ref{lem:eliminiate_vertex} except
the construction time.

To bound the construction time, we first consider the case $\MM^{(0)}$
has only $2^{c_{2}r}n\log n$ non-zeros. In that case, the algorithm skips 
step 1 because the matrix is already sparse. Lemma
\ref{lem:each_phase_blow_up} shows that during the $s^{th}$ phase,
the $\textsc{Sparsify}$ call requires us to construct an extra vertex
sparsifier chain for a matrix with $n^{((s+1)r)}$ variables and at
most $2^{c_{2}r\log^{2}((s+1)r)}n^{((s+1)r)}\log n^{((s+1)r)}$ non-zeros.
Also, we know that the $\textsc{Sparsify}$ returns a matrix with
$n^{((s+1)r)}$ variables and $2^{2c_{2}r\log^{2}((s+1)r)}n^{((s+1)r)}\log n^{((s+1)r)}$
non-zero. Hence, the cost of remaining iteration (excluding the recursion
created afterward) is larger than the cost to construct the extra
vertex sparsifier chain required at the $s^{th}$ phase.

Hence, considering this recursion factor, the running time of the
$s^{th}$ phase is multiplied by a factor of $2^{s}$.

Since there are $O(\log n/r)$ phases and $r\log^{2}r=o(\log n)$,
the total depth of the algorithm is 
\begin{eqnarray*}
 &  & \sum_{s=1}^{O(\log n/r)}2^{s}\left(r\log(sr)\log^2 n^{(sr)}+\log^{2}n^{(sr)}\log\log n^{(sr)}\right)\\
 & = & 2^{O(\log n/r)}O\left(r\log\log n\log^2 n^{(last)}+\log^{2}n^{(last)}\log\log n^{(last)}\right)\\
 & = & 2^{O(\log n/r)}O\left(r^{2}\log\log n+r^{2}\log r\right)\\
 & = & 2^{O(\log n/r)}r^{2}\log\log(n)\\
 & = & 2^{O(\log n/r)}
\end{eqnarray*}
and the total work of the algorithm is 
\begin{eqnarray*}
 &  & \sum_{s=1}^{O(\log n/r)}2^{s}\left(2^{O(r\log^{2}(sr))}n^{(sr+1)}\log^2 n^{(sr+1)}\right)\\
 & = & 2^{O(r\log^{2}r)}n\log^2 n.
\end{eqnarray*}

For general $m$, during the first step, $\textsc{Sparsify}$, we
need to solve a certain SDD matrix with at most $m^{(0)}2^{-c_{2}r}$
non-zeros and $n^{(0)}$ variables. To solve that SDD matrix, we use
$\textsc{RecursiveConstruct}_{r}$ to construct a vertex sparsifier
chain and use the chain to solve that $O(\log(n))$ different right
hand sides. Using $r\log^{2}r=o(\log n)$, the total depth for this
algorithm is
\[
O\left(\log_{2^{r}}\left(\frac{m}{n\log n}\right)2^{O(\log n/r)}\right)=\log\left(m\right)2^{O(\log n/r)}=2^{O(\log n/r)}.
\]
and the total work of the algorithm is
\begin{eqnarray*}
 &  & m\log n+2^{O(r\log^{2}r)}n\log^2 n\log_{2^{r}}\left(\frac{m}{n\log n}\right)\\
 & = & m\log n+2^{O(r\log^{2}r)}n\log^2 n\log\left(\frac{m}{n\log n}\right).
\end{eqnarray*}
Note that the first term dominate if $\frac{m}{n}\geq2^{O(r\log^{2}r)}$
and hence we can simplify the term to
\[
m\log n+2^{O(r\log^{2}r)}n\log^2 n.
\]
\end{proof}
The following theorem follows from Lemma \ref{lem:recursive_con} by setting $r=\log\log\log n$.
\begin{theorem}
\label{thm:non_combin_result}
Given any SDD matrix $\MM$ with $n$ variables and $m$ non-zeros.
We can find an implicit block-Cholesky factorization for the matrix
$\MM$ in $O(m\log n+n\log^{2+o(1)}n)$ work and $O(n^{o(1)})$ depth
such that for any vector $b$, we can compute an $\epsilon$ approximation
solution to $\MM^{-1}b$ in $O((m+n\log^{1+o(1)}n)\log(1/\epsilon))$
work and $O(\log^{2}n\log\log n\log(1/\epsilon))$ depth.\end{theorem}

\section*{Acknowledgements}

We thank Michael Cohen for notifying us of several issues
in previous versions of this manuscript.


\begin{thebibliography}{CKM{\etalchar{+}}14}

\bibitem[AZLO15]{Allen-ZhuLL15}
Zeyuan Allen-Zhu, Zhenyu Liao, and Lorenzo Orecchia.
\newblock Spectral sparsification and regret minimization beyond matrix
  multiplicative updates.
\newblock In {\em Proceedings of the Forty-Seventh Annual ACM on Symposium on
  Theory of Computing}, STOC '15, pages 237--245, New York, NY, USA, 2015. ACM.

\bibitem[BGH{\etalchar{+}}06]{SupportGraph}
M.~Bern, J.~Gilbert, B.~Hendrickson, N.~Nguyen, and S.~Toledo.
\newblock Support-graph preconditioners.
\newblock {\em SIAM J. Matrix Anal. \& Appl}, 27(4):930--951, 2006.

\bibitem[BHV08]{BomanHV04}
Erik~G. Boman, Bruce Hendrickson, and Stephen~A. Vavasis.
\newblock Solving elliptic finite element systems in near-linear time with
  support preconditioners.
\newblock {\em SIAM J. Numerical Analysis}, 46(6):3264--3284, 2008.

\bibitem[Bra77]{brandt1977multi}
Achi Brandt.
\newblock Multi-level adaptive solutions to boundary-value problems.
\newblock {\em Mathematics of computation}, 31(138):333--390, 1977.

\bibitem[BSS12]{BSS}
Joshua Batson, Daniel~A Spielman, and Nikhil Srivastava.
\newblock Twice-{Ramanujan} sparsifiers.
\newblock {\em SIAM Journal on Computing}, 41(6):1704--1721, 2012.

\bibitem[CKM{\etalchar{+}}11]{ChristianoEtAl}
Paul Christiano, Jonathan~A. Kelner, Aleksander Madry, Daniel~A. Spielman, and
  Shang-Hua Teng.
\newblock Electrical flows, laplacian systems, and faster approximation of
  maximum flow in undirected graphs.
\newblock In {\em Proceedings of the 43rd annual ACM symposium on Theory of
  computing}, STOC '11, pages 273--282, New York, NY, USA, 2011. ACM.

\bibitem[CKM{\etalchar{+}}14]{CohenKMPPRX}
Michael~B. Cohen, Rasmus Kyng, Gary~L. Miller, Jakub~W. Pachocki, Richard Peng,
  Anup~B. Rao, and Shen~Chen Xu.
\newblock Solving sdd linear systems in nearly mlog1/2n time.
\newblock In {\em Proceedings of the 46th Annual ACM Symposium on Theory of
  Computing}, STOC '14, pages 343--352, New York, NY, USA, 2014. ACM.

\bibitem[CLM{\etalchar{+}}14]{cohen2014uniform}
Michael~B Cohen, Yin~Tat Lee, Cameron Musco, Christopher Musco, Richard Peng,
  and Aaron Sidford.
\newblock Uniform sampling for matrix approximation.
\newblock {\em arXiv preprint arXiv:1408.5099}, 2014.

\bibitem[DS08]{daitch2008faster}
Samuel~I Daitch and Daniel~A Spielman.
\newblock Faster approximate lossy generalized flow via interior point
  algorithms.
\newblock In {\em Proceedings of the 40th annual ACM symposium on Theory of
  computing}, pages 451--460. ACM, 2008.

\bibitem[Fed64]{fedorenko1964speed}
Radii~Petrovich Fedorenko.
\newblock The speed of convergence of one iterative process.
\newblock {\em USSR Computational Mathematics and Mathematical Physics},
  4(3):227--235, 1964.

\bibitem[Hac82]{hackbusch1982multi}
Wolfgang Hackbusch.
\newblock {\em Multi-grid convergence theory}.
\newblock Springer, 1982.

\bibitem[Hac85]{hackbusch1985multi}
Wolfgang Hackbusch.
\newblock {\em Multi-grid methods and applications}, volume~4.
\newblock Springer-Verlag Berlin, 1985.

\bibitem[HZ96]{halperin1996optimal}
Shay Halperin and Uri Zwick.
\newblock An optimal randomised logarithmic time connectivity algorithm for the
  erew pram.
\newblock {\em Journal of Computer and System Sciences}, 53(3):395--416, 1996.

\bibitem[KFS13]{krishnan2013efficient}
Dilip Krishnan, Raanan Fattal, and Richard Szeliski.
\newblock Efficient preconditioning of laplacian matrices for computer
  graphics.
\newblock {\em ACM Transactions on Graphics (TOG)}, 32(4):142, 2013.

\bibitem[KMP12]{KelnerMillerPeng}
Jonathan~A. Kelner, Gary~L. Miller, and Richard Peng.
\newblock Faster approximate multicommodity flow using quadratically coupled
  flows.
\newblock In {\em Proceedings of the 44th symposium on Theory of Computing},
  STOC '12, pages 1--18, New York, NY, USA, 2012. ACM.

\bibitem[KOSZ13]{KOSZ}
Jonathan~A Kelner, Lorenzo Orecchia, Aaron Sidford, and Zeyuan~Allen Zhu.
\newblock A simple, combinatorial algorithm for solving sdd systems in
  nearly-linear time.
\newblock In {\em Proceedings of the 45th annual ACM symposium on Symposium on
  theory of computing}, pages 911--920. ACM, 2013.

\bibitem[Kou14]{Koutis14}
Ioannis Koutis.
\newblock Simple parallel and distributed algorithms for spectral graph
  sparsification.
\newblock In {\em Proceedings of the 26th ACM Symposium on Parallelism in
  Algorithms and Architectures}, SPAA '14, pages 61--66, New York, NY, USA,
  2014. ACM.

\bibitem[LM10]{LeightonM10}
Frank~Thomson Leighton and Ankur Moitra.
\newblock Extensions and limits to vertex sparsification.
\newblock In {\em Proceedings of the 42nd {ACM} Symposium on Theory of
  Computing, {STOC} 2010, Cambridge, Massachusetts, USA, 5-8 June 2010}, pages
  47--56, 2010.

\bibitem[LPS88]{LPS}
A.~Lubotzky, R.~Phillips, and P.~Sarnak.
\newblock {Ramanujan} graphs.
\newblock {\em Combinatorica}, 8(3):261--277, 1988.

\bibitem[LRS13]{LeeRS13}
Yin~Tat Lee, Satish Rao, and Nikhil Srivastava.
\newblock A new approach to computing maximum flows using electrical flows.
\newblock In {\em Proceedings of the 45th annual ACM symposium on Symposium on
  theory of computing}, STOC '13, pages 755--764, New York, NY, USA, 2013. ACM.

\bibitem[LS13]{lsMaxflow}
Yin~Tat Lee and Aaron Sidford.
\newblock Path finding ii: An$\backslash$\~{} o (m sqrt (n)) algorithm for the
  minimum cost flow problem.
\newblock {\em arXiv preprint arXiv:1312.6713}, 2013.

\bibitem[Mad13]{Madry13}
Aleksander Madry.
\newblock Navigating central path with electrical flows: From flows to
  matchings, and back.
\newblock In {\em 54th Annual {IEEE} Symposium on Foundations of Computer
  Science, {FOCS} 2013, 26-29 October, 2013, Berkeley, CA, {USA}}, pages
  253--262, 2013.

\bibitem[Mar88]{Margulis88}
G.~A. Margulis.
\newblock Explicit group theoretical constructions of combinatorial schemes and
  their application to the design of expanders and concentrators.
\newblock {\em Problems of Information Transmission}, 24(1):39--46, July 1988.

\bibitem[Moi13]{Moitra13}
Ankur Moitra.
\newblock Vertex sparsification and oblivious reductions.
\newblock {\em {SIAM} J. Comput.}, 42(6):2400--2423, 2013.

\bibitem[MP13]{MillerP13}
Gary~L. Miller and Richard Peng.
\newblock Approximate maximum flow on separable undirected graphs.
\newblock In {\em Proceedings of the Twenty-Fourth Annual ACM-SIAM Symposium on
  Discrete Algorithms}, pages 1151--1170. SIAM, 2013.

\bibitem[MSS15]{IF4}
Adam~W Marcus, Nikhil Srivastava, and Daniel~A Spielman.
\newblock Interlacing families {IV}: {Bipartite} {Ramanujan} graphs of all
  sizes.
\newblock {\em arXiv preprint arXiv:1505.08010}, 2015.
\newblock to appear in FOCS 2015.

\bibitem[MV77]{ICC}
J.~A. Meijerink and H.~A. van~der Vorst.
\newblock An iterative solution method for linear systems of which the
  coefficient matrix is a symmetric $m$-matrix.
\newblock {\em Mathematics of Computation}, 31(137):148--162, 1977.

\bibitem[Nic78]{nicolaides1978multigrid}
RA~Nicolaides.
\newblock On multigrid convergence in the indefinite case.
\newblock {\em Mathematics of Computation}, pages 1082--1086, 1978.

\bibitem[NN12]{napov2012algebraic}
Artem Napov and Yvan Notay.
\newblock An algebraic multigrid method with guaranteed convergence rate.
\newblock {\em SIAM journal on scientific computing}, 34(2):A1079--A1109, 2012.

\bibitem[OV11]{OrecchiaV11}
Lorenzo Orecchia and Nisheeth~K. Vishnoi.
\newblock Towards an sdp-based approach to spectral methods: a
  nearly-linear-time algorithm for graph partitioning and decomposition.
\newblock In {\em Proceedings of the Twenty-Second Annual ACM-SIAM Symposium on
  Discrete Algorithms}, SODA '11, pages 532--545. SIAM, 2011.

\bibitem[PS14]{PengS14}
Richard Peng and Daniel~A. Spielman.
\newblock An efficient parallel solver for {SDD} linear systems.
\newblock In {\em Symposium on Theory of Computing, {STOC} 2014, New York, NY,
  USA, May 31 - June 03, 2014}, pages 333--342, 2014.

\bibitem[SS11]{SpielmanS08:journal}
D.~Spielman and N.~Srivastava.
\newblock Graph sparsification by effective resistances.
\newblock {\em SIAM Journal on Computing}, 40(6):1913--1926, 2011.

\bibitem[ST11]{SpielmanT11}
D.~Spielman and S.~Teng.
\newblock Spectral sparsification of graphs.
\newblock {\em SIAM Journal on Computing}, 40(4):981--1025, 2011.

\bibitem[ST14]{SpielmanTengLinsolve}
Daniel~A. Spielman and Shang-Hua Teng.
\newblock Nearly-linear time algorithms for preconditioning and solving
  symmetric, diagonally dominant linear systems.
\newblock {\em SIAM. J. Matrix Anal. \& Appl.}, 35:835–885, 2014.

\bibitem[Tch36]{Tchudakoff}
Nikolai Tchudakoff.
\newblock On the difference between two neighbouring prime numbers.
\newblock {\em Rec. Math. [Mat. Sbornik] N.S.}, 1(6):799--814, 1936.

\bibitem[Vai90]{Vaidya}
Pravin~M. Vaidya.
\newblock Solving linear equations with symmetric diagonally dominant matrices
  by constructing good preconditioners.
\newblock Unpublished manuscript UIUC 1990. A talk based on the manuscript was
  presented at the IMA Workshop on Graph Theory and Sparse Matrix Computation,
  October 1991, Minneapolis., 1990.

\bibitem[ZGL03]{Zhu03}
X.~Zhu, Z.~Ghahramani, and J.~D. Lafferty.
\newblock Semi-supervised learning using gaussian fields and harmonic
  functions.
\newblock {\em ICML}, 2003.

\end{thebibliography}
\newcommand{\etalchar}[1]{$^{#1}$}

\iffull{
\begin{appendix}
\section{Weighted Expander Constructions}
\label{sec:weightedExp}

\def\edg#1{\pmb{\boldsymbol{[}} #1 \pmb{\boldsymbol{]}}}
\def\edgu#1{\pmb{\boldsymbol{(}} #1 \pmb{\boldsymbol{)}}}

In this section, we give a linear time algorithm for computing linear
  sized spectral sparsifiers of complete and bipartite
  product demand graphs.
Recall that the \textit{product demand graph} with vertex set $V$ and demands $\dd : V \rightarrow \mathbb{R}_{> 0}$
  is the complete graph
  in which the weight of edge $(u,v)$ is the product $d_{u} d_{v}$.
Similarly, the \textit{bipartite demand graph} with vertex set $U \union V$
  and demands $\dd : U \union V \rightarrow \mathbb{R}_{> 0}$ is the
  complete bipartite graph on which the weight of the edge $(u,v)$ is the product $d_{u} d_{v}$.
Our routines are based on reductions to the unweighted, uniform case.
In particular, we
\begin{itemize}
\item [1.] Split all of the high demand vertices into many vertices that all have the same demand.
  This demand will still be the highest.

\item [2.] Given a graph in which almost all of the vertices have the same highest demand,
  we \begin{itemize}
\item [a.] drop all of the edges between vertices of lower demand,
\item [b.] replace the complete graph between the vertices of highest demand with an expander, and
\item [c.] replace the bipartite graph between the high and low demand vertices with
  a union of stars.
\end{itemize}
\item [3.] To finish, we merge back together the vertices that split off from each original vertex.
\end{itemize}

We start by showing how to construct the expanders that we need for step (2b).
We state formally and analyze the rest of the algorithm for the
  complete case in the following two sections.
We explain how to handle the bipartite case in Section \ref{subsec:bipartite}.

Expanders give good approximations to unweighted complete graphs,
  and our constructions will use the spectrally best expanders---Ramanunan graphs.
These are defined in terms of the eigenvalues of their adjacency matrices.
We recall that the adjacency matrix of every $d$-regular graph has eigenvalue $d$
  with multiplicity $1$ corresponding to the constant eigenvector.
If the graph is bipartite, then it also has an eigenvalue of $-d$ corresponding
  to an eigenvector that takes value $1$ on one side of the bipartition and $-1$
  on the other side.
These are called the \textit{trivial} eigenvalues. 
A $d$-regular graph is called a Ramanujan graph if all of its non-trivial eigenvalues
  have absolute value at most $2 \sqrt{d-1}$.
Ramanujan graphs were constructed independently by Margulis~\cite{Margulis88}
  and Lubotzky, Phillips, and Sarnak~\cite{LPS}.
The following theorem and proposition summarizes part of their results.

\begin{theorem}\label{thm:LPS}
Let $p$ and $q$ be unequal primes congruent to $1$ modulo 4.
If $p$ is a quadratic residue modulo $q$, then there is a non-bipartite
  Ramanujan graph of degree $p+1$ with $q^{2} (q-1)/2$ vertices.
If $p$ is not a quadratic residue modulo $q$, then there is a bipartite
  Ramanujan graph of degree $p+1$ with $q^{2} (q-1)$ vertices.
\end{theorem}

The construction is  explicit.

\begin{proposition}\label{pro:LPS}
If $p < q$, then
  the graph guaranteed to exist by Theorem~\ref{thm:LPS} can be constructed in
  parallel depth $O (\log n)$ and work $O (n)$, where $n$ is its number of vertices.
\end{proposition}
\begin{proof}[Sketch of proof.]
When $p$ is a quadratic residue modulo $q$, the graph is a Cayley graph of
  $PSL (2,Z/qZ)$.
In the other case, it is a Cayley graph of $PGL (2,Z/qZ)$.
In both cases, the generators are determined by the $p+1$ solutions
  to the equation $p = a_{0}^{2} + a_{1}^{2} + a_{2}^{2} + a_{3}^{2}$
  where $a_{0} > 0$ is odd and $a_{1}, a_{2}$, and $a_{3}$ are even.
Clearly, all of the numbers $a_{0}$, $a_{1}$, $a_{2}$ and $a_{3}$
  must be at most $\sqrt{p}$.
So, we can compute a list of all sums $a_{0}^{2} + a_{1}^{2}$
  and all of the sums $a_{2}^{2} + a_{3}^{2}$
  with work $O (p)$, and thus a list of all $p+1$
  solutions with work $O (p^{2}) < O (n)$.

As the construction requires arithmetic modulo $q$, it is convenient
  to compute the entire multiplication table modulo $q$.
This takes time $O (q^{2}) < O (n)$.
The construction also requires the computation of a square root of $-1$
  modulo $q$, which may be computed from the multiplication table.
Given this data, the list of edges attached to each vertex of the graph
  may be produced using linear work and logarathmic depth.
\end{proof}

For our purposes, there are three obstacles to using these graphs:
\begin{itemize}
\item [1.] They do not come in every degree.
\item [2.] They do not come in every number of vertices.
\item [3.] Some are bipartite and some are not.
\end{itemize}
We handle the first two issues by observing that the primes
  congruent to 1 modulo 4 are sufficiently dense.
To address the third issue, we give a procedure to convert a non-bipartite expander into a bipartite expander, and \textit{vice versa}.

An upper bound on the gaps between consecutive primes congruent to 1 modulo 4 can
  be obtained from the following theorem of Tchudakoff.

\begin{theorem}[\cite{Tchudakoff}]
For two integers $a$ and $b$, let
 $p_{i}$ be the $i$th prime congruent to $a$ modulo $b$.
For every $\epsilon > 0$,
\[
p_{i+1} - p_{i} \leq O (p_{i}^{3/4 + \epsilon }).
\]
\end{theorem}

\begin{corollary}\label{cor:tchudakoff}
There exists an $n_{0}$ so that for all $n \geq  n_{0}$
  there is a prime congruent to 1 modulo 4 between $n$ and $2 n$.
\end{corollary}

We now explain how we convert between bipartite and non-bipartite expander graphs.
To convert a non-bipartite expander into a bipartite expander, we take its double-cover.
We recall that if $G = (V,E)$ is a graph with adjacency matrix $\AA$, then its double-cover
  is the graph with adjacency matrix
\[
  \begin{pmatrix}
0 & \AA \\
\AA^{T} & 0
\end{pmatrix}.
\]
It is immediate from this construction that the eigenvalues of the adjacency matrix
  of the double-cover
  are the union of the eigenvalues of $\AA$ with the eigenvalues of $-\AA$.
\begin{proposition}\label{pro:doubleCover}
Let $G$ be a connected, $d$-regular graph in which all matrix eigenvalues
  other than $d$ are bounded in absolute value by $\lambda$.
Then, all non-trivial adjacency matrix eigenvalues of the double-cover of $G$
  are also bounded in absolute value by $\lambda$.
\end{proposition}

To convert a bipartite expander into a non-bipartite expander, we will simply
  collapse the two vertex sets onto one another.
If $G = (U \union V, E)$ is a bipartite graph, 
  we specify how the vertices of $V$ are mapped onto $U$ by a permutation $\pi : V \rightarrow U$.
We then define the \textit{collapse} of $G$ induced by $\pi$ 
  to be the graph with vertex set $U$ 
  and edge set
\[
  \setof{  (u, \pi (v)) : (u,v) \in E }.
\]
Note that the collapse will have self-loops at vertices $u$ for which $(u,v) \in E$
  and $u = \pi (v)$.
We assign a weight of $2$ to every self loop.
When a double-edge would be created, that is when $(\pi (v), \pi^{-1} (u))$ is also an edge in the graph,
  we give the edge a weight of $2$.
Thus, the collapse can be a weighted graph.

\begin{proposition}\label{pro:collapse}
Let $G$ be a $d$-regular bipartite graph with all non-trivial adjacency matrix eigenvalues
  bounded by $\lambda$, and let $H$ be a collapse of $G$.
Then, every vertex in $H$ has weighted degree $2d$ 
   and all adjacency matrix eigenvalues of $H$ other than $d$ are bounded in absolute value by $2 \lambda$.
\end{proposition}
\begin{proof}
To prove the bound on the eigenvalues, let $G$ have adjacency matrix
\[
\begin{pmatrix}
0 & \AA \\
\AA^{T} & 0
\end{pmatrix}.
\]
After possibly rearranging rows and columns, we may assume that
  the adjacency matrix of the collapse is given by
\[
  \AA + \AA^{T}.
\]
Note that the self-loops, if they exist, correspond to diagonal entries of value $2$.
Now, let $\xx$ be a unit vector orthogonal to the all-1s vector.
We have
\[
  \xx^{T} (\AA + \AA^{T}) \xx
=
\begin{pmatrix}
\xx
\\
\xx 
\end{pmatrix}^{T}
\begin{pmatrix}
0 & \AA \\
\AA^{T} & 0
\end{pmatrix}
\begin{pmatrix}
\xx
\\
\xx 
\end{pmatrix}
\leq 
\lambda \norm{
\begin{pmatrix}
\xx
\\
\xx 
\end{pmatrix}
}^{2}
\leq
2 \lambda ,
\]
as the vector $[\xx ;\xx]$ is orthogonal to the eigenvectors of the trivial
  eigenvalues of the adjacency matrix of $G$.
\end{proof}

We now state how bounds on the eigenvalues of the adjacency matrices of graphs
  lead to approximations of complete graphs and complete bipartite graphs.

\begin{proposition}\label{pro:expanderApprox}
Let $G$ be a graph with $n$ vertices, possibly with self-loops and weighted edges,
  such that every vertex of $G$  has weighted degree $d$ and
  such that all non-trivial eigenvalues of the adjacency matrix of $G$
  have absolute value at most $\lambda \leq d/2$.
If $G$ is not bipartite, then
  $(n/d) \LL_{G}$ is an $\epsilon$-approximation of $K_{n}$ for $\epsilon = (2 \ln 2) (\lambda)/d$.
If $G$ is bipartite, then
  $(n/d) \LL_{G}$ is an $\epsilon$-approximation of $K_{n,n }$ for $\epsilon = (2 \ln 2) (\lambda)/d$.
\end{proposition}
\begin{proof}
Let $\AA$ be the adjacency matrix of $G$.
Then,
\[
  \LL_{G} = d \II  - \AA .
\]

In the non-bipartite case, we observe that all of the non-zero eigenvalues
  of $\LL_{K_{n}}$ are $n$,
  so for all vectors $x$ orthogonal to the constant vector,
\[
  x^{T}  \LL_{K_{n}} x = n x^{T}x.
\]
As all of the non-zero eigenvalues of $\LL_{G}$
  are between $d - \lambda$ and $d + \lambda$,
for all vectors $x$ orthogonal to the constant vector
\[
  n \left(1-\frac{\lambda }{d} \right)  x^{T} x 
\leq   x^{T}  (n/d)\LL_{G} x
\leq 
  n \left(1+\frac{\lambda }{d} \right)  x^{T} x.
\]
Thus,
\begin{equation*}
 \left(1-\frac{\lambda }{d} \right) \LL_{K_{n}} 
\pleq \LL_{G} \pleq 
 \left(1+\frac{\lambda }{d} \right) \LL_{K_{n}} .
\end{equation*}

In the bipartite case, we naturally assume that the bipartition is the same in both $G$ and $K_{n,n}$.
Now, let $\xx$ be any vector on the vertex set of $G$.
Both the graphs $K_{n,n}$ and $(n/d) G$ have Laplacian matrix eigenvalue
  $0$ with the constant eigenvector, and eigenvalue $2 n$ with eigenvector
  $[\bvec{1};-\bvec{1}]$.
The other eigenvalues of the Laplacian of $K_{n,n}$ are $n$, while the
  other eigenvalues of the Laplacian of $(n/d) G$ are between
\[
  n \left(1 - \frac{\lambda}{d} \right)
\quad \text{and} \quad 
  n \left(1 + \frac{\lambda}{d} \right).
\]
Thus,
\[
 \left(1-\frac{\lambda }{d} \right) \LL_{K_{n,n}} 
\pleq \LL_{G} \pleq 
 \left(1+\frac{\lambda }{d} \right) \LL_{K_{n,n}} .
\]

The proposition now follows from our choice of $\epsilon$, which guarantees that
\[
  e^{-\epsilon} \leq 1 - \lambda /d
\quad \text{and} \quad 1 + \lambda /d
\leq e^{\epsilon},
\]
provided that $\lambda /d \leq  1/2$.
\end{proof}

\begin{lemma}
\label{lem:explicitExpanders}
There are algorithms that on input $n$ and $\epsilon > n^{-1/6}$
   produce a graph having $O (n/\epsilon^{2})$ edges that is an
  $O (\epsilon)$ approximation of $K_{n'}$ or $K_{n',n'}$
  for some $n \leq n' \leq 8n$.
These algorithms run in $O (\log n)$ depth and $O (n / \epsilon^{2})$ work.
\end{lemma}
\begin{proof}
We first consider the problem of constructing an approximation of $K_{n', n'}$.
By Corollary~\ref{cor:tchudakoff} there
  is a constant $n_{0}$ so that if $n > n_{0}$, then
  there is a prime $q$ that is equivalent to
  $1$ modulo $4$ so that $q^{2} (q-1)$ is between   and $n$ and $8 n$.
Let $q$ be such a prime and let $n' = q^{2} (q-1)$.
Similarly, for $\epsilon$ sufficiently small, there is a prime $p$
  equivalent to $1$ modulo $4$ that is between
  $\epsilon^{-2}/2$ and $\epsilon^{-2}$.
Our algorithm should construct the corresponding Ramanujan graph, as described
  in Theorem~\ref{thm:LPS} and Proposition~\ref{pro:LPS}.
If the graph is bipartite, then Proposition~\ref{pro:expanderApprox} tells us
  that it provides the desired approximation of $K_{n',n'}$.
If the graph is not biparite, then we form its double cover to obtain
   a bipartite graph and use Proposition~\ref{pro:doubleCover} 
  and Proposition~\ref{pro:expanderApprox} to see that it provides the desired
  approximation of $K_{n',n'}$.

The non-bipartite case is similar, except that we require a prime $q$
  so that $q^{2} (q-1)/2$ is between $n$ and $8 n$, and we use
  a collapse to convert a bipartite expander to a non-bipartite one,
  as analyzed in Proposition~\ref{pro:collapse}.
\end{proof}

In Section \ref{sec:depth}, we just need to know that there exist
  graphs of low degree that are good approximations of 
  complete graphs.
We may obtain them from the recent theorem of Marcus, Spielman and Srivastava
  that there exist bipartite Ramanujan graphs of every degree and number of vertices
  \cite{IF4}.

\begin{lemma}
\label{lem:existExpanders}
For every integer $n$ and even integer $d$, 
  there is a weighted graph on $n$ vertices of degree at most
  $d$  that is a $4 / \sqrt{d} $ approximation
  of $K_{n}$.
\end{lemma}
\begin{proof}
The main theorem of \cite{IF4} tells us that there is a bipartite Ramanujan
  graph on $2n$ vertices of degree $k$ for every $k \leq n$.
By Propositions \ref{pro:collapse} and \ref{pro:expanderApprox},
  a collapse of this graph
  is a weighted graph of degree at most $2k$
  that is a $(4 \ln 2)/\sqrt{k}$ approximation of $K_{n,n}$.
The result now follows by setting $d = 2k$.
\end{proof}

\subsection{Sparsifying Complete Product Demand Graphs}
\label{subsec:complete}

Our algorithm for sparsifying complete product demand graphs begins by
  splitting the vertices of highest demands into many vertices.
By \textit{splitting} a vertex, we mean replacing it by many
  vertices whose demands sum to its original demand.
In this way, we obtain a larger product demand graph.
We observe that we can obtain a sparsifier of the original graph by
  sparsifying the larger graph, and then collapsing back together
  the vertices that were split.

\begin{proposition}\label{pro:splitProduct}
Let $G$ be a product demand graph with vertex set 
  $\setof{1, \dots ,n}$ 
  and demands $\dd$,
  and let $\Ghat = (\Vhat, \Ehat )$ be a product demand graph with
  demands $\ddhat$.
If there is a partition of $\Vhat$ into sets $S_{1}, \dots , S_{n}$  
  so that for all $i \in V$, $\sum_{j \in S_{i}} \hat{d}_{j} = d_{i}$,
  then $\Ghat$ is a \textit{splitting} of $G$ and there is a matrix
  $\MM$ so that
\[
  \LL_{G} = \MM \LL_{\Ghat} \MM^{T}.
\]
\end{proposition}
\begin{proof}
The $(i,j)$ entry of matrix $\MM$ is $1$ if and only if $j \in S_{i}$.
Otherwise, it is zero.
\end{proof}

We now show that we can sparsify $G$ by sparsifying $\Ghat$.
 
\begin{proposition}
\label{pro:collapseLoewner}
Let $\Ghat_{1}$ and $\Ghat_{2}$ be graphs on the same vertex set $\Vhat$ such
  that $\Ghat _{1}\approx_{\epsilon}\Ghat _{2}$ for some $\epsilon$.
Let $S_{1}, \dots , S_{n}$ be a partition of $\Vhat$, and let $G_{1}$
  and $G_{2}$ be the graphs obtained by collapsing together all the
  vertices in each set $S_{i}$ and eliminating any self loops that are
  created.
Then
\[
G_{1}\approx_{\epsilon}G_{2}.
\]
\end{proposition}
\begin{proof}
Let $\MM$ be the matrix introduced in Proposition \ref{pro:splitProduct}.
Then,
\[
  \LL_{G_{1}} = \MM \LL_{\Ghat_{1}} \MM^{T} \quad \text{and} \quad 
  \LL_{G_{2}} = \MM \LL_{\Ghat_{2}} \MM^{T}.
\]
The proof now follows from   Fact \ref{fact:orderCAC}.
\end{proof}

For distinct vertices $i$ and $j$, we let $\edgu{i,j}$ denote the graph with an edge of weight $1$ between vertex $i$ and vertex $j$.
If $i = j$, we let $\edgu{i,j}$ be the empty graph.
With this notation, we can express the product demand graph 
  as
\[
  \sum_{i < j} d_{i} d_{j} \edgu{i,j}
=
  \frac{1}{2} \sum_{i,j \in V} d_{i} d_{j}\edgu{i,j}.
\]

This notation also allows us to precisely express our algorithm for sparsifying
  product demand graphs.

\begin{algbox}
$G'=\textsc{WeightedExpander}(\dd,\epsilon)$ 
\begin{enumerate}

\item Let $\nhat$ be the least integer greater than
  $2 n / \epsilon^{2}$ such that the algorithm described in Lemma \ref{lem:explicitExpanders}
  produces an $\epsilon$-approximation of $K_{\nhat}$.

\item Let $t = \frac{\sum_{k}d_{k}}{\nhat}$.

\item Create a new product demand graph $\Ghat$ with demand vector $\hat{\dd}$
 by 
splitting each vertex $i$ into a set of $\ceil{d_{i}/t}$ vertices, $S_i$:
\begin{enumerate}
\item $\floor{d_{i}/t}$ vertices with demand $t$.
\item one vertex with demand $d_{i} - t \floor{d_{i}/t}$.
\end{enumerate} 

\item Let $H$ be a set of $\nhat$ vertices in $\Ghat$ with demand $t$,
  and let $L$ contain the other vertices.  Set $k = \sizeof{L}$.



\item 
Partition $H$ arbitrarily into sets $V_{1}, \dots , V_{k}$, so that
  $\sizeof{V_{i}} \geq \floor{\nhat / k}$ for all $1 \leq i \leq k$.
  
\item 
Use the algorithm described in Lemma \ref{lem:explicitExpanders} to
  produce $\tilde{K}_{HH}$, an $\epsilon$-approximation of the complete graph on $H$.
Set 
\[
\Gtil = t^2 \tilde{K}_{HH} + \sum_{l \in L} 
  \frac{\sizeof{H}}{\sizeof{V_{l}}} \sum_{h \in V_{l}} 
     \dhat_{l} \dhat_{h} \edgu{l,h}.
\]

\item Let $G'$ be the graph obtained by collapsing together all vertices
  in each set $S_{i}$.
\end{enumerate}
\end{algbox}



This section and the next are devoted to the analysis of this algorithm.
Given Proposition~\ref{pro:collapseLoewner}, we just need to show that 
  $\Gtil$ is a good approximation to $\Ghat$.

\begin{proposition}\label{pro:numVertsAfterSplit}
The number of vertices in $\Ghat$ is at most $n + \nhat$.
\end{proposition}
\begin{proof}
The number of vertices in $\Ghat$ is
\[
  \sum_{i \in V} \ceil{d_{i} / t}
\leq
  n +   \sum_{i \in V} d_{i} / t
=
  n + \nhat .
\]
\end{proof}

So, $k \leq n$ and $\nhat \geq 2 k / \epsilon^{2}$.
That is, $\sizeof{H} \geq 2 \sizeof{L} / \epsilon^{2}$.
In the next section, we prove the lemmas that show that for these special product demand graphs  $\Ghat $ in which
  almost all weights are the maximum,
  our algorithm produces a graph $\Gtil$ that is a good approximation of $\Ghat$.

\begin{theorem}
\label{thm:expanderFull}
Let $0 < \epsilon < 1$ and 
  let $G$ be a product demand graph with $n$ vertices and demand vector $\dd$.
Given $\dd$ and $\epsilon$ as input, \textsc{WeightedExpander} produces
  a graph $G'$ with $O (n / \epsilon^{4})$ edges that is 
  an $O (\epsilon)$ approximation of $G$.
Moreover, \textsc{WeightedExpander} runs in $O (\log n)$ depth
  and $O (n / \epsilon^{4})$ work.
\end{theorem} 
\begin{proof}
The number of vertices in the graph $\Ghat$ will be between
  $n + 2 n / \epsilon^{2}$ and $n + 16 n / \epsilon^{2}$.
So, the algorithm described in Lemma \ref{lem:explicitExpanders}  will take
  $O (\log n)$ depth and $O (n / \epsilon^{4})$ work to produce an
  $\epsilon$ approximation of the complete graph on $\nhat$ vertices.
This dominates the computational cost of the algorithm.

Proposition
  \ref{pro:collapseLoewner} tells us that
  $G'$ approximates $G$ at least as well as $\Gtil$ approximates
  $\Ghat$.
To bound how well $\Gtil$ approximates $\Ghat$,
  we use two lemmas that are stated in the next section.
Lemma \ref{lemma:light_vertex_not_important} shows
  that
\[
  \Ghat_{HH} + \Ghat_{LH} \approx_{O(\epsilon^{2})} \Ghat .
\]
Lemma \ref{lem:replaceLH} shows that 
\[
 \Ghat_{HH} + \Ghat_{LH}
\approx_{4 \epsilon}
\Ghat_{HH} + \sum_{l \in L} 
  \frac{\sizeof{H}}{\sizeof{V_{l}}} \sum_{h \in V_{l}} 
     \dhat_{l} \dhat_{h} \edgu{l,h}.
\]
And, we already know that $t^2 \tilde{K}$ is an $\epsilon$-approximation of
  $\Ghat_{HH}$.
Fact \ref{frac:orderComposition} says that we can combine these three approximations to conclude that
  $\Gtil$ is an $O (\epsilon)$-approximation of $\Ghat$.

\end{proof}

\subsection{Product demand graphs with most weights maximal}
In this section, we consider product demand graphs in which almost all weights are the maximum.
For simplicity, we make a slight change of notation from the previous section.
We drop the hats, we let $n$ be the number of vertices in the product demand graph,
  and we order the demands so that
\[
  d_{1} \leq d_{2} \leq \dots \leq d_{k} \leq d_{k+1} = \dots = d_{n} = 1.
\]

We let $L = \setof{1, \dots , k}$ and $H = \setof{k+1, \dots , n}$
  be the set of low and high demand vertices, respectively.
Let $G$ be the product demand graph corresponding to $\dd$, and let
  $G_{LL}$, $G_{HH}$ and $G_{LH}$ be the subgraphs containing the
  low-low, high-high and low-high edges repsectively.
We now show that little is lost by dropping the edges in $G_{LL}$
  when $k$ is small.

Our analysis will make frequent use of the
following Poincare inequality: 
\begin{lemma} \label{lemma:poincare}Let
$c \edgu{u,v}$ be an edge of weight $c$ and let $P$ be a path from
from $u$ to $v$ 
  consisting of edges of weights $c_{1},c_{2},\cdots,c_{k}$.
Then 
\[
c \edgu{u,v} \preceq c\left(\sum c_{i}^{-1}\right)P.
\]
\end{lemma}

As the weights of the edges we consider in this section are determined
  by the demands of their vertices,
  we introduce the notation
\[
  \edg{i,j} = d_{i} d_{j} \edgu{i,j}.
\]
With this notation, we can express the product demand graph 
  as
\[
  \sum_{i < j} \edg{i,j}
=
  \frac{1}{2} \sum_{i,j \in V} \edg{i,j}.
\]

\begin{lemma} \label{lemma:light_vertex_not_important}
If $\left|L\right|\leq\left|H\right|$,
then 
\[
G_{HH}+G_{LH}\approx_{3\frac{\left|L\right|}{\left|H\right|}} G.
\]
\end{lemma} 
\begin{proof}
The lower bound $G_{HH}+G_{LH}\preceq G_{HH}+G_{LH}+G_{LL}$
follows from $G_{LL}\succeq0$.

Using lemma \ref{lemma:poincare} and the assumptions $d_{l} \leq 1$
  for $l \in L$ and and $d_{h} = 1$ for $h\in H$, we derive for every $l_{1}, l_{2} \in L$,
\begin{align*}
\edg{l_{1}, l_{2}}
& =  \frac{1}{\left|H\right|^{2}}\sum_{h_{1},h_{2}\in H} \edg{l_{1}, l_{2}}\\
\intertext{\text{(by Lemma \ref{lemma:poincare})}}
& \preceq  \frac{1}{\left|H\right|^{2}}
  \sum_{h_{1},h_{2}\in H}d_{l_{1}}d_{l_{2}}
  \left(\frac{1}{d_{l_{1}}d_{h_{1}}}+\frac{1}{d_{h_{1}}d_{h_{2}}}+\frac{1}{d_{h_{2}}d_{l_{2}}}\right)
 \left(\edg{l_{1},h_{1}}+\edg{h_{1},h_{2}}+\edg{h_{2},l_{2}}\right)\\
 & \preceq 
 \frac{3}{\left|H\right|^{2}} \sum_{h_{1},h_{2}\in H}  \left(\edg{l_{1},h_{1}}+\edg{h_{1},h_{2}}+\edg{h_{2},l_{2}}\right)\\
 & = 
 \frac{3}{\left|H\right|}\sum_{h\in H} \left(\edg{l_{1},h}+\edg{l_{2},h}\right)+\frac{6}{\left|H\right|^{2}}G_{HH}.
\end{align*}
So,
\begin{align*}
G_{LL} & =
  \frac{1}{2} \sum_{l_{1}, l_{2} \in L} \edg{l_{1}, l_{2}}
\\
 & \preceq  \frac{1}{2}
  \sum_{l_{1},l_{2}}\left(\frac{3}{\left|H\right|}\sum_{h\in H}
  \left(\edg{l_{1},h}+\edg{l_{2},h}\right)+\frac{6}{\left|H\right|^{2}}G_{HH}\right)\\
 & =  \frac{3\left|L\right|}{\left|H\right|}G_{LH}+\frac{3\left|L\right|^{2}}{\left|H\right|^{2}}G_{HH}.
\end{align*}
The assumption $\sizeof{L} \leq \sizeof{H}$ then allows us to conclude
\[
G_{HH}+G_{LH}+G_{LL}\preceq\left(1+3\frac{\left|L\right|}{\left|H\right|}\right)\left(G_{HH}+G_{LH}\right).
\]
\end{proof} 

Using a similar technique, we will show that the edges between $L$ and $H$
  can be replaced by the union of a small number of stars.
In particular, we will partition the vertices of $H$ into $k$ sets,
  and for each of these sets we will create one star connecting
  the vertices in that set to a corresponding vertex in $L$.

We employ the following consequence of the Poincare inequality in Lemma \ref{lemma:poincare}.

\begin{lemma} \label{lemma:light_can_be_merge}
For any $\epsilon \leq 1$, $l \in L$ and $h_{1}, h_{2} \in H$,
\[
\epsilon \edg{h_{1},l}+ (1/2)\edg{h_{1},h_{2}}
\approx_{4 \sqrt{\epsilon}}
\epsilon \edg{h_{2},l}+ (1/2)\edg{h_{1},h_{2}}.
\]
\end{lemma} \begin{proof} 
By applying Lemma \ref{lemma:poincare} and
  recalling that $d_{h_{1}} = d_{h_{2}} = 1$ and $d_{l} \leq 1$, we compute
\begin{align*}
\edg{h_{1}, l}
 & \preceq   d_{h_{1}} d_{l} 
  \left(\frac{\sqrt{\epsilon}}{ d_{h_{1}} d_{h_{2}}}+\frac{1}{ d_{h_{2}} d_{l} }\right)
  \left(\frac{1}{\sqrt{\epsilon}}\edg{h_{1},h_{2}}+\edg{h_{2},l}\right)\\
 & \preceq  \frac{1+\sqrt{\epsilon}}{\sqrt{\epsilon}}\edg{h_{1},h_{2}}
    + (1+\sqrt{\epsilon})\edg{h_{2},l}
\\
 & \preceq  (1+\sqrt{\epsilon})\edg{h_{2},l}+\frac{2}{\sqrt{\epsilon}}\edg{h_{1},h_{2}}.
\end{align*}
Multiplying both sides by $\epsilon$ and adding $(1/2) \edg{h_{1}, h_{2}}$ then gives
\begin{align*}
\epsilon \edg{h_{1}, l} + (1/2) \edg{h_{1}, h_{2}}
& \pleq 
(1+\sqrt{\epsilon}) \epsilon \edg{h_{2}, l}
+
(2 \sqrt{\epsilon} + 1/2) \edg{h_{1},h_{2}}
\\
& \pleq
(1 + 4 \sqrt{\epsilon}) \left( \epsilon \edg{h_{2},l}+ (1/2)\edg{h_{1},h_{2}}  \right)
\\
& \pleq
e^{4 \sqrt{\epsilon }} \left( \epsilon \edg{h_{2},l}+ (1/2)\edg{h_{1},h_{2}}  \right).
\end{align*}
By symmetry, we also have
\[
\epsilon \edg{h_{2}, l} + (1/2) \edg{h_{1}, h_{2}}
\pleq 
e^{4 \sqrt{\epsilon }} \left( \epsilon \edg{h_{1},l}+ (1/2)\edg{h_{1},h_{2}}  \right).
\]
\end{proof}

\begin{lemma}\label{lem:replaceLH}
Recall that $L = \setof{1,\dots ,k}$ and let
  $V_{1}, \dots , V_{k}$ be a partition of $H = \setof{k+1, \dots , n}$
  so that $\sizeof{V_{l}} \geq s$ for all $l$.
Then,
\[
G_{HH} + G_{LH} 
\approx_{4 / \sqrt{s}}
G_{HH} + \sum_{l \in L} \frac{\sizeof{H}}{\sizeof{V_{l}}} \sum_{h \in V_{l}} \edg{l,h}.
\]
\end{lemma}
\begin{proof}
Observe that 
\[
G_{LH} = \sum_{l \in L} \sum_{h \in H} \edg{l,h}.
\]
For each $l \in L$, $h_{1} \in H$ and $h_{2} \in V_{l}$ we apply
  Lemma \ref{lemma:light_can_be_merge} to show that
\[
  \frac{1}{\sizeof{V_{l}}} \edg{l, h_{1}}
+ \frac{1}{2} \edg{h_{1}, h_{2}}
\approx_{4 / \sqrt{s}}
  \frac{1}{\sizeof{V_{l}}} \edg{l, h_{2}}
  + \frac{1}{2} \edg{h_{1}, h_{2}}.
\]
Summing this approximation over all $h_{2} \in V_{l}$
  gives
\[
 \edg{l, h_{1}}
+   \sum_{h_{2} \in V_{l}}
  \frac{1}{2} \edg{h_{1}, h_{2}}
\approx_{4 / \sqrt{s}}
  \sum_{h_{2} \in V_{l}}
\left(  \frac{1}{\sizeof{V_{l}}}
   \edg{l, h_{2}}
  + \frac{1}{2} \edg{h_{1}, h_{2}} \right)
.
\]
Summing the left-hand side of this this approximation
  over all $l \in L$ and $h_{1} \in H$ gives
\[
  \sum_{l \in L, h_{1} \in H}  \edg{l, h_{1}}
+ 
   \sum_{h_{2} \in V_{l}}
  \frac{1}{2} \edg{h_{1}, h_{2}}
=
  \sum_{l \in L, h_{1} \in H}  \edg{l, h_{1}}
+
  \frac{1}{2}
  \sum_{h_{1} \in H, l \in L}  
   \sum_{h_{2} \in V_{l}} \edg{h_{1}, h_{2}}
=
  G_{LH} + G_{HH}.
\]
On the other hand, the sum of the right-hand terms gives
\[
G_{HH} + 
  \sum_{l \in L, h_{1} \in H}
  \sum_{h_{2} \in V_{l}}
 \frac{1}{\sizeof{V_{l}}}
   \edg{l, h_{2}}
=
G_{HH} + 
  \sum_{l \in L}
  \sum_{h_{2} \in V_{l}}
 \frac{\sizeof{H}}{\sizeof{V_{l}}}
   \edg{l, h_{2}}.
\]
\end{proof}

\subsection{Weighted Bipartite Expanders}
\label{subsec:bipartite}

This construction extends analogously to bipartite product graphs.
The bipartite product demand graph of vectors $(\dd^{A},\dd^{B})$
is a complete bipartite graph whose weight between vertices $i\in A$
and $j\in B$ is given by $w_{ij}=d_{i}^{A}d_{j}^{B}$. Without
loss of generality, we will assume $d_{1}^{A}\geq d_{2}^{A}\geq\cdots\geq d_{n^{A}}^{A}$
and $ d_{1}^{B}\geq d_{2}^{B}\geq\cdots\geq d_{n^{B}}^{B}$.

As the weights of the edges we consider in this section are determined
  by the demands of their vertices,
  we introduce the notation
\[
  \edg{i,j} = d^{A}_{i} d^{B}_{j} \edgu{i,j}.
\]

Our construction is based on a similar observation that if most vertices
on $A$ side have $d_{i}^{A}$ equaling to $d_{1}^{A}$ and most
vertices on $B$ side have $d_{i}^{B}$ equaling to $d_{1}^{B}$,
then the uniform demand graph on these vertices dominates the graph.

\begin{algbox}
$G'=\textsc{WeightedBipartiteExpander}(\dd^{A},\dd^{B},\epsilon)$ 
\begin{enumerate}

\item Let $n'=\max(n^{A},n^{B})$ and $\nhat$ be the least integer greater than
  $2 n' / \epsilon^{2}$ such that the algorithm described in
  Lemma \ref{lem:explicitExpanders} produces an $\epsilon$-approximation of $K_{\nhat ,\nhat}$.

\item Let $t^{A}=\frac{\sum_{k}d_{k}^{A}}{\nhat}$ and  $t^{B}=\frac{\sum_{k}d_{k}^{B}}{\nhat}$.

\item Create a new bipartite demand graph $\Ghat$ with demands
$\ddhat^{A}$ and $\ddhat^{B}$ follows:
\begin{enumerate}
  \item On the side $A$ of the graph, for each vertex $i$, create a subset $S_i$ consisting of $\ceil{d^{A}_{i}/t^A}$ vertices:
  \begin{enumerate}
  \item   $\floor{d^{A}_{i}/t^A}$ with demand $t^A$.
    \item  one vertex with demand $d^{A}_{i} - t^A \floor{d^{A}_{i}/t^A}$.
  \end{enumerate} 

\item Let $H^{A}$ contain $\hat{n}$ vertices of $A$ of with demand $t^{A}$, and let
   $L^{A}$ contain the rest.
Set $k^{A} = \sizeof{L^{A}}$.

\item Create the side $B$ of the graph with partition $H^{B}, L^{B}$
and demand vector $\ddhat^{B}$ similarly.
  \end{enumerate}

\item Partition $H^{A}$ into sets of size
$\sizeof{V^{A}_{i}} \geq \floor{\nhat / k^{A}}$, one corresponding
to each vertex $l \in L^{A}$.
Partition $V_{B}$ similarly.

\item 
Let $\tilde{K}_{H^{A} H^{B}} $ be a bipartite expander produced
by Lemma~\ref{lem:explicitExpanders} that $\epsilon$-approximates $K_{\hat{n}  \hat{n}}$, identified with the vertices $H^{A}$ and $H^{B}$.
  
 Set 
\[
\Gtil = t^{A} t^{B} \tilde{K} + \sum_{l \in L^{A}} 
  \frac{\sizeof{H^{B}}}{\sizeof{V^{B}_{l}}} \sum_{h \in V^{B}_{l}} 
     \dhat^{A}_{l} \dhat^{B}_{h} \edgu{l,h}
+ \sum_{l \in L^{B}} 
  \frac{\sizeof{H^{A}}}{\sizeof{V^{A}_{l}}} \sum_{h \in V^{A}_{l}} 
     \dhat^{B}_{l} \dhat^{A}_{h} \edgu{l,h}.
\]

\item Let $G'$ be the graph obtained by collapsing together all vertices 
  in each set $S^{A}_{i}$ and $S^{B}_{i}$.
\end{enumerate}
\end{algbox}

Similarly to the nonbipartite case, the Poincare inequality show that
the edges between low demand vertices can be completely omitted if
there are many high demand vertices which allows the demand routes
through high demand vertices.

\begin{lemma} \label{lemma:light_vertex_not_important2}Let $G$
be the bipartite product demand graph of the demand $(\dd_{i}^{A},\dd_{j}^{B})$.
Let $H^{A}$ a subset of vertices on $A$ side with demand higher
than the set of remaining vertices $L^{A}$ on $A$ side.
Define $H^{B},L^{B}$ similarly. Assume that $\sizeof{L^{A}}\leq\sizeof{H^{A}}$
and $\sizeof{L^{B}}\leq\sizeof{H^{B}}$, then 
\[
G_{H^{A}H^{B}}+G_{H^{A}L^{B}}+G_{L^{A}H^{B}}\approx_{3\max\left(\frac{\sizeof{L^{A}}}{\sizeof{H^{A}}},\frac{\sizeof{L^{B}}}{\sizeof{H^{B}}}\right)}G.
\]
\end{lemma} \begin{proof} 
The proof is analogous to Lemma~\ref{lemma:light_vertex_not_important},
but with the upper bound modified for bipartite graphs.

For every edge $l_A, l_B$, we embed it evenly into paths of
the form $l_A, h_B, h_A, l_B$ over all choices of $h_A$ and $h_B$.
The support of this embedding can be calculated using
Lemma~\ref{lemma:poincare}, and the overall accounting
follows in the same manner as Lemma~\ref{lemma:light_vertex_not_important}.
\end{proof}

It remains to show that the edges between low demand and high demand
vertices can be compressed into a few edges.
The proof here is also analogous to Lemma~\ref{lemma:light_can_be_merge}:
we use the Poincare inequality to show that all
demands can routes through high demand vertices.
The structure of the bipartite graph makes it helpful
to further abstract these inequalities via the following
Lemma for four edges.

\begin{lemma} \label{lem:light_can_be_merge_2}Let
$G$ be the bipartite product demand graph of the demand $(\dd_{i}^{A},\dd_{j}^{B})$.
Given $h_{A},l_{A}\in A$ and $h_{B,1},h_{B,2}\in B$. Assume that
$d_{h_{A}}^{A}=d_{h_{B,1}}^{B}=d_{h_{B,2}}^{B}\geq d_{l_{A}}^{A}$.
For any $\epsilon<1$ , we have 
\[
\epsilon\edg{l_{A},h_{B,1}}+\edg{h_{A},h_{B,2}}+\edg{h_{A},h_{B,1}}\approx_{3\sqrt{\epsilon}}\epsilon \edg{l_{A},h_{B,2}}+\edg{h_{A},h_{B,2}}+\edg{h_{A},h_{B,1}}.
\]
\end{lemma} \begin{proof} 
Using Lemma $\ref{lemma:poincare}$ and
$d_{h_{A}}^{A}=d_{h_{B,1}}^{B}=d_{h_{B,2}}^{B}\geq d_{l_{A}}^{A}$,
we have 
\begin{align*}
 &  \edg{l_{A},h_{B,1}} \\
& \preceq  d_{l_{A}}^{A} d_{h_{B,1}}^{B}\left(\frac{1}{ d_{l_{A}}^{A} d_{h_{B,2}}^{B}}+\frac{\sqrt{\epsilon}}{d_{h_{A}}^{A} d_{h_{B,2}}^{B}}+\frac{\sqrt{\epsilon}}{ d_{h_{A}}^{A} d_{h_{B,1}}^{B}}\right)\left(\edg{l_{A},h_{B,2}}+\frac{1}{\sqrt{\epsilon}}\edg{h_{A},h_{B,2}}+\frac{1}{\sqrt{\epsilon}}\edg{h_{A},h_{B,1}}\right)\\
 & \preceq  (1+2\sqrt{\epsilon})\edg{l_{A},h_{B,2}}+\frac{1+2\sqrt{\epsilon}}{\sqrt{\epsilon}}\edg{h_{A},h_{B,2}}+\frac{1+2\sqrt{\epsilon}}{\sqrt{\epsilon}}\edg{h_{A},h_{B,1}}.
\end{align*}
Therefore,
\begin{align*}
 &   \epsilon\edg{l_{A},h_{B,1}}+\edg{h_{A},h_{B,2}}+\edg{h_{A},h_{B,1}} \preceq  (1+3\sqrt{\epsilon})\left(\epsilon\edg{l_{A},h_{B,2}}+\edg{h_{A},h_{B,2}}+\edg{h_{A},h_{B,1}}\right).
\end{align*}
The other side is similar due to the symmetry.\end{proof} 

\begin{theorem}
\label{lem:BiExpanderFull}
Let $0 < \epsilon < 1$ and 
  let $G$ be a bipartite demand graph with $n$ vertices and demand vector $(\dd^{A},\dd^{B})$.
 \textsc{WeightedBipartiteExpander} produces
  a graph $G'$ with $O (n / \epsilon^{4})$ edges that is 
  an $O (\epsilon)$ approximation of $G$.
Moreover, \textsc{WeightedBipartiteExpander} runs in $O (\log n)$ depth
  and $O (n / \epsilon^{4})$ work.
\end{theorem} 

\begin{proof}
The proof is analogous to Theorem~\ref{thm:expanderFull}.
After the splitting, the demands in $H^{A}$ are higher than the demands in $L^{A}$ and so is $H^{B}$ to $L^{B}$.
Therefore, Lemma \ref{lemma:light_vertex_not_important2} shows that
  that

\[
\Ghat_{H^{A}H^{B}}+\Ghat_{H^{A}L^{B}}+\Ghat_{L^{A}H^{B}} \approx_{3 \epsilon^{2}/2} \Ghat .
\]
By a proof analogous to Lemma \ref{lem:replaceLH}, one can use Lemma \ref{lem:light_can_be_merge_2} to show that 
\[
 \Ghat_{H^{A}H^{B}}+\Ghat_{H^{A}L^{B}}+\Ghat_{L^{A}H^{B}}
\approx_{O(\epsilon)}
\Ghat_{H^{A}H^{B}} + \frac{\sizeof{H^{B}}}{\sizeof{V^{B}_{l}}} \sum_{h \in V^{B}_{l}} 
     \dhat^{A}_{l} \dhat^{B}_{h} \edgu{l,h}
+ \sum_{l \in L^{B}} 
  \frac{\sizeof{H^{A}}}{\sizeof{V^{A}_{l}}} \sum_{h \in V^{A}_{l}} 
     \dhat^{B}_{l} \dhat^{A}_{h} \edgu{l,h}.
\]
And, we already know that $t^A t^B \tilde{K}$ is an $\epsilon$-approximation of
  $\Ghat_{H^{A}H^{B}}$.
Fact \ref{frac:orderComposition} says that we can combine these three approximations to conclude that
  $\Gtil$ is an $O (\epsilon)$-approximation of $\Ghat$.
%
 \end{proof} 

\end{appendix}
}{
\newpage

{\Large
This page is intentionally left (almost) blank.

The full version starts after this page.
}

}

\end{document}